\documentclass{article}
\usepackage[final,nonatbib]{neurips_2022_modified}
\usepackage[utf8]{inputenc} 
\usepackage[T1]{fontenc}    
\usepackage[sort, numbers]{natbib}
\usepackage[hypertexnames=false]{hyperref}       
\hypersetup{
colorlinks = true,
allcolors = blue 
}
\usepackage{url}            
\usepackage{booktabs}       
\usepackage{amsfonts}       
\usepackage{nicefrac}       
\usepackage{microtype}      
\usepackage[dvipsnames]{xcolor}
\usepackage{amsmath}
\usepackage{amsthm}
\usepackage{bm}
\usepackage{graphicx}
\usepackage{macros} 
\usepackage{multirow}
\usepackage{cleveref}

\usepackage{amsmath, amssymb, amsthm}
\usepackage{graphicx,psfrag}
\usepackage{mathtools}
\usepackage{bm}
\usepackage{bbm}
\usepackage{enumerate}
\usepackage{subcaption}
\usepackage{url} 
\usepackage{booktabs}
\usepackage{multirow}
\usepackage{color}
\usepackage{wrapfig}
\usepackage[linesnumbered,ruled,vlined]{algorithm2e}

\makeatletter
\newcommand{\algorithmfootnote}[2][\footnotesize]{%
  \let\old@algocf@finish\@algocf@finish
  \def\@algocf@finish{\old@algocf@finish
    \leavevmode\rlap{\begin{minipage}{\linewidth}
    #1#2
    \end{minipage}}%
  }%
}
\makeatother

\newtheorem{lemma}{Lemma}

\newtheorem{proposition}{Proposition}
\newtheorem{theorem}{Theorem}

\theoremstyle{definition}
\newtheorem{example}{Example}

\theoremstyle{remark}
\newtheorem{remark}{Remark}

\def\ind{\mathbbm{1}}   
\newcommand{\given}{\,|\,}

\SetKwInput{KwInput}{Input}                
\SetKwInput{KwOutput}{Output}              

\def\bbR{\mathbb{R}}
\def\bbP{\mathbb{P}}
\def\gap{\mathrm{Gap}}

\def\bbN{\mathbb{N}}
\def\bbE{\mathbb{E}}
\def\bP{\mathbf{P}}
\def\bK{\mathbf{K}}
\def\MH{\mathbf{K}_{\mathrm{RW}}}
\def\MTM{\mathbf{P}_{\mathrm{MTM}}}

\def\cE{\mathcal{E}}
\def\cX{\mathcal{X}}
\def\cN{\mathcal{N}}

\def\tmix{t_{\mathrm{mix}}}

\DeclarePairedDelimiter{\TV}{\lVert}{\rVert_{\mathrm{TV}}}
\DeclarePairedDelimiter{\inner}{\langle}{\rangle_{\pi}}

\def\tcb{\textcolor{blue}}
\title{Rapidly Mixing Multiple-try Metropolis Algorithms for Model Selection Problems}

\author{%
  Hyunwoong Chang\thanks{Equal contribution} \\
  Department of Statistics\\
  Texas A\&M University\\
  \texttt{hwchang@stat.tamu.edu} \\
  \And
  Changwoo J. Lee$^*$ \\
  Department of Statistics\\
  Texas A\&M University\\
  \texttt{c.lee@stat.tamu.edu} \\
  \AND
  Zhao Tang Luo \\
  Department of Statistics\\
  Texas A\&M University\\
  \texttt{ztluo@stat.tamu.edu} \\
  \And
  Huiyan Sang\\
  Department of Statistics\\
  Texas A\&M University\\
  \texttt{huiyan@stat.tamu.edu} \\
  \And
  Quan Zhou\\
  Department of Statistics\\
  Texas A\&M University\\
  \texttt{quan@stat.tamu.edu} \\
}

\begin{document}
\maketitle
\begin{abstract}
The multiple-try Metropolis (MTM) algorithm is an extension of the Metropolis-Hastings (MH) algorithm by selecting the proposed state among multiple trials according to some weight function.
Although MTM has gained great popularity owing to its faster empirical convergence and mixing than the standard MH algorithm, 
its theoretical mixing property is rarely studied in the literature due to its complex proposal scheme. 
We prove that MTM can achieve a mixing time bound smaller than that of MH by a factor of the number of trials under a general setting applicable to high-dimensional model selection problems with discrete state spaces.  
Our theoretical results motivate a new class of weight functions called \textit{locally balanced weight functions} and guide the choice of the number of trials, 
which leads to improved performance over standard MTM algorithms.
We support our theoretical results by extensive simulation studies and real data applications with several Bayesian model selection~problems.
\end{abstract}

\section{Introduction}\label{sec:intro}

 The Markov chain Monte Carlo (MCMC) method has become a powerful and standard tool for Bayesian posterior computation. In particular, the Metropolis-Hastings (MH) algorithm is widely employed in many statistical and machine learning models to sample from posterior distributions with intractable normalizing constants. A prominent usage example of MH is the vast class of Bayesian model selection problems with discrete-valued high-dimensional parameters, whose posterior probabilities are infeasible to evaluate because the size of model space potentially grows (super-)exponentially with the number of parameters. Examples of Bayesian model selection problems and their applications are ubiquitous,  including Bayesian variable selection (BVS)~\citep[][and references therein]{tadesse2021handbook}, stochastic block model (SBM)~\citep{holland1983stochastic, nowicki2001estimation}, Bayesian structure learning~\citep{madigan1995bayesian,friedman2003being, chang2022order}, change-point detection model~\citep{barry1993bayesian,chib1998estimation}, spatial clustering models~\citep{li2019spatial,luo2021bayesian,lee2021t}, and many others.

Despite the popularity of MCMC algorithms, their convergence rate analysis is often challenging. Mixing time is a key concept of interest when analyzing the convergence of MCMC, which defines the number of iterations needed to converge to a stationary distribution within a small total variation distance~\citep{levin2017markov}. We say a chain is \textit{rapidly mixing} if the mixing time grows at most polynomially in some complexity parameters. 
Mixing time bound has been studied extensively for problems such as card shuffling~\citep{aldous1986shuffling}, random walks on groups~\citep{aldous1983random}, graph colorings~\citep{jerrum1995very}, and the Ising model~\citep{dobrushin1992wulff},  across the fields of probability, physics, and computer science~\citep{randall2006rapidly}.  
Although there is a rich literature on discrete-state-space mixing time analysis in many areas, the focus has been primarily on approximate counting and random generation, but only recently has research on the MCMC convergence in the context of statistical model selection begun to emerge. 
\cite{yang2016computational} proved rapid mixing of an MH algorithm in a high-dimensional BVS problem, and~\cite{pollard2019rapid} improved the bound when a ``warm start'' is available. Rapid mixing results of MH algorithms are established for Bayesian community detection problem with SBM~\citep{zhuo2021mixing} and for learning equivalence class in a high-dimensional Bayesian structure learning~ \cite{zhou2021complexity}. 
Recently, locally informed MH algorithms have been developed to incorporate local information about the target distribution. 
\cite{zanella2020informed} proposed a locally balanced proposal and proved their optimality in terms of Peskun ordering \citep{peskun1973optimum} for discrete space models without providing mixing time rate analysis. \cite{zhou2021dimension} also proposed an informed MCMC and proved a dimension-free mixing time bound but only for a BVS problem.

The multiple-try Metropolis (MTM) algorithm~\citep{liu2000multiple} is an extension of the MH algorithm, which uses a weight function to choose the proposed state among randomly sampled trials from a proposal distribution. The variety of weight functions and the number of trials enrich the algorithm class but result in a much more complex transition probability than standard MH algorithms, hampering the theoretical development on the mixing time of MTM. As a result, there is a lack of theoretical guidance for practitioners on the choices of weight functions and number of trials in MTM and also the choice between MTM and other MCMC algorithms.  
Only very recently, \cite{yang2021convergence} considered the MTM independence sampler (MTM-IS), a special case of MTM where the proposal does not depend on the current state. They conducted the convergence rate analysis of MTM-IS, which shows it is less efficient than the simpler approach of repeated Metropolised independent sampling at the same computational cost, but they did not discuss the choices of weight functions and the number of trials. As acknowledged in \cite{yang2021convergence}, MTM is significantly different from MTM-IS, and its convergence rate is more challenging to analyze. In this paper, we establish the mixing time bound of MTM for general model selection problems. Our main contributions are the following:
\begin{enumerate}[1)] 
    \item We show that MTM with existing popular choices of weight functions can have mixing issues and propose a new class of weight functions called \textit{locally balanced weight functions}.
    \item With locally balanced weight functions, we prove that the mixing time bound of the MTM algorithm is smaller than that of the MH algorithm by a factor of the number of trials $N$ in a general model selection setting, under some regularity conditions including a rate condition on $N$. 
    \item We provide theoretical justification for the counterintuitive phenomenon that an increase in the number of trials $N$ may not always lead to a proportionate improvement in mixing, and suggest a theoretically guided algorithm to choose $N$ for practical use.
    \item We validate our theoretical findings via extensive simulation studies and real data applications with various model selection problems: Bayesian variable selection, stochastic block models, spatial clustering models, and structure learning models.
\end{enumerate}

\section{Preliminary}\label{sec:prelim}

\subsection{Notation}\label{subsec:note}

For a positive integer $m \in \bbN$, let $[m] = \{1, \dots, m\}$ and $\lfloor m \rfloor$ be the largest integer no greater than $m$. Let $|\cdot|$ denote the cardinality of a set. Let $\cX = \cX_p$ denote a finite state space 
where $|\cX|$ grows in a complexity parameter $p \in \bbN$. We use $\bP: \cX \times \cX \rightarrow [0,1]$ to denote the transition probability matrix of an irreducible, aperiodic and reversible Markov chain. 
We denote by $\mathbf{I}$ an identity matrix. 
We call a probability distribution $\pi$ on $\cX$ a \textit{stationary distribution} or a \textit{target distribution} of a chain $\bP$ if it satisfies $ \sum_{z \in \cX} \pi (z) \bP(z,x) = \pi(x)$ for all $x \in \cX$. 
In Bayesian inference, $\pi$ is a posterior distribution. Let $\cN \colon \cX \rightarrow 2^\cX$ be a set-valued map called a neighborhood relation, which maps a state to a set of states. We say the neighborhood relation $\cN$ is symmetric if $x' \in \cN(x)$ iff $x \in \cN(x')$ for all $x \neq x'$. We define the random walk proposal matrix as $\MH(x,y) = \ind_{\cN(x)}(y) /|\cN(x)|$ where $\ind_{\cN(x)}$ denotes the indicator function of $\cN(x)$.
We say  $\mathbf{K}$ is symmetric if $\mathbf{K}(x,y) = \mathbf{K}(y,x)$. 
We denote a graph by $(V, E)$ where $V$ is the vertex set and $E \subset V \times V$ is the set of edges. We say $\delta(x,y) = (v_0, \dots, v_m)$ is an $E$-path (with length $m$) from $x$ to $y$ if $e_i = (v_{i-1}, v_{i}) \in E$ for $i \in [m]$, $v_0 = x, v_m = y$, and all vertices are distinct.

\subsection{Multiple-try Metropolis algorithm}\label{subsec:MTM}
We present the MTM algorithm~\citep{liu2000multiple} on a discrete state space in Algorithm~\ref{alg:MTM}. 
Here $N$ trial states $y_1,\dots,y_N$ are first sampled from the proposal $\MH(x,\cdot)$, and then the proposed state $y$ is selected among the trials using the weight function
\begin{equation}
\label{eq:weightft}
    w(y\given x) = \pi(y)\MH(y,x)\lambda(y,x), 
\end{equation}
where $\lambda(y,x)$ can be any symmetric function that is positive whenever $\MH(y,x)>0$. Note that Algorithm~\ref{alg:MTM} becomes the standard random walk MH algorithm when the number of trials $N$ is one. 

There have been efforts to generalize the MTM algorithm~\citep{craiu2007acceleration,casarin2013interacting, martino2018review}. 
\cite{pandolfi2010generalization} extends the weight function class by proposing a general form of acceptance probability to achieve the reversibility of a chain when the weight function does not belong to \eqref{eq:weightft}. 
While one can easily find a weight function satisfying~\eqref{eq:weightft}, 
there is little research on the most appropriate form of weight functions.
\cite{liu2000multiple} investigated some particular class of weight functions and claimed the performance is insensitive to the choice of weight functions within this class;   $w(y \given x) = \pi(y)$ (when $\MH$ is symmetric) and $w(y \given x) = \pi(y)/\MH(x,y)$ are popular choices in this class~\citep{bedard2012scaling,martino2013flexibility}. 

Another important component of the MTM method is the number of trials $N$. 
It has been reported with empirical evidence that increasing $N$ does not necessarily result in a corresponding improvement in mixing under the random walk proposal~\citep{martino2017issues}. They illustrate a counterintuitive phenomenon that the acceptance probability for a proposal move to higher probability regions can be extremely low. As a solution, they suggest randomly selecting $N$ at each iteration. 
Despite the anecdotal observation that certain choices of $N$ make the algorithm less efficient, the principled way to choose $N$ is rarely discussed in the literature.

\begin{algorithm}[t]
\caption{Multiple-try Metropolis (MTM) algorithm with proposal $\MH$}
\label{alg:MTM}
\KwInput{An initial state $x_0$, a neighborhood relation $\cN$, the number of trials $N$, a weight function $w$ defined in \eqref{eq:weightft}, the number of Markov chain iterations $T$.}

\For{$t = 0, \dots, T-1$}{
    Step 1. Draw $y_1, \dots, y_N$ uniformly at random from $\cN(x_{t})$, and compute $w(y_j \given x_{t})$ for $j = 1 \dots, N$.\\
    Step 2. Select $j \in [N]$ with probability proportional to $w(y_j \given x_{t})$ and define $y = y_j$.\\
    Step 3. Sample $x^\star_1, \dots, x^\star_{N-1}$ uniformly at random from $\cN(y)$ and define
    \begin{align}\label{eq:acceptance_probability}
        \alpha = \min \left\{1, \frac{w(y \given x_{t}) + \sum_{l \in [N] \setminus \{j\}} w( y_l \given x_{t})}{w( x_{t} \given y) + \sum_{l \in [N-1]} w(  x_l^\star \given y)} \right\}.
    \end{align}
    \\
    Step 4. With probability $\alpha$, accept $y$ and let $x_{t+1} = y$; otherwise, let $x_{t+1} = x_t$.\\
}
\KwOutput{A set of Markov chain samples $\{x_t\}_{t=1}^T$.}
\end{algorithm}

\subsection{Analysis on mixing time via geometric tools}\label{subsec:mixing}

A theoretical foundation for employing MCMC methods to sample from a posterior distribution on a finite state space is the well-known convergence theorem, which states that an irreducible and aperiodic Markov chain with a stationary distribution converges to the stationary distribution regardless of the initial distribution~\citep[Theorem 4.9]{levin2017markov}.
However, since it is impossible to simulate a chain infinitely long, our main interest is to analyze how well the chain approximates the target distribution after a certain number of steps. To this end, we introduce  $\epsilon$-mixing time $\tmix(\epsilon) = \max_{x \in \cX} \min \left \{t \in \bbN \colon \TV{\bP^t_{\rm{lazy}} (x, \cdot)- \pi(\cdot)} \leq \epsilon   \right\}$,
where $\TV{\cdot}$ is the total variation (TV) distance, $\epsilon$ is a small positive constant, and $\bP_{\rm{lazy}} = (\bP + \mathbf{I})/2$ denotes the lazy version of $\bP$,  which is introduced merely for theoretical convenience since all eigenvalues of $\bP_{\rm{lazy}}$ are non-negative (the use of lazy Markov chains is standard in mixing time analysis). 
In words, $\tmix$ quantifies the necessary number of steps for the chain to have a small TV distance from the stationary distribution. Among several techniques for characterizing the mixing property~\citep{bubley1997path, saloff1997lectures, levin2017markov, guruswami2016rapidly}, a path method is employed to derive our main result~\citep{diaconis1991geometric,sinclair1992improved}.  We also refer readers to \cite[Section 13.4]{levin2017markov} and \cite[Section 3]{saloff1997lectures} for more details about the path method. 
Intuitively, the mixing time depends on how well the states ``communicate'',  especially among those which have high values of $\pi$.
In contrast to samplers defined on a continuous state space with a Euclidean topology, posterior landscape and modality on a finite state space are more difficult to envision since it highly depends on how one defines the neighborhood relation. 
Given a transition probability matrix $\bP$, we define the neighborhood of a state $x$ by $\cN(x) = \{x' \in \cX \setminus \{x\}: \bP(x,x') > 0\}$, the set of reachable states by one-step transition from $x$, and we say $x$ is a mode if $\pi(x) > \max_{x' \in \cN(x)} \pi(x')$.
This specification of the neighborhood enables us to define a graph $(\cX, E)$ where $E = \{(x,x') \in \cX \times \cX : x' \in \cN(x)\}$ is an (undirected since we have a symmetric neighborhood relation by the reversibility of $\bP$) edge set. For each $(x,y) \in \cX \times \cX$ with $x \neq y$, we choose only one $E$-path from $x$ to $y$ and denote it as $\delta(x,y)$. A \textit{path ensemble} is defined as a collection of the paths $\Delta = \{\delta(x,y) : x,y \in \cX, x \neq y\}$.
A path ensemble is a ``configuration of path network'' that helps quantify the quality of the communication within the state space $\cX$. The following proposition, which serves as a base of our main result, gives a mixing time bound by means of a path ensemble.
\begin{proposition}[\cite{sinclair1992improved}]\label{prop:path}
For any path ensemble $\Delta$, 
\begin{align*}
    \tmix(\epsilon) \leq 2 \rho(\Delta) \ell(\Delta) \left[\log (1/ \epsilon) + \log \left \{  \min_{x \in \cX} \pi(x) \right\}^{-1}\right],
\end{align*}
where $\ell(\Delta) = \max_{x,y} |\delta(x,y)|$ and
\begin{align}\label{eq:congestion}
    \rho(\Delta) = \max_{(u,v) \in E} \frac{1}{\pi(u)\bP(u,v)} \sum_{x,y : \delta(x,y) \ni (u,v) }\pi(x)\pi(y). 
\end{align}
\end{proposition}
For completeness, we provide a proof of Proposition~\ref{prop:path} with a detailed discussion in Appendix~\ref{subsec:pf_path}. The quantity $\rho(\Delta)$ is called the \textit{congestion parameter}, and $\ell(\Delta)$ is the length of the longest path in $\Delta$. To interpret $\rho(\Delta)$ and $\ell(\Delta)$, some terminologies from transportation networks in graph theory are often useful. Consider $\pi(u)\bP(u,v)$ as the capacity of an edge $(u,v)$, and $\pi(x) \pi(y)$ as the unit flow of a path $\delta(x,y)$. Then $\rho(\Delta)$ quantifies the maximum loading of an edge under the configuration $\Delta$. If one edge with a small capacity is traversed by a huge number of paths with a high unit flow, it results in a constriction. Consider $\ell(\Delta)$ as the diameter of the configuration $\Delta$. If the diameter is extremely large, it causes inefficient communication, especially among the states lying on the margin of the configuration.
If it is possible to design the path system in a way that no edge is overloaded by paths and no two states are too far apart, then the chain can have a good mixing behavior. One of the key steps in the proof of the main result is to identify a good path ensemble $\Delta^*$ (see Appendix~\ref{subsec:pf_thm1}).

\section{Main result}\label{sec:main}

In this section, we prove the rapid mixing of the MTM algorithm for general model selection problems on a finite model space $\cX$, where we assume that there is only one state with the highest posterior mass, denoted by $x^*$. 
Throughout the section, we may assume that $|\cX|$ grows with the complexity parameter $p$ (e.g., the number of variables in BVS), and the posterior distribution $\pi$ and other related objects are implicitly parameterized by $p$.

\subsection{A mixing time bound with locally balanced weight functions}\label{subsec:weight}

While extending the possible form of weight functions has been the main focus in the literature~\citep{pandolfi2010generalization, martino2013flexibility}, the choice of weight function is rarely discussed. We introduce a class of weight functions by using a function $h \colon \bbR_{> 0} \rightarrow \bbR_{> 0}$ that satisfies $h(u) = uh(1/u)$ for all $u > 0$; such a function $h$ is called a \textit{balancing function}~\citep{martino2017metropolis,zanella2020informed}. The class of balancing functions is very broad: we can choose an arbitrary non-negative function $h$ on $(0,1]$, and then $h(u)$ on $u \in (1, \infty)$ is defined by the ``balancing rule'', $h(u) = uh(1/u)$. 
Typical choices of balancing functions include $\sqrt{u}$, $\min \{1, u\}$, $\max \{1, u\}$, $u/(u+1)$, and $u+1$. 
The following proposition states that Algorithm~\ref{alg:MTM} with the proposed weight function induces a reversible Markov chain with $\pi$ as its stationary distribution.

\begin{proposition}[Locally balanced weight functions]\label{prop:weight}
Suppose that a weight function is defined as
\begin{align}\label{eq:weight}
    w(y \given x) = h\left( \frac{\pi(y)\MH(y,x)}{\pi(x)\MH(x,y)}\right),
\end{align}
where $h$ is a balancing function. Then, a Markov chain $\MTM$ induced by Algorithm~\ref{alg:MTM} has a stationary distribution $\pi$.
\end{proposition}

The proof is deferred to Appendix~\ref{subsec:pf_prop1}. We call weight functions of the form~\eqref{eq:weight} \textit{locally balanced weight functions}; though we only consider random walk proposals in this paper, $\MH$ in Proposition~\ref{prop:weight} can be replaced by any other  proposal matrix.
Compared with the popular choice of weight function in the existing MTM literature such as $w(y \given x) = \pi(y)$, the locally balanced weight function can be viewed as a transformed and scaled version with respect to the posterior probability at current state $x$. For example, if $h(u)=\sqrt{u}$ and $\MH$ is symmetric, then $w(y \given x)=\sqrt{\pi(y)/\pi(x)}$. 

Before providing a rigorous mixing time bound of MTM under our new choice of $w(y\given x)$, we give an intuitive explanation for why the existing unscaled weight function may have poor mixing properties. 
First, it is reasonable to assume that the true data-generating model, $x^{*}$, will have the largest posterior probability when the posterior contraction or the model selection consistency holds asymptotically. Indeed, such results have been established for a variety of model selection problems~\citep{yang2016computational, zhou2021complexity, zhuo2021mixing}. 
On the other hand, a model space $\mathcal{X}$ is often endowed with a natural neighborhood relation such that the degree of ``model fit'' (measured by goodness of fit and penalization on model size) between neighboring models tends to have a small difference. This implies that a small modification of a model may lead to a marginal change in the posterior probability. 
For example, in a Bayesian variable selection problem with the number of variables $p$, 
the model space can be represented as a set of binary vectors $\gamma\in\{0, 1\}^p$ where $\gamma_j = 1$ (resp. $\gamma_j = 0$) indicates $j$-th variable is included (resp. excluded) in the model. Let us assume that $\gamma^*$ is the true data generating model and has the highest posterior probability, and define a Hamming distance $d_{\mathrm{H}}(\gamma, \gamma') = \sum_{j=1}^p \ind \{\gamma_j \neq \gamma'_j\} $ for $\gamma, \gamma' \in \{0,1\}^p$. 
\begin{wrapfigure}{r}{0.3\textwidth}
  \vspace{-1em}
  \begin{center}
    \includegraphics[width=0.3\textwidth]{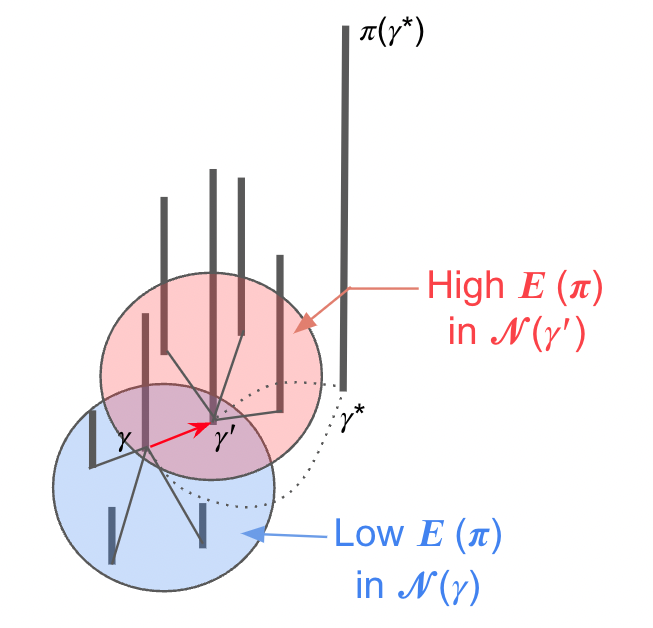}
  \end{center}
  \caption{An illustration of the reason that the  weight function $w(y \given x) = \pi(y)$ may fail. Each vertical bar represents the posterior probability on the corresponding state, and thin lines connect neighboring states. The expected value of the posterior probability of the neighborhood of $\gamma'$ (red) is typically greater than that of $\gamma$ (blue).}
  \vspace{-1em}
  \label{fig:weight}
\end{wrapfigure}
One often defines the neighborhood $\mathcal{N}(\gamma)$ as the set of states with a small Hamming distance to $\gamma$~\citep{hans2007shotgun, yang2016computational,zhou2021dimension}, and models in the neighborhood tend to have similar posterior probabilities. 
A state with a large Hamming distance from $\gamma^*$ usually has a small posterior probability because its structure is largely different from the true data-generating model. 
Figure~\ref{fig:weight} depicts such a relationship between posterior probabilities and structural similarities among the models. 
Given that $\gamma$ is the current state, Step 2 of Algorithm~\ref{alg:MTM} tends to propose a neighboring state $\gamma'$ with a higher posterior probability, so $\gamma'$ tends to be closer to $\gamma^*$ in Hamming distance, i.e. $d_{\mathrm{H}}(\gamma', \gamma^*) < d_{\mathrm{H}}(\gamma, \gamma^*)$. On average, the Hamming distance between a neighboring state of $\gamma'$ and $\gamma^*$ will be smaller than that between a neighboring state of $\gamma$ with $\gamma^*$. Therefore, the expected value of posterior probability of a neighboring state of $\gamma'$ (in the denominator part in \eqref{eq:acceptance_probability}, shown in red in Figure \ref{fig:weight}) will be typically larger than that of a neighboring state of $\gamma$ (in the numerator part in \eqref{eq:acceptance_probability}, shown in blue in Figure \ref{fig:weight}). 
If we resort to the unscaled weight function $w(\gamma' \given \gamma) = \pi(\gamma')$, this can cause a significantly low acceptance probability, especially when the number of trials $N$ is large (due to the law of large numbers, the acceptance ratio converges to the ratio of average posterior probabilities of the two neighborhoods); see also Appendix~\ref{subsec:ex} for a concrete example.
A locally balanced weight function~\eqref{eq:weight}, however, is able to mitigate this counterintuitive phenomenon by ``scaling" the weight function with respect to the current state. We shall demonstrate the effectiveness of locally balanced weight functions in the simulation studies. Although the scope of our study is limited to finite state spaces, employing \eqref{eq:weight} also improves the algorithm on continuous state spaces. We provide the relevant discussion in Appendix~\ref{subsec:state_space}.

In order for the MTM procedure to efficiently sample states from a high posterior region, it is sensible to assign a larger weight to those states with a higher  posterior probability in the proposal, i.e. Step 2 of Algorithm~\ref{alg:MTM}. To this end, we only require $h$ to be non-decreasing. Now we introduce the main result, which states that the mixing time bound can be improved linearly in the number of trials $N$ for our locally balanced weighting scheme.  

\begin{theorem}\label{thm:mixing}
Let $\MTM$ denote a Markov chain induced by Algorithm~\ref{alg:MTM} with a locally balanced weight function~\eqref{eq:weight} with some non-decreasing balancing function $h$. Define
\begin{align}
    \mathcal{S}(x) =  \{ x' \in \mathcal{N}(x) \colon \pi(x') / \pi(x) > p^{t_1} \},
\end{align}
 for some constant $t_1$ and $s_0 =  \max_{x \in \cX}|\mathcal{S}(x)|$. Suppose that the following conditions hold with $p \geq 2$ and $t_2,t_3,t_4 \geq 0$ being some constants that satisfy $t_1 < t_2$, $t_3\leq t_4 < t_2$.
\begin{enumerate}[(i)]
    \item There exists $x' \in \cN(x)$ such that $\pi(x') / \pi(x) \geq p^{t_2}$ for every $x \neq x^*$.
    \item For any $x \in \cX$, $ p^{t_3} \leq |\mathcal{N}(x)| \leq p^{t_4}$.
    \item $N = o\left(\min \left\{ \frac{h(p^{t_2+t_3-t_4})}{h(p^{t_1+t_4-t_3})}, \frac{p^{t_3}}{s_0}\right \}\right)$. 
\end{enumerate}
Then, we have the $\epsilon$-mixing time of $\MTM$
\begin{align}\label{eq:mixing}
    \tmix(\epsilon) = O\left( \frac{1}{N} (p^{-t_4} - p^{-t_2})^{-1} \ell(\Delta^*)\log \left (  \min_{x \in \cX} \pi(x) \right)^{-1}\right), 
\end{align}
where $\Delta^*$ is a path ensemble defined in Appendix~\ref{subsec:pf_thm1}.
\end{theorem}

\begin{remark}
By the recent result of~\cite[Lemma 3]{zhou2022rapid} based on a refined path argument, we can remove the term $\ell(\Delta^*)$ from~\eqref{eq:mixing} without changing the order of the bound. 
\end{remark}

\textit{Sketch of the proof.} The complete proof is deferred to Appendix~\ref{subsec:pf_thm1}. 
The first key step of our proof is to identify a suitable path ensemble $\Delta^*$ by defining a function $g: \cX \rightarrow \cX$ such that $g(x)$ has the highest posterior probability in the neighborhood of the state $x$. With $\pi$ being unimodal and $\cX$ being finite, repeated composition of $g$ for any $x$ will always lead to the mode $x^*$, that is $g^m(x) = x^*$ for some $m\in\mathbb{N}$, which we utilize to construct $\delta^*(x,y)$ for any $x,y \in \cX, x \neq y$. See Appendix~\ref{subsec:ex_construction_path} for an toy example. 
The other key step is to bound the congestion parameter $\rho(\Delta^*)$ defined in \eqref{eq:congestion}. To this end, we prove that the lower bound of the transition probability $\mathbf{P}(x, g(x))$ gets closer to $Np^{-t_4}$ asymptotically, for any state $x \in \cX\backslash \{x^*\}$.

Condition (i) implies that the posterior distribution is unimodal with the peak at the model $x^*$. We emphasize that the unimodality is with respect to the neighborhood relation $\cN$ and is usually satisfied with some appropriate choice of a large enough neighborhood. For example, for Bayesian variable selection, it is shown in~\cite{yang2016computational} that their proposed posterior may not be unimodal if a neighborhood is defined as a set of models within 1-Hamming distance, but expanding to the 2-Hamming distance neighborhood satisfies Condition (i) with $t_2 = 2$. Similarly, the community detection problem considered in~\cite{zhuo2021mixing} and the structure learning problem in~\cite{zhou2021complexity} satisfy (i) with some positive constant $t_2$ by choosing an appropriate $\mathcal{N}$.
Next, Condition (ii) states that the neighborhood size $|\cN(x)|$ should neither be too huge nor vary much from state to state. 
The purpose of this condition is to control the ratio $\MH(y,x) / \MH(x,y)$ in~\eqref{eq:weight} so that the posterior ratio $\pi(y) / \pi(x)$ dominates the term. 
Last but not least, Condition (iii) implies that the number of trials $N$ should not be arbitrarily large. This is consistent with  the empirical observation of~\cite{martino2017issues} that additional trials do not always result in better mixing. 
Condition (iii) also implies that $s_0 = o(p^{t_3})$ since $N\geq 1$. This seems to be a strong condition, since the number of neighboring states with a relatively large posterior probability with respect to the current state can be very large if we consider a state with the smallest posterior that is surrounded by states with greater posterior. Yet, it is a common practice to restrict the support of the prior to rule out unrealistic models by imposing some high-dimensional regularity assumptions, such as sparsity. 
For example, \cite{yang2016computational} introduced the parameter of the maximum number of important covariates in variable selection, and \cite{zhou2021complexity} considered the maximum in-degree and out-degree in structure learning to restrict the model space. In the same context, \cite{zhuo2021mixing} suggested the use of a feasible set for the initial partition in the community detection problem. In this regard, we may consider $\cX$ as the restricted space so that $s_0$ can be controlled. 
Most importantly, Condition (iii) provides a key idea on how to select the number of trials $N$, which will be discussed in Section~\ref{subsec:choice}.

\begin{remark}[Problem-specific applications]
Theorem~\ref{thm:mixing}
can be used to prove rapid mixing of the MTM algorithm for Bayesian variable selection~\citep{yang2016computational} and structure learning~\citep{zhou2021complexity} with order $N$ improvement of the mixing bound, given that the choice of $N$ satisfies the condition (iii).
\end{remark}

\begin{remark}[Computational benefits]
Since MTM with $N$ trials requires calculating $2N-1$ weight functions at each iteration, improving the mixing time bound by the factor of $N$ leads to the same overall computational complexity as the single-try MH algorithm until the convergence of the Markov chain. However, compared to the single-try MH with $N$ times longer sequential iterations, the calculation of $N$ weight functions can be efficiently done in parallel. 
For example, calculating multiple weight functions can often be converted into a series of matrix multiplication problems, and optimized linear algebra libraries such as BLAS \citep{blackford2002updated} can be used to exploit multiple execution units and pipe-lining. 
As examples, we outline parallelization strategies for the Bayesian variable selection (Section~\ref{subsec:simul:BVS}) and stochastic block model (Section~\ref{subsec:simul:SBM}) in Appendix \ref{subsec:parallelization}. 
Indeed, Table~\ref{table:th} in our simulation study results shows that MTM with moderate choice of $N$ significantly reduces the wall-clock time until the convergence of the Markov chain.

\end{remark}

\subsection{Choice of the number of trials}\label{subsec:choice}

Motivated by Condition (iii) in Theorem~\ref{thm:mixing}, we propose an algorithm to choose the number of trials $N$ under a general setting
applicable to high-dimensional model selection problems. For conciseness, we shall only focus on illustrating the case when we have balancing function $h(u) = \sqrt{u}$ and symmetric proposals (known $t_3=t_4$) so that Condition (iii) becomes $N = o\left(\min\{p^{(t_2-t_1)/2},p^{t_3}/s_0\}\right)$, although the algorithm can be easily generalized to other choices of $h$ and non-symmetric proposals.  

\begin{algorithm}[h]
\caption{The choice of the number of trials, when $h(u) = \sqrt{u}$ and known $t_3=t_4$.}
\label{alg:N}
\KwInput{An initial state $x_0$, neighborhood $\cN$, constant $\psi\in (0,1)$.}
Step 1. Calculate $\log_p (\pi(y_j)/ \pi(x_0))$ for all $y_j\in \cN(x_0)$, $j=1,\dots,|\cN(x_0)|$.\\
Step 2. Run k-means algorithm with k=2 on log-ratios $\{\log_p (\pi(y_j)/ \pi(x_0))\}$ to obtain a partition, $\mathcal{C}_1$ (set of small log-ratios) and $\mathcal{C}_2$ (set of big log-ratios).\\
Step 3.  Let $\hat{t}_2 = \min \mathcal{C}_2$, $\hat{t}_1 = \max \mathcal{C}_1$ and $\hat{s}_0 = |\mathcal{C}_2|$. $^\dagger$\algorithmfootnote{$^\dagger$ If $\hat{t}_2<t_4$, redefine $\mathcal{C}_1$,  $\mathcal{C}_2$ such that $\hat{t}_1 = \max \mathcal{C}_1 \le t_4<\min \mathcal{C}_2 = \hat{t}_2$. If no such $\mathcal{C}_2$ exists, set $\hat{t}_2=t_4, \hat{s}_0=1$ and $\hat{t}_1=\max \mathcal{C}_1$. }\\
\While{$ p^{(\hat{t}_2-\hat{t}_1)/2} < p^{t_3}/\hat{s}_0$}{Update $\mathcal{C}_1 \leftarrow \mathcal{C}_1\backslash \{\hat{t}_1\}$, $\hat{t}_1 \leftarrow \max \mathcal{C}_1$ and $\hat{s}_0 \leftarrow \hat{s}_0 + 1$. }
\KwOutput{The number of trials $N = \lfloor (p^{t_3}/\hat{s}_0)^\psi \rfloor$.}
\end{algorithm}

\begin{wrapfigure}{r}{0.4\textwidth}
  \vspace{-1.5em}
  \begin{center}
    \includegraphics[width=0.4\textwidth]{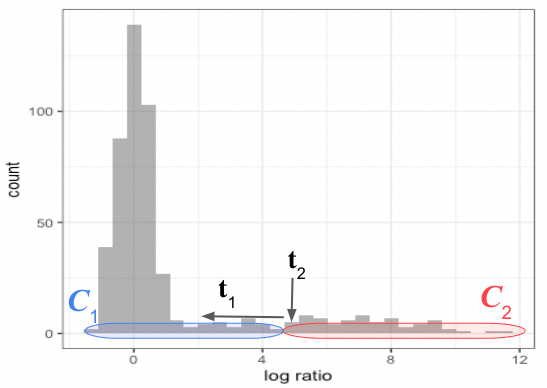}
  \end{center}
  \vspace{-1em}
  \caption{An illustration of Algorithm~\ref{alg:N}.}
  \vspace{-1em}
  \label{fig:N}
\end{wrapfigure}
Algorithm~\ref{alg:N} is depicted in Figure~\ref{fig:N}.
In Step 1, the log probability ratio is calculated for all neighboring states of the initial state $x_0$. Step 2 identifies two subsets $\mathcal{C}_1$ and $\mathcal{C}_2$, corresponding to ``bad moves'' and ``good moves'' respectively. In Step 3, motivated by Condition (i), we set $\hat{t}_2 = \min \mathcal{C}_2$, a lower bound of log probability ratios among good moves. 
After $t_2$ is fixed, we first set $\hat{t}_1 = \max \mathcal{C}_1$, $\hat{s}_0 = |\mathcal{C}_2|$ and gradually decrease $\hat{t}_1$ (so that $\hat{s}_0$ increases) until the ``crossing'' $p^{(\hat{t}_2-\hat{t}_1)/2} \ge p^{t_3}/\hat{s}_0$ happens.
The choice of $N = \lfloor (p^{t_3}/\hat{s}_0)^\psi \rfloor$ immediately after the crossing point approximates the worst-case scenario of $\min\{p^{(t_2-t_1)/2},p^{t_3}/s_0\}$, hence providing a conservative estimate of $N$.
We set $\psi=0.9$ to match with the asymptotic dominance condition~(iii).

We run Algorithm~\ref{alg:N} only once for an initial state $x_0$. It is based on the insight that for model selection problems, log probability ratios can often be partitioned into ``bad moves'' and ``good moves''. Although we considered any $x\in \mathcal{X}\backslash \{x^*\}$ in Condition (i), analyzing the initial state $x_0$ is not only computationally simple but also gives a good guess of $t_2$ since as chain proceeds to the highest posterior state $x^*$, the number of good moves and the magnitude of ratios generally decreases, for example in the Bayesian variable selection problem \citep{zhou2021dimension}. 
Another possible way of using Algorithm~\ref{alg:N} is to re-evaluate $N$ after a certain number of MCMC iterations. 
We note that another clustering algorithm can be substituted for the k-means algorithm in Step 2.

\section{Simulation studies}\label{sec:simulation}

In this section, we present the results for two Bayesian model selection problems, BVS and SBM. We defer the result for spatial clustering models to Appendix~\ref{sec:append.SCM} due to the page limit. 
Throughout the simulation studies, we use symmetric proposal schemes. We run single-try MH and MTM with different choices of $N \in \{ 5,10,50,100,500,1000,2000,5000\}$ and four different choices of weight functions, one with (unscaled) ordinary weight function $w_{\mathrm{ord}}(y\given x) = \pi(y)$, and three with the proposed locally balanced weight functions $w_{\mathrm{sqrt}}(y\given x) = \sqrt{\pi(y)/\pi(x)}$, $w_{\mathrm{min}}(y\given x) = \min\{1, \pi(y)/\pi(x)\}$, and $w_{\mathrm{max}}(y\given x) = \max\{1, \pi(y)/\pi(x)\}$. 
We consider the scenarios where the data-generating model $x^*$ receives the highest posterior probability to verify our main theoretical results. Similarly to \cite{zhou2021dimension}, two performance measures are considered: 1) $H = \min\{t:x_t = x^*\}$ (hitting iteration): the number of MCMC iterations until the Markov chain $(x_0,x_1,\dots)$ hits $x^*$, 2) $T_H$ (wall-clock hitting time): the wall-clock time taken until the chain hits $x^*$.

\subsection{Bayesian variable selection (BVS)}\label{subsec:simul:BVS}
Consider a high-dimensional linear model with response vector $\bm{y}\in\mathbb{R}^n$ and design matrix $\bm{X}\in\mathbb{R}^{n\times p}$, where the number of predictors $p$ is much larger than the sample size $n$. 
BVS seeks to find the best subset of predictors, denoted as a binary vector $\gamma\in\{0,1\}^p$ described in Section \ref{subsec:weight}, using the sparsity-inducing prior. We adopt the prior introduced by \cite{yang2016computational}:
\begin{align*}
\text{Linear model:}& \qquad \bm{y}=\bm{X}_{\gamma} \bm\beta_{\gamma}+\bm\epsilon, \quad \bm\epsilon \sim \mathsf{N}\left(\bm{0}, \phi^{-1} \bfI_{n}\right)\\
\text{Prior on $\bm\beta_\gamma$ and $\phi$:}&\qquad \bm\beta_{\gamma} \mid \phi,\gamma \sim \mathsf{N}\left(\bm{0}, \mathscr{G}\phi^{-1}(\bm{X}_{\gamma}^\top \bm{X}_{\gamma})^{-1}\right), \quad  \pi(\phi) \propto 1/\phi \\
\text{Sparsity prior:}&\qquad \pi(\gamma) \propto p^{-\kappa|\gamma|} \ind\left[|\gamma| \leq s_{max}\right],
\end{align*}
where $\bm{X}_\gamma$ is a submatrix of $\bm{X}$ consisting of all $j$-th columns with $\gamma_j = 1$, $\bm\beta_\gamma\in\bbR^{|\gamma|}$ is the subvector of $\bm\beta\in\bbR^{p}$ with nonzero coefficients, where $|\gamma| = \sum_{j=1}^p\gamma_j$. The hyperparameters $\mathscr{G}$, $\kappa$ and $s_{\mathrm{max}}$ control 
the weight of prior effect \citep{zellner1986assessing}, sparsity strength and the maximum model size, respectively.
We sample $\gamma$ from the posterior distribution $\pi(\gamma\given \bm{y})$ where other parameters $\bm\beta_\gamma$ and $\phi$ are marginalized out; see Appendix \ref{sec:append.BVS} for details. 
We consider the following random walk proposal: 
\begin{equation*}
    \MH(\gamma,\gamma') = \begin{cases}
    \frac{1}{p}\ind_{\cN_1(\gamma)} (\gamma') \text{ (single flip)}\quad &\text{ if $|\gamma|<s_{\mathrm{max}}$,}\\
    \frac{1}{2p}\ind_{\cN_1(\gamma)} (\gamma')+\frac{1}{2|\gamma|(p-|\gamma|)}\ind_{\cN_2(\gamma)} (\gamma')\text{ (single or double flip)} &\text{ if $|\gamma|=s_{\mathrm{max}}$,}
    \end{cases}
\end{equation*}
where $\mathcal{N}_1(\gamma)= \{\gamma': d_{\mathrm{H}}(\gamma,\gamma') =1\}$ and $\mathcal{N}_2(\gamma)= \{\gamma'= (\gamma\cup \{j\})\backslash \{\ell\} : j\not\in \gamma, \ell\in \gamma \}$.
With $n=1000$ and $p=5000$, we consider different settings of design matrix where each row $\bm{x}_i$ is i.i.d. sampled from $\bm{x}_i \sim \mathsf{N}(\bm{0},\bfI_p)$ (independent) or $\bm{x}_i \sim \mathsf{N}(\bm{0}, \bm\Sigma)$, $\Sigma_{jj'} = \exp(-|j-j'|)$ (dependent). The response vector $\bm{y}$ is generated from $\bm{y} \sim \mathsf{N}( \bm{X} \bm\beta, \bfI_n)$, with $\bm\beta$ having only 10 nonzero entries at the first 10 coordinates: 
$
\bm\beta = \text{SNR}\sqrt{(\log p) /n}(2,-3,2,2,-3,3,-2,3,-2,3)
$
where SNR $\in\{2,4\}$ is a signal-to-noise ratio. For each of the four  settings, we simulate 50 replicate datasets.
Per each replicate, algorithms are randomly initialized with state $\gamma_0$ such that $\gamma_0 \cap \gamma^* =\emptyset$ and $d_{\mathrm{H}}(\gamma_0,\gamma^*) = 20$ which implies $H=20$ is the minimum required hitting iteration.
All hyperparameter setups, the number of MCMC iterations, and other details are deferred to Appendix~\ref{sec:append.BVS}.

\begin{figure}[h]
    \centering
    \begin{subfigure}{0.34\textwidth}
    \includegraphics[trim={0 0 0 0.5cm}, clip, width=\textwidth]{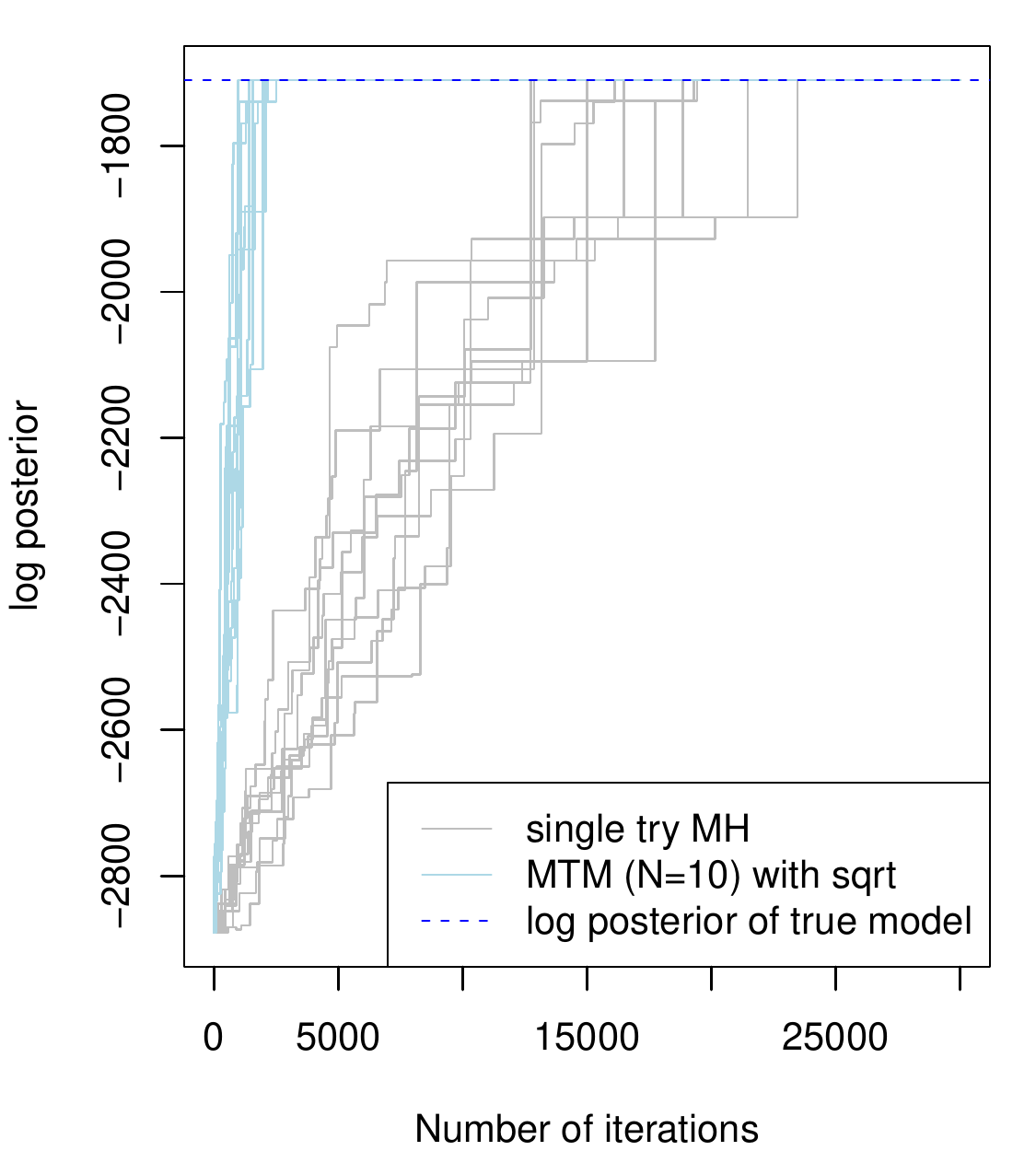}
    \end{subfigure}
    \hfill
    \begin{subfigure}{0.64\textwidth}
    \includegraphics[trim={0 0.5cm 0 0}, clip, width=\textwidth]{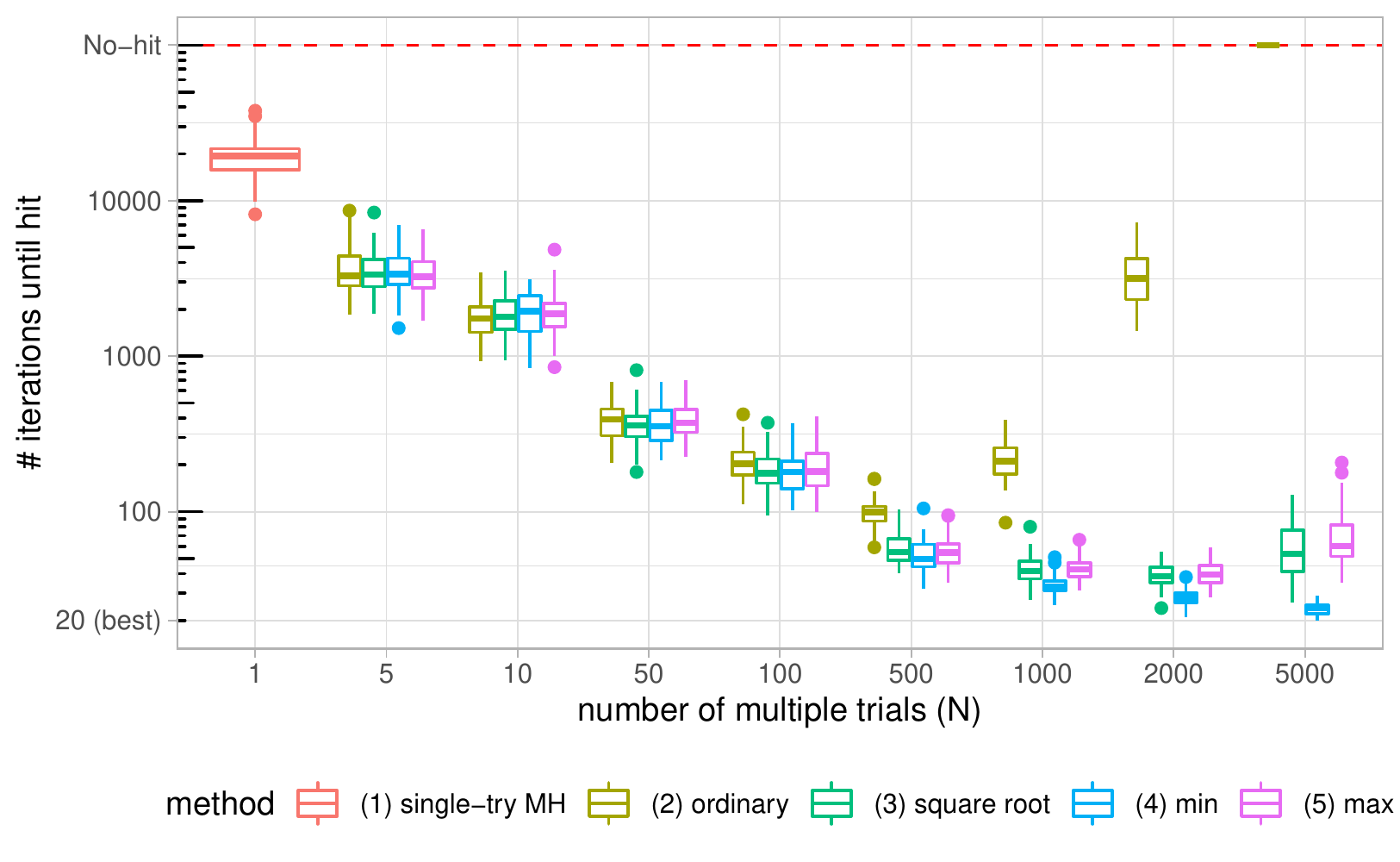}
    \end{subfigure}
    \caption{BVS simulation results for independent design and SNR $= 4$. (Left) Trace plots of unnormalized log-posterior probability using single-try MH and MTM $(N=10)$ with square root weighting function. (Right) Boxplot of $H$, the number of iterations until hit, against different choices of weight function and $N$ based on 50 replicates.}
    \label{fig:bvs}
\end{figure}

Results under independent design with SNR $=4$ are summarized in Figure~\ref{fig:bvs}. The trace plot shows that single-try MH reaches the true state $\gamma^*$ at around 20,000 iterations, whereas the MTM with $N=10$ reaches the true state $\gamma^*$ at around 2,000 iterations, smaller by a factor of 10. 
The boxplot confirms our findings that a larger number of trials $N$ is not always desirable, especially when using the ordinary weight function $w_{\text{ord}}$. The proposed locally balanced weight functions using three different choices of $h$ perform much better than $w_{\mathrm{ord}}$ for large $N$. 
Especially, the performance of $w_{\mathrm{min}}$ seems to be robust to the choice of large $N$ in the high SNR case. 
This can be explained by the fact that the denominator in \eqref{eq:acceptance_probability} is often larger in high SNR data, and the choice of $\min\{1,\pi(y)/\pi(x)\}$ limits its growth; see also Section~\ref{subsec:weight}. 
The choice of $N$ suggested from Algorithm~\ref{alg:N} with $\psi = 0.9$ is overall reasonable, which has a median of 349 over 50 datasets.
The detailed results for other settings are available in Appendix~\ref{sec:append.BVS}.

\subsection{Stochastic block model (SBM)}\label{subsec:simul:SBM}
SBM \citep{holland1983stochastic} is a popular generative model of an undirected graph which assumes a block structure of edge connection probabilities $\bm{Q}$ to describe community structure.
We consider Bayesian SBM \citep{nowicki2001estimation} by assigning prior on $\bm{Q}$ and partition $\mathbf{z}=(z_1,\dots,z_p)\in\{1,\ldots,K\}^p$ with fixed number of blocks $K$, where $z_i$ indicates membership label of $i$-th node. The main goal is to find partition $\mathbf{z}$ (up to a label permutation) that best describes the graph, denoted as an adjacency matrix $\bm{A}\in\{0,1\}^{p\times p}$.
We follow the prior of \cite{zhuo2021mixing}: \vspace{-1mm}
 \begin{align*}
    \text{Edge appearance: }&\quad  A_{ij} \mid \bmQ, \mathbf{z} \stackrel{\text{ind}}{\sim} \mathsf{Bernoulli}(Q_{z_iz_j}), 1\le i<j\le p \\
    \text{Prior on blockwise probabilities: }&\quad Q_{uv} \stackrel{\text{iid}}{\sim} \mathsf{Beta}(\kappa_1,\kappa_2), 1 \le u \le v \le K\\
    \text{Prior on partition:} &\quad  \pi(\mathbf{z}) \propto \ind (\mathbf{z} \in S_\alpha),
\end{align*} 
with restricted partition space $S_\alpha= \{\mathbf{z}: \sum_{i=1}^p \ind (z_i = u) \in \left[\frac{p}{\alpha K},\frac{\alpha p }{K}\right] \text{ for all } u = 1,\dots,K\}$ for some $\alpha>0$, which excludes partitions whose block sizes differ too much. 
We sample $\mathbf{z}$ from the posterior  $\pi(\mathbf{z}\given \bm{A})$ where $\bm{Q}$ is marginalized out; see Appendix \ref{sec:append.SBM} for details. We consider the following proposal where $\tilde{d}_\mathrm{H}$ is a permutation-invariant Hamming distance \citep{zhang2016minimax}:
\begin{equation*}
    \MH(\mathbf{z},\mathbf{z}') = 1/(p(K-1))\ind_{\cN(\mathbf{z})} (\mathbf{z}'), \quad \mathcal{N}(\mathbf{z}) = \{\mathbf{z}': \tilde{d}_\mathrm{H}(\mathbf{z},\mathbf{z}')=  1\}  \quad \text{(single flip).}
\end{equation*}
With the number of nodes $p=1000$, we consider two different settings of the blocks $K\in\{2,5\}$ with true partition $\mathbf{z}^*=(1,\dots,1,\cdots,K,\dots,K)$ being balanced ($p/K$ times each) which ensures $\mathbf{z}^*\in S_\alpha$. In each setting of $K$, we generate a graph from the homogeneous SBM \citep{geng2019probabilistic,zhuo2021mixing}, where $Q_{uv} = a$ if $u=v$ and $Q_{uv} = b$ otherwise.
 We choose probability pairs $(a,b)$ based on the Chernoff-Hellinger divergence ($\mathrm{CH}$)~\citep{abbe2017community}, defined as $\mathrm{CH} \coloneqq p(\sqrt{a}-\sqrt{b})^2/(K\log p)$ for $a,b\asymp (\log p) / p$, which can be interpreted as a signal-to-noise ratio. Our choices are $\mathrm{CH} \approx 2$ (weak signal) and $\mathrm{CH} \approx 10$ (high signal), both being greater than 1 which is necessary to satisfy posterior consistency~\citep{zhang2016minimax,abbe2017community,zhuo2021mixing}. 
For each of the the four combinations of $(K,\mathrm{CH})$, we simulate 50 replicate datasets. For each replicate, algorithms are randomly initialized with state $\mathbf{z}_0$ such that $\tilde{d}_\mathrm{H}(\mathbf{z}_0,\mathbf{z}^*)=400$ which implies $H=400$ is the minimum required hitting iteration. All hyperparameter setups, the number of MCMC iterations, and other details are deferred to Appendix~\ref{sec:append.SBM}.

\begin{figure}[h]
    \centering
    \begin{subfigure}{0.30\textwidth}
    \includegraphics[width=\textwidth]{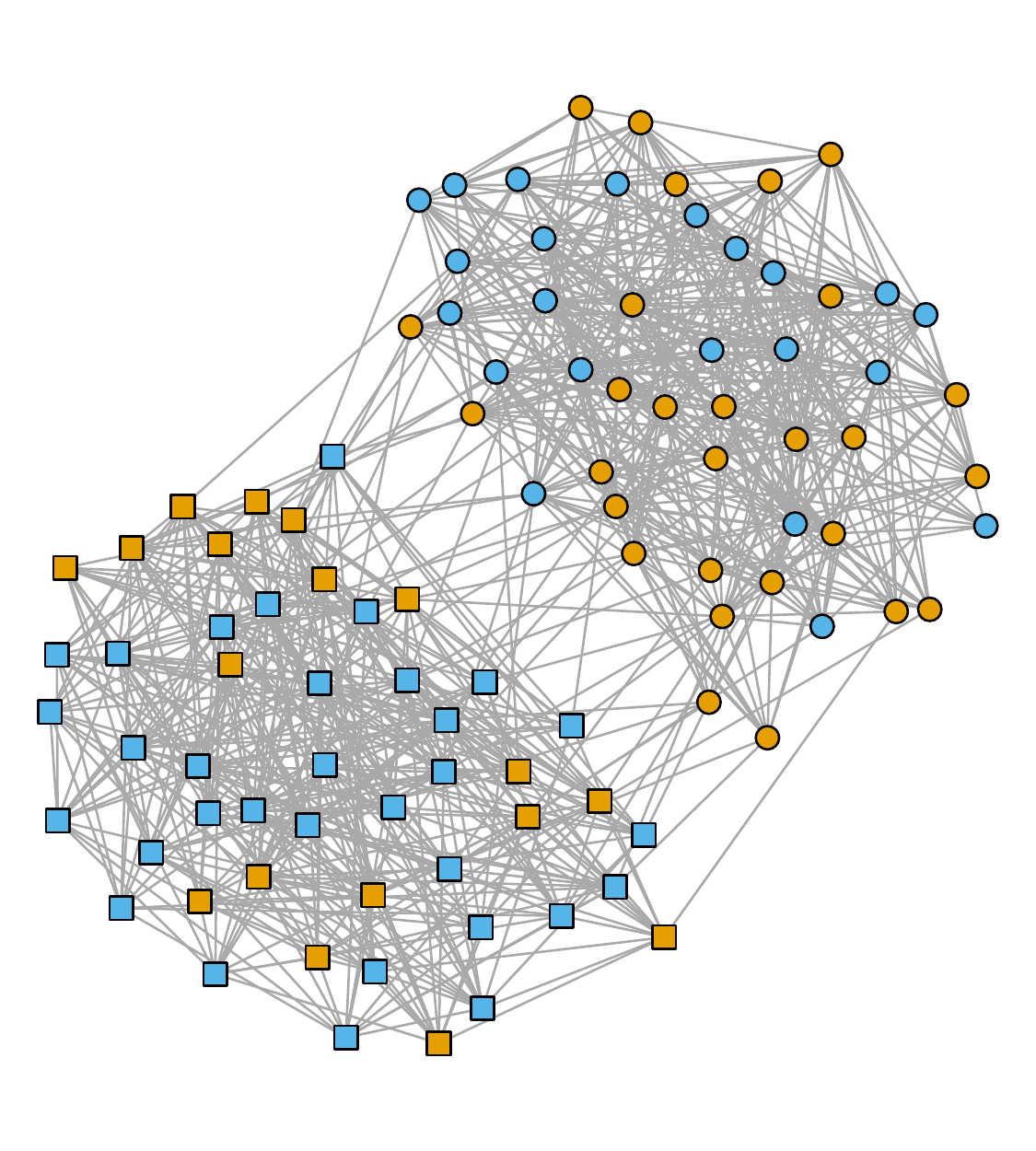}
    \end{subfigure}
    \hfill
    \begin{subfigure}{0.60\textwidth}
    \includegraphics[width=\textwidth]{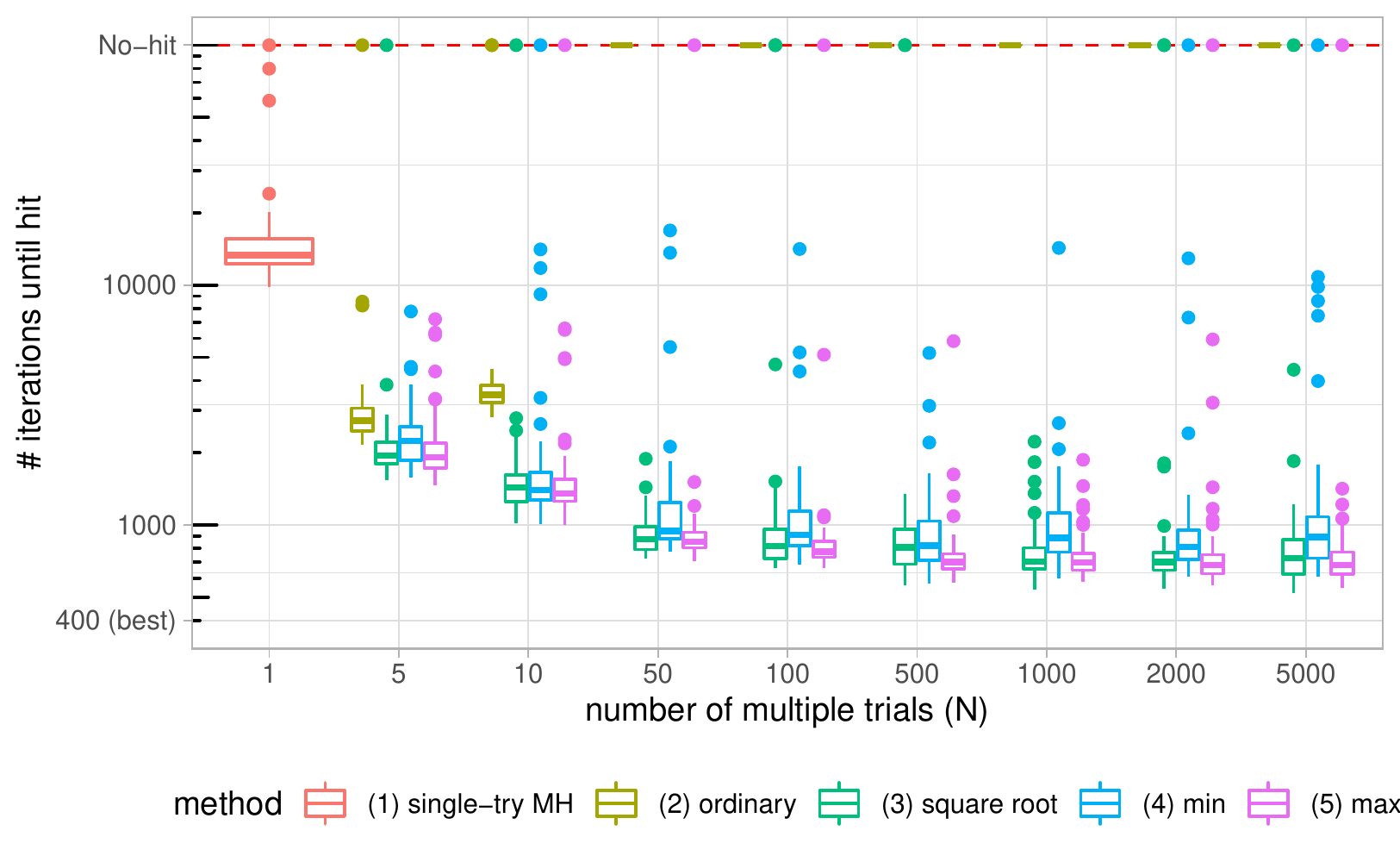}
    \end{subfigure}
    \caption{(Left) An example graph generated from SBM with $p=100$, $K=2, a=0.32$, and $b=0.02$ so that $\mathrm{CH} \approx2$. Node shape (square, circle) denotes the true partition. Node color (blue, yellow) represents the initial state $\mathbf{z}_0$ satisfying $\tilde{d}_\mathrm{H}(\mathbf{z}_0,\mathbf{z}^*)=40$. (Right) Boxplot of $H$ against different weight functions and $N$ based on 50 replicates when $p=1000, K=2$, and $\mathrm{CH} \approx2$.} 
    \label{fig:sbm}
\end{figure}

Results for $K=2$ and $\mathrm{CH}\approx 2$ are summarized in Figure~\ref{fig:sbm}. 
The boxplot of $H$ is similar to that of Figure~\ref{fig:bvs}: $w_{\mathrm{ord}}$ quickly deteriorates as $N$ increases, but the locally balanced weight functions mostly converge with a larger choice of $N$, although some Markov chains fail to converge due to the low $\mathrm{CH}$.
Especially, now the performance of $w_{\mathrm{max}}$ or $w_{\mathrm{sqrt}}$ is slightly better than $w_{\mathrm{min}}$ when $N$ is large. 
This can be explained by the fact that in low signal case ($\mathrm{CH}\approx 2$), there are few neighboring states that have high probability ratios. It is important to catch such states to increase the acceptance ratio \eqref{eq:acceptance_probability}, and the $w_\mathrm{min}$ cannot do this efficiently as it is upper bounded by $1$.
The choice of $N$ suggested from Algorithm~\ref{alg:N} with $\psi=0.9$ has a median of 15 over 50 datasets.
The detailed simulation results for other settings are available at Appendix~\ref{sec:append.SBM}.

Finally, Table~\ref{table:th} provides the summary of wall-clock hitting time $T_H$ for the previous two examples, which shows the clear computational benefit of MTM over the single-try MH thanks to the parallelism. As $N$ increases, computation becomes more demanding compared to the decreasing rate of $H$, so the optimal choice of $N$ with respect to $T_H$ is often less than the minimizer of $H$.
\begin{table}[t]
    \small
  \caption{Summary of $T_H$, median wall-clock time (in seconds) until the Markov chain hit the highest posterior state, for the previous two examples using the square root weight function over 50 replicates.}
  \centering
\label{table:th}
  \begin{tabular}{c c c c c c c c c c c}
    \toprule
    $N$&  & 1 & 5 & 10 & 50 & 100 & 500 & 1000 & 2000 & 5000 \\
    \midrule
    \multirow{2}{*}{$T_H$} & BVS, indep, SNR$=4$  & 1.30 & 0.81 & 0.46 & 0.11 & 0.07 & \textbf{0.04} & 0.05 & 0.09 & 0.30 \\
    \cmidrule(lr){2-11}
     & SBM, $K=2$, $\mathrm{CH}\approx 2$ &  0.82 & 0.40 & \textbf{0.35} & 0.48 & 0.75 & 1.49 & 2.15 & 3.78 & 9.30\\
    \bottomrule
  \end{tabular}
\end{table}

\textbf{Spatial clustering models and additional information.} 
In Appendix \ref{sec:append.SCM}, we present simulation study results with spatial clustering models where the key findings align with those from BVS and SBM. In Appendix \ref{sec:append.multimodal}, we also analyze the behavior of MTM algorithms on multimodal target distributions. The code is available at \url{https://github.com/changwoo-lee/rapidMTM}.

\section{Real data applications and discussion}\label{sec:realanddiscussion}

\textbf{Real data applications. }
We carry out two real data application analyses to corroborate our findings beyond the scope of the simulation settings. The performance measures suggested in Section~\ref{sec:simulation}, the hitting iteration $H$ and the corresponding wall-clock hitting time $T_H$, are no longer available for real data analysis since we cannot identify the true data generating state $x^*$. Instead, we evaluate the performance based on the acceptance rate and the number of unique states visited. 
Both real data analyses confirm our theory: the performance deteriorates significantly as the number of trials $N$ grows when using $w_{\mathrm{ord}}$, but it does not when using $w_{\mathrm{sqrt}}$, $w_{\mathrm{min}}$, or $w_{\mathrm{max}}$.

The first real data application is to find a subset of genetic variants from a genome-wide association study (GWAS) dataset that best explains the cup-to-disk ratio averaged over two eyes, which is used to assess the risk of glaucoma. We employ BVS model described in Section~\ref{subsec:simul:BVS}, with the number of samples $n=5418$ and the number of gene variants $p = 7255$ which are selected after the preliminary screening procedure described in~\cite[Section 6]{zhou2021dimension}. We tabulate the acceptance rate and the number of unique states visited along the different number of trials $N$, and also report the posterior inclusion probabilities of the top 10 genetic variants in Appendix~\ref{subsec:real:BVS} and \ref{sec:addtable} with discussion.

The other application is to learn the underlying directed acyclic graph (DAG) model for the single-cell RNA dataset on Alzheimer’s disease~\cite{jiang2020scread}. We preprocess the dataset as~\cite[Section 6]{chang2022order}, which yields the sample size $n=1666$ and 73 genes, and we utilize the DAG model described in~\cite[Section 2.2]{chang2022order}, whose model size is equal to $73! \approx 4.5 \times 10^{105}$. We defer the performance result and other details (including the log-posterior trace plots for 4 different weight functions) to Appendix~\ref{subsec:real:structure_learning}.

\textbf{Discussion and future work. }
We prove that, under some assumptions, the mixing time bound of the MTM algorithm is smaller than that of the MH algorithm by a factor of the number of trials. Motivated by the observation that popular choices of weight function can cause mixing problems, we propose locally balanced weight functions for which a mixing time bound is proven. 
We also suggest a theoretically guided choice of $N$ for practical use.
Future research may investigate the mixing time of other generalizations of MTM such as the case with correlated weight functions. It will also be of interest to extend MTM and analyze its mixing time to better handle the potential multimodality in the target distribution, by combining it with techniques such as annealing or tempering \cite{casarin2013interacting}.

\begin{ack}
We thank James Cai for providing us with the processed single-cell data set on Alzheimer's disease.  
We thank the submitters and participants of the two dbGaP studies (phs000308.v1.p1 and phs000238.v1.p1), which were funded by NIH. 
We thank Florian Maire, Giacomo Zanella, and Philippe Gagnon for the helpful discussion at ISBA 2022. The research of Changwoo Lee, Zhao Tang Luo, and Huiyan Sang was partially supported by NSF DMS-1854655 and DMS-2210456. 
\end{ack}

\bibliography{arxiv_cameraready.bib}
\bibliographystyle{plain}

\newpage 

\begin{center}
\LARGE{\textbf{Appendices}}
\end{center}

\setcounter{section}{0}
\renewcommand\thesection{\Alph{section}}
\renewcommand\thesubsection{\Alph{section}.\arabic{subsection}}

\setcounter{equation}{6}
\setcounter{figure}{4}
\setcounter{table}{1}
In Appendix~\ref{sec:proof}, we provide proofs of Proposition~\ref{prop:path}, Proposition~\ref{prop:weight}, and Theorem~\ref{thm:mixing} in the main text.
In Appendix~\ref{sec:append.simulation}, we provide more details of the Bayesian variable selection (BVS) and stochastic block model (SBM) in Section~\ref{sec:simulation} as well as a detailed simulation study on the spatial clustering model (SCM).
In addition, we study the performance of multiple-try Metropolis for the case with multimodal target distributions, following the BVS simulation setting of~\cite{zhou2021dimension}. 
In Appendix~\ref{sec:appendix:real}, we present details of two real data application analyses.
In Appendix~\ref{sec:adddiscussion}, we add a more detailed discussion on parallelization, state space of interest, and the behavior of MTM on continuous state space.
Finally, we provide additional tables on the real data analysis results in Appendix~\ref{sec:addtable}.

\section{A path method for proving the mixing time bound for multiple-try Metropolis algorithm}\label{sec:proof}

\subsection{Proof of Proposition~\ref{prop:path}}\label{subsec:pf_path}

This section aims to provide a summary of the existing results on proving mixing time bound via path methods. We refer readers to~\cite{levin2017markov, saloff1997lectures, diaconis1991geometric, sinclair1992improved} for more details.
Let $\bP(x,y)$ denote the transition probability for an irreducible, aperiodic chain on the finite state space $\cX$.
Assume $\bP$ satisfies the detailed balance condition with respect to the probability distribution $\pi$, that is, $\pi(x)\bP(x,y) = \pi(y) \bP(y,x)$ for $x, y \in \cX$, which leads to $\pi$ being stationary for $\bP$~\citep[Proposition 1.20]{levin2017markov}. $\bP$ may be thought as a $|\cX| \times |\cX|$ stochastic matrix, which means $\bP(x, y) \geq 0$ and $\sum_{z \in \cX} \bP(x, z) = 1$ for all $x,y \in \cX$, and $\pi$ can be regarded as a $|\cX|$-dimensional stochastic vector since $\sum_{x \in \cX} \pi(x) = 1$. By the spectral decomposition, we can sort the eigenvalues of $\bP$ as
\begin{align*}
    1 =  \lambda_0 > \lambda_1 \geq \cdots \geq \lambda_{|\cX|-1} > -1, 
\end{align*}
due to $\bP$ being stochastic, irreducible and aperiodic~\citep[Lemma 12.1]{levin2017markov}. 
Let $\lambda_{\mathrm{max}} = \max \{ \lambda_1, |\lambda_{|\cX|-1} |\}$. 
We say $ \gap (\bP) = 1-\lambda_{\mathrm{max}}$ is the \textit{spectral gap} of the chain $\bP$. Intuitively, if the spectral gap is close to zero, the chain requires a large number of steps to be close to the stationary distribution in total variation distance. The following lemma draws the connection between the spectral gap and mixing time defined in Section~\ref{subsec:mixing}.
\begin{lemma}\label{lemma:mix}
Let $\tmix(\epsilon)$ denote $\epsilon$-mixing time defined in Section~\ref{subsec:mixing}, then
\begin{align*}
    \tmix(\epsilon) \leq  2\{\gap (\bP)\}^{-1} \left[\log (1/ \epsilon) + \log \left \{  \min_{x \in \cX} \pi(x) \right\}^{-1}\right].
\end{align*}
\end{lemma}
\begin{proof}
We consider $\bP_{\mathrm{lazy}} = (\mathbf{I} + \bP)/2$ so that all eigenvalues of $\bP_{\mathrm{lazy}}$ are positive and the spectral gap becomes $\gap(\bP_{\mathrm{lazy}}) = 1 - \lambda_1$.
By (1.10) of~\cite[Proposition 3]{diaconis1991geometric}, 
\begin{align*}
    4\max_{x \in \cX}\TV{\bP_{\mathrm{lazy}}^t(x,\cdot) - \pi(\cdot)}^2 \leq \max_{x \in \cX} \{(1 - \pi(x)) /\pi(x)\} \exp[-2t \gap(\bP_{\mathrm{lazy}})],
\end{align*}
for $t \in \bbN$. We set $ \max_{x \in \cX}  \{(1 - \pi(x)) /\pi(x)\} \exp[-2t \gap(\bP_{\mathrm{lazy}})] \leq 4 \epsilon^2$, and solving the inequality with $t$ gives
\begin{align*}
    & t \geq \{ \gap(\bP_{\mathrm{lazy}})\}^{-1} \left( \frac{1}{2} \max_{x \in \cX}  \log(\pi(x)^{-1} -1) - \log2 + \log (1/\epsilon) \right).
\end{align*}
Since
\begin{align*}
    &\{ \gap(\bP_{\mathrm{lazy}})\}^{-1} \left( \frac{1}{2} \max_{x \in \cX}  \log(\pi(x)^{-1} -1) - \log2 + \log (1/\epsilon) \right) \\ & \leq  \{\gap (\bP_{\mathrm{lazy}})\}^{-1} \left[\log (1/ \epsilon) + \log \left \{  \min_{x \in \cX} \pi(x) \right\}^{-1}\right] \\
    & \leq  2 \{\gap (\bP)\}^{-1} \left[\log (1/ \epsilon) + \log \left \{  \min_{x \in \cX} \pi(x) \right\}^{-1}\right],
\end{align*}
it follows that to achieve ``$\epsilon$-mixing'', it suffices to choose 
\begin{align*}
    t \geq 2 \{\gap (\bP)\}^{-1} \left[\log (1/ \epsilon) + \log \left \{  \min_{x \in \cX} \pi(x) \right\}^{-1}\right]. 
\end{align*}
This concludes the proof.
\end{proof}
Our next interest is to find the lower bound of the spectral gap, which leads to the upper bound of the mixing time. It is often the case that the functional analysis tool is useful to establish such bound. Let $\bbR^{\cX} = \{f \colon \cX \rightarrow \bbR \}$ and 
$\ell^2(\pi) \subset \bbR^{\cX} $ be the vector space equipped with an inner product $\inner{\cdot, \cdot}$, which is defined by $\inner{f_1,f_2} = \sum_{x \in \cX} f_1(x)f_2(x)\pi(x)$ for all $f_1, f_2 \in \ell^2(\pi)$. We can regard the transition probability $\bP$ as a function operator in $\ell^2(\pi)$, which can be defined as $\bP f (x) = \sum_{y \in \cX} \bP(x,y) f(y)$. We define the \textit{Dirichlet form} associated to the pair $(\bP, \pi)$ by
\begin{align*}
    \cE(f_1,f_2) = \inner{(\mathbf{I}-\bP)f_1, f_2} \quad \text{ for } f_1, f_2 \in \ell^2(\pi).
\end{align*}
By the reversibility of $\bP$, we can easily check that $\cE(f) = \cE(f,f) = \frac{1}{2} \sum_{x,y \in \cX} [f(x) - f(y)]^2 \pi(x) \bP(x,y)$. The spectral gap can be defined 
using the Dirichlet form as follows \citep[Remark 13.8]{levin2017markov}:
\begin{align*}
    \gap (\bP) = \min_{\stackrel{f \in \ell^2(\pi)}{\mathrm{Var}_{\pi}(f) \neq 0}} \frac{\cE(f)}{\mathrm{Var}_{\pi}(f)},
\end{align*}
where $\mathrm{Var}_{\pi}(f) = \sum_{x\in \cX} (f(x) - \bbE_{\pi} f)^2 \pi(x)$. This definition also has a link to the famous \textit{Poincar\'e inequality}~\cite{bakry2014analysis}, that is,  $\mathrm{Var}_{\pi}(f) \leq C \cE(f)$ for all $f \in \ell^2(\pi)$, because the smallest constant $C$ is equal to $\{\gap (\bP)\}^{-1}$. The next lemma uses Poincar\'e inequality and an arbitrary path ensemble $\Delta$ defined in the main text. 

\begin{lemma}[Corollory 6, \cite{sinclair1992improved}]\label{lemma:gap}
For an arbitrary path ensemble $\Delta$,
\begin{align*}
    \gap (\bP) \geq \frac{1}{\rho(\Delta)  l(\Delta)}.
\end{align*}
where $\ell(\Delta) = \max_{x,y} |\delta(x,y)|$ and
\begin{align*}
    \rho(\Delta) = \max_{(u,v) \in E} \frac{1}{\pi(u)\bP(u,v)} \sum_{x,y : \delta(x,y) \ni (u,v) }\pi(x)\pi(y) 
\end{align*}
\end{lemma}
\begin{proof}
We follow the proof given in~\cite[Theorem 3.2.1]{saloff1997lectures}. 
For each $(x,y) \in \cX \times \cX$ and for any function $f \in \bbR^\cX$, we can write $f(y) - f(x) = \sum_{(u,v) \in \delta(x,y)} f(v) - f(u)$. By using Cauchy-Schwarz, multiplying $\pi(x)\pi(y) /2$, and summing over $x$ and $y$,
\begin{gather*}
 |f(y) - f(x)|^2 \leq |\delta(x,y)| \sum_{(u,v) \in \delta(x,y)} |f(v) - f(u)|^2 \\ 
\Longrightarrow \quad  \underbrace{\frac{1}{2} \sum_{x,y} |f(y) - f(x)|^2 \pi(x) \pi(y)}_{= \mathrm{Var}_\pi(f)} \leq \frac{1}{2} \sum_{x,y} |\delta(x,y)| \sum_{(u,v) \in \delta(x,y)} |f(v) - f(u)|^2 \pi(x) \pi(y),
\end{gather*}
where the right-hand side becomes
\begin{align*}
    & \frac{1}{2} \sum_{(u,v) \in E} \left \{ \frac{1}{\pi(u) \bP(u,v)} \sum_{x,y : \delta(x,y) \ni (u,v)} |\delta(x,y)| \pi(x) \pi(y) \right\} |f(v) - f(u)|^2 \pi(u) \bP(u,v) \\
    \leq & \underbrace{\max_{(u,v) \in E}  \left \{ \frac{1}{\pi(u) \bP(u,v)} \sum_{x,y : \delta(x,y) \ni (u,v)} \pi(x) \pi(y) \right\}}_{ \rho(\Delta)}  \underbrace{\left(\max_{x,y} |\delta(x,y)|\right)}_{ \ell(\Delta)}  \underbrace{\left(\frac{1}{2}\sum_{(u,v) \in E} |f(v) - f(u)|^2 \pi(u) \bP(u,v) \right)}_{ \cE(f)}.
\end{align*}
This satisfies the Poincar\'e inequality, which yields the conclusion. 
\end{proof}
By combining the results of Lemma~\ref{lemma:mix} and Lemma~\ref{lemma:gap}, we get the conclusion of Proposition~\ref{prop:path}. 

\subsection{Proof of Proposition~\ref{prop:weight}}\label{subsec:pf_prop1}

\begin{proof}
Recall that the form of the weight function suggested in~\cite{liu2000multiple} is given by
\begin{align*}
    w(y \given x) = \pi(y) \MH (y,x) \lambda(y,x),
\end{align*}
where $\lambda(x,y)=\lambda(y,x)$ is a non-negative symmetric function in $x$ and $y$, and satisfies $\lambda(x,y) > 0$ whenever $\MH(x,y) > 0$.
If we put
\begin{align*}
\lambda(x,y) = \frac{1}{ \pi(y)\MH(y,x)}  h\left( \frac{\pi(y)\MH(y,x)}{\pi(x)\MH(x,y)}\right),
\end{align*}
it is easy to check that the conditions are met. 
\end{proof}

\subsection{Proof of Theorem~\ref{thm:mixing}}\label{subsec:pf_thm1}

In this section, we prove our main result by using Proposition~\ref{prop:path}. The main step of the proof is identifying the path ensemble $\Delta^*$ that makes the mixing time bound tight. To this end, we need to choose exactly one path $\delta^*(x,y)$ for each tuple $(x,y) \in \cX \times \cX$. (Note that $x$ and $y$ cannot be identical by the definition of path). From the intuition described in Section~\ref{subsec:mixing}, an edge $(u,v)$ has a large capacity if $\pi(u)$ and $\pi(v)$ are large. For example, if an edge contains the highest posterior state $x^*$, we can let the edge be traversed by a large number of paths. Given an edge with a small capacity, however, we need to ensure that the edge overlies with a small number of paths.
Importantly, we do not let a path $\delta^*(x,y)$ pass through an edge with a small capacity if both $\pi(x)$ and $\pi(y)$ are large, so that the edge can maintain a small unit flow size.
We may envision the topography of the path ensemble; $x^*$ becomes the hub, while states with low posterior probability are located on the outskirt. 

We construct the path ensemble $\Delta^*$ according to the description above. The construction of $\Delta^*$ is similar to that of~\cite{yang2016computational} and~\cite{zhou2021complexity}. With a neighborhood relation $\cN$ that satisfies the conditions in Theorem~\ref{thm:mixing},  $g \colon \cX \rightarrow \cX$\footnotetext[1]{If multiple states tie, we randomly pick one of them.}, 
\begin{align}\label{eq:g}
    g(x)= 
    \begin{cases}
    \arg \max_{x' \in \cN(x)} \pi(x') ^1
     & \text{ if } x \neq x^*, \\
    x^* & \text{ otherwise.}
    \end{cases}
\end{align}
By Condition (i) in Theorem~\ref{thm:mixing}, there exists $m \in \bbN$ such that $g^m(x) = (\overbrace{g\circ\dots\circ g}^{m\text{ times}})(x)=x^*$ for any $x \in \cX$ and $x^*$ is the only fixed point of $g$. ($x^*$ can be thought as an attractor in dynamic systems.) For all $x,y \in \cX$ with $x \neq y$, we have three cases; (i) $g^m(x) = y$ for some $m \in \bbN \setminus \{0\}$, (ii) $g^m(y) = x$ for some $m \in \bbN \setminus \{0\}$, or (iii) neither (i) nor (ii). If $(x,y) \in \cX \times \cX$ belongs to (i), we define $\delta^*(x,y) = (x, g(x), \dots, g^m(x) = y)$. Similarly, if $(x,y) \in \cX \times \cX$ belongs to (ii), let $\delta^*(x,y) = (x = g^m(y), \dots, g(y) , y)$. For the case (iii), if $m_1, m_2 \in \bbN \setminus \{0\}$ are the minimum numbers that satisfy $g^{m_1}(x) = x^*, g^{m_2}(y) = x^*$, respectively, we let $\delta^*(x,y) = (x, g(x), \dots, g^{m_1}(x) = x^* = g^{m_2}(y), \dots, g(y), y)$. This yields the path ensemble $\Delta^*$. 
We provide a toy example on how to construct $g$ and any path $\delta^*(x,y)$ for each tuple $(x,y) \in \cX \times \cX$ associated with $g$ in Appendix~\ref{subsec:ex_construction_path}.

Next, we make a bound for the congestion parameter $\rho(\Delta^*)$. We let $\Lambda(u) = \{x \in \cX: u = g^{k}(x), k \in \bbN\}$ denote the ancestor set of $u$ with respect to $g$. If $(u,v) \in \delta^*(x,y)$ for some $x,y \in \cX$, we can easily verify that $x \in \Lambda(u)$ by the construction of $\delta^*$. This implies $\{(x,y) \in \cX \times \cX: (u, v) \in \delta^*(x,y) \} \subseteq \Lambda(u) \times \cX$. It follows that
\begin{align*}
    \rho(\Delta^*) & \leq \max_{(u, v) : v =g(u)} \frac{1}{\pi(u) \bP(u, v)} \sum_{(x,y) \in \Lambda(u) \times \cX} \pi(x)\pi(y) \\
    & = \max_{(u, v) : v = g(u)} \frac{1}{\pi(u) \bP(u, v)} \left(\sum_{x \in \Lambda(u)} \pi(x) \right) \left( \sum_{y \in \cX} \pi(y) \right) \\
    & = \max_{(u, v) : v = g(u)} \frac{\pi(\Lambda(u))}{\pi(u) \bP(u, v)} \\
    & \leq \max_{(u, v) : v = g(u)} \frac{1}{(1 - p^{-(t_2 - t_4)}) \bP(u,v)}.
\end{align*}
The last inequality holds by the following. Let $g^{-k}(u) = \{x \in \cX \colon g^{k}(x) = u, g^{k-1}(x) \neq u\}$ for $k \in \bbN$, then $\Lambda(u) = \uplus_{k \in \bbN} g^{-k}(u)$. By Condition (ii) in Theorem~\ref{thm:mixing}, we have $|g^{-k}(u)| \leq p^{kt_4}$ which yields,
\begin{align*}
    \frac{\pi(\Lambda(u))}{\pi(u)} = \sum_{k \in \bbN} \frac{\pi(g^{-k}(u))}{\pi(u)} \leq \sum_{k \in \bbN} p^{-k(t_2 -t_4)} =  \frac{1}{1 - p^{-(t_2 - t_4)}},
\end{align*}
where we use Condition (i) and the definition of $g$, and $t_2 > t_4$ in Theorem~\ref{thm:mixing}.

Finally, we show that $\bP(x, g(x)) \geq C' \frac{N}{p^{t_4}}$ for $x \neq x^*$ and for some universal constant $C' > 0 $. Recall that (see also \cite[Theorem 1]{liu2000multiple})
\begin{gather*}
    \bP(x,g(x)) = N \sum_{y_1, \dots, y_{N-1}} \sum_{x^\star_1, \dots, x^\star_{N-1}} \frac{w(g(x) \given x)}{w(g(x) \given x)+ \sum_{j=1}^{N-1} w(y_j \given x)} \min\left\{1, \frac{w(g(x) \given x) + \sum_{j=1}^{N-1} w(y_j \given x) }{w(x \given g(x))+ \sum_{j=1}^{N-1} w(x^\star_j \given g(x))} \right\}\times\\
    \quad \quad  \MH(x, g(x)) \MH(x,y_1) \cdots \MH(x,y_{N-1}) \MH(g(x),x^\star_1) \cdots \MH(g(x),x^\star_{N-1}) 
    \\
    \quad \quad \quad \quad \quad \geq \left( N \sum_{y_1, \dots, y_{N-1}} \frac{w(g(x) \given x)}{w(g(x)\given x)+ \sum_{j=1}^{N-1}w(y_j \given x)} \MH(x, g(x)) \MH(x,y_1) \cdots \MH(x,y_{N-1}) \right) \times  \\
    \quad \quad \quad \quad \quad \quad \left( \sum_{x^\star_1, \dots, x^\star_{N-1}} \min\left\{1, \frac{w(g(x)\given x) }{w(x \given g(x))+ \sum_{j=1}^{N-1} w(x^\star_j \given g(x))} \right\} \MH(g(x),x^\star_1) \cdots \MH(g(x),x^\star_{N-1}) \right), 
\end{gather*}
and we denote the first and the second terms of the right-hand side by
\begin{align}
    \bK(x,g(x)) &= N \sum_{y_1, \dots, y_{N-1}} \frac{w(g(x) \given x)}{w(g(x) \given x)+ \sum_{j=1}^{N-1}w(y_j \given x)} \MH(x, g(x)) \MH(x,y_1) \cdots \MH(x,y_{N-1}), \label{eq:proposal} \\
    \eta (x,g(x)) &= \sum_{x^\star_1, \dots, x^\star_{N-1}} \min\left\{1, \frac{w(g(x) \given x) }{w(x \given g(x))+ \sum_{j=1}^{N-1} w(x^\star_j \given g(x))} \right\} \MH(g(x),x^\star_1) \cdots \MH(g(x),x^\star_{N-1}). \label{eq:acceptance}
\end{align}
Hence, we have $\bP(x,g(x))\ge \bK(x,g(x)) \eta(x,g(x))$.
We remark that the formulation of $\bK(x,g(x))$ in \eqref{eq:proposal} is based on the exchangeability, where the $N$-th trial state $y_N=g(x)$ is assumed to be selected as a proposal state and its probability is multiplied by $N$. However, we will not utilize such exchangeability to calculate a lower bound of $\bK(x,g(x))$. Instead, we define an event $A$ that selects $g(x)$ as the proposal from Step 1 and Step 2 in Algorithm~\ref{alg:MTM}, so that $\bbP(A) = \bK(x,g(x))$. We aim to lower bound the $\mathbb{P}(A)$ by using the law of total probability. To this end, we introduce the event $F$ for Step 1 in Algorithm~\ref{alg:MTM} that we include the state $g(x)$ at least once among the $N$ trials while we don't sample any ``high'' posterior states of the neighborhood of $x$ for the rest of trials, i.e. they do not belong to the set $\mathcal{S}(x)$. Using conditional probability rule, the probability of the event $F$ is equal to
\begin{align}
    \bbP(F) =  \left(\frac{|\cN(x)| - |\mathcal{S}(x)| + 1}{|\cN(x)|} \right)^N \left(1 - \left(\frac{|\cN(x)| - |\mathcal{S}(x)| }{|\cN(x)| - |\mathcal{S}(x)| + 1} \right)^N\right),
\end{align}
where we take the probability over the uniform samples in Step 1 in Algorithm~\ref{alg:MTM}. 
Using the inequality 
$a^n - b^n = (a - b)(a^{n-1} + a^{n-2} b + \cdots + b^{n-1}) \geq (a - b) n b^{n-1}$ for any $a \geq b \geq 0$, 
we find that the lower bound of $\bbP(F)$ can be obtained by
\begin{align*}
    \bbP(F) &= \left(\left(1 - \frac{ |\mathcal{S}(x)| - 1}{|\cN(x)|} \right)^N - \left(1 - \frac{ |\mathcal{S}(x)| }{|\cN(x)|} \right)^N  \right) \\
    &\geq \frac{N}{|\cN(x)|}\left(1-\frac{|\mathcal{S}(x)|}{|\cN(x)|}\right)^{N-1} \\
    &\geq \frac{N}{p^{t_4}}\left(1-(N-1)\frac{s_0}{p^{t_3}}\right) \\
    & \geq \frac{N}{p^{t_4}}(1+o(1)),
\end{align*}
where Condition (ii) and Bernoulli's inequality are used in the second inequality, and Condition (iii) is used in the last inequality. 
We further define $F_k$ as the event $F$ with $k$ number of $g(x)$ among $N$ trials. (Note that since we sample trials with replacement in Step 1 of Algorithm~\ref{alg:MTM}, we may sample $g(x)$ multiple times.) Observe that $F = \uplus_{k =1 }^N F_k$. Given the event $F_k$, the probability to select $g(x)$ in (2) of Algorithm~\ref{alg:MTM} is upper bounded by
\begin{align*}\label{eq:lower_bound_proposal}
     \bbP(A \mid F_k) &= \frac{k w(g(x) \given x)}{k w(g(x) \given x) + \sum_{i = 1}^{N-k} w(y_j \given x)}  \\
     & = \left\{1 + k^{-1}\sum_{i = 1}^{N-k} \frac{w(y_j \given x)}{w(g(x) \given x)} \right \}^{-1} \\
     & \stackrel{(\star)}{\geq}  \left\{1 + N \frac{h(p^{t_1+t_4-t_3})}{h(p^{t_2+t_3-t_4})} \right \}^{-1}  = 1 + o(1).
\end{align*}
To see $(\star)$, we have used the fact that $h$ is a non-decreasing function, Condition (iii), and we have 
\begin{gather*}
        \frac{\pi(y_j)}{\pi(x)} \cdot \frac{\MH(y_j,x)}{\MH(x,y_j)} \leq p^{t_1}  p^{t_4-t_3}, \\ 
    \frac{\pi(g(x))}{\pi(x)} \cdot \frac{\MH(g(x),x)}{\MH(x,g(x))} \geq p^{t_2} p^{t_3-t_4}, 
\end{gather*}
from Condition (i) and (ii).
Note that the right-hand side of the inequality $(\star)$ does not depend on $k$. Using the law of total probability, we yield the lower bound of $\bbP(A) = \bK(x,g(x))$ by combining the previous results:
\begin{align}\label{eq:K}
     \bK(x,g(x)) &= \bbP(A|F) \bbP(F) + \bbP(A|F^c) \bbP(F^c) \geq \bbP(A|F) \bbP(F) \nonumber\\
         & = \sum_{k=1}^N \bbP(A|F_k) \bbP(F_k) \geq \sum_{k=1}^N (1+o(1)) \bbP(F_k) \nonumber\\
         & = (1+o(1)) \bbP(F) \geq \frac{N}{p^{t_4}} (1 + o(1)),
\end{align}

Similarly, we can calculate the lower bound of $\eta(x,g(x))$. We consider the event $G$ that we don't select any ``high'' posterior states of the neighborhood of $g(x)$ for $N-1$ trials, that is, any of them are not in the set $\mathcal{S}(g(x))$. A simple calculation yields 
\begin{align*}
    \bbP (G) \geq \left( \frac{|\cN(g(x))| - s_0}{|\cN(g(x))|}\right)^{N-1}.
\end{align*}
Under the event $G$, on the other hand,
\begin{align*}
     \frac{w(g(x) \given x) }{w(x \given g(x))+ \sum_{j=1}^{N-1} w(x^\star_j \given g(x))} 
    = & \left(\frac{w(x \given g(x))}{w(g(x) \given x)} + \sum_{j=1}^{N-1}  \frac{w(x^\star_j\given g(x))}{w(g(x) \given x)} \right)^{-1}  \\
    =& \left( \frac{\pi(x)\MH(x,g(x))}{\pi(g(x))\MH(g(x),x)} + \sum_{j=1}^{N-1}  \frac{w(x^\star_j \given g(x))}{w(g(x) \given x)} \right)^{-1}  \\
    \geq &  \left( p^{-t_2 -t_3 + t_4} + (N-1) \frac{h(p^{t_1 -t_3 +t_4})}{h(p^{t_2 +t_3 -t_4})} \right)^{-1} ,
\end{align*}
where the second equality is due to the property of the balancing function $h(u) = u h(1/u)$ and the last inequality follows from a similar argument as before using Conditions (i), (ii) and
\begin{align*}
    \frac{\pi(x^\star_j)}{\pi(g(x))} \cdot \frac{\MH(x^\star_j,g(x))}{\MH(g(x),x^\star_j)} \leq p^{t_1}  p^{t_4-t_3}, 
\end{align*}
by non-decreasing $h$.
Then for $x \neq x^*$, by Conditions (ii), (iii) and  we have
\begin{align}\label{eq:eta}
    \eta(x, g(x))  &\geq   \min \left\{1, \left( p^{-t_2 -t_3 + t_4} + (N-1) \frac{h(p^{t_1 -t_3 +t_4})}{h(p^{t_2 +t_3 -t_4})}  \right)^{-1}\right\} \left( \frac{|\cN(g(x))| - s_0}{|\cN(g(x))|}\right)^{N-1}
    \nonumber\\
    & \geq (1+o(1)) (1 - s_0/p^{t_3})^{N-1}\nonumber\\ 
    & \geq 1 -\frac{N s_0}{p^{t_3} } + o(1)\nonumber\\
    &= 1 + o(1).
\end{align}
Combining the lower bounds~\eqref{eq:K} and~\eqref{eq:eta} leads to $\mathbf{P}(x,g(x))\ge C'\frac{N}{p^{t_4}}$, which concludes the proof of the theorem.

\subsection{An example on a different weight function.}\label{subsec:ex}

The example below shows an undesirable behavior of the MTM algorithm if we use a weight function which is not in the class of~\eqref{eq:weight}. Here an undesirable behavior means that the acceptance probability is close to zero even when $\pi(x_{\mathrm{prop}}) \gg \pi(x_{\mathrm{curr}})$, where we denote $x_{\mathrm{curr}},x_{\mathrm{prop}}$ as the current state and the proposed state from (2) in Algorithm~\ref{alg:MTM}, respectively. 

\begin{example}
We consider the weight function $w(y \given x) = \pi(y)$. For the sake of simplicity, let $\vert \cN(x) \vert = p^{t_3}$ and for $x' \in \cN(x)$,
\begin{align*}
\frac{\pi(x')}{ \pi(x)} = 
    \begin{cases}
p^{t_2}, &\quad \quad \text{for }x' \in \mathcal{S}(x),\\
p^{t_1/2}, &\quad \quad \text{for }x' \not\in \mathcal{S}(x),
\end{cases}
\end{align*}
for all $x \neq x^*$. Let $x_j, x_j^{\star}$ be uniform samples from $\cN(x_{\mathrm{curr}})$ and $\cN(x_{\mathrm{prop}})$, respectively for $j \in [N-1]$. Assume $x_{\mathrm{prop}}\in\mathcal{S}(x_{\mathrm{curr}})$ to reflect $\pi(x_{\mathrm{prop}}) \gg \pi(x_{\mathrm{curr}})$. Although the conditions on Theorem~\ref{thm:mixing} are met (except for those related to a balancing function $h$), the acceptance probability from $x_{\mathrm{curr}}$ to $x_{\mathrm{prop}}$ is upper bounded as
\begin{align*}
    \alpha(x_{\mathrm{curr}}, x_{\mathrm{prop}}) &  = \min \left \{1,  \frac{ \sum_{i = 1}^{N-1} \pi(x_i) + \pi(x_{\mathrm{prop}})}{ \sum_{i = 1}^{N-1} \pi(x^{\star}_i) + \pi(x_{\mathrm{curr}})}\right\} \\ 
    & \leq \frac{N p^{t_2} \pi (x_{\mathrm{curr}})    }{ (1+ (N-1) p^{t_2 + t_1/2} )\pi (x_{\mathrm{curr}})}\\
    & \leq \frac{2}{p^{t_1/2}} = o(1).
\end{align*}
Notice that the acceptance probability $\alpha(x_{\mathrm{curr}}, x_{\mathrm{prop}}) = 1$ in the MH algorithm.  
\end{example}

\newpage 

\subsection{A toy example on path construction.}\label{subsec:ex_construction_path}

We provide a simple example on how to construct the canonical path ensemble $\Delta^*$ described in~\ref{subsec:pf_thm1}. Define $\cX = \{0,1\}^3$ as our state space and let $x^* = (1, 1, 0)$ be the mode of the distribution on $\cX$. 
Let the target distribution $\pi(x) \propto \exp (-d_{\mathrm{H}}(x, x^*))$, where $d_{\mathrm{H}}$ is a Hamming distance.
We specify a neighborhood of any $x \in \cX$ as $\cN(x) = \{y \in \cX: d_{\mathrm{H}}(y, x) = 1\}$, where $d_{\mathrm{H}}$ is a Hamming distance. In the left panel of Figure~\ref{fig:path}, neighboring states are linked to black undirected edges. We follow the rule described in~\eqref{eq:g} to define a transition function $g: \cX \rightarrow \cX$ and black directed edges indicate the defined moves by a function $g$. The height of each bar indicates $\pi(x)$ associated with the corresponding state $x$. In the right panel of Figure~\ref{fig:path}, we provide three examples to illustrate the three possible cases to construct a path from the defined transition $g$, described in~\ref{subsec:pf_thm1}. 
\begin{itemize}
    \item The red directed edges indicate a path from $(0, 0, 1) $ to $(1,0,0)$, which corresponds to the case (i), since $(1,0,0) = g^2((0,0,1))$,
    \item The green directed edge indicates a path from $(1,1 , 0) $ to $(0,1,0)$, which corresponds to the case (ii), since $(1,1,0) = g((0,1,0))$,
    \item The blue directed edges indicate a path from $(0,0,1)$ to $(1,1,1)$, which corresponds to the case (iii).
\end{itemize}

\begin{figure}[h]
    \centering
    \begin{subfigure}{0.49\textwidth}
    \includegraphics[width=\textwidth]{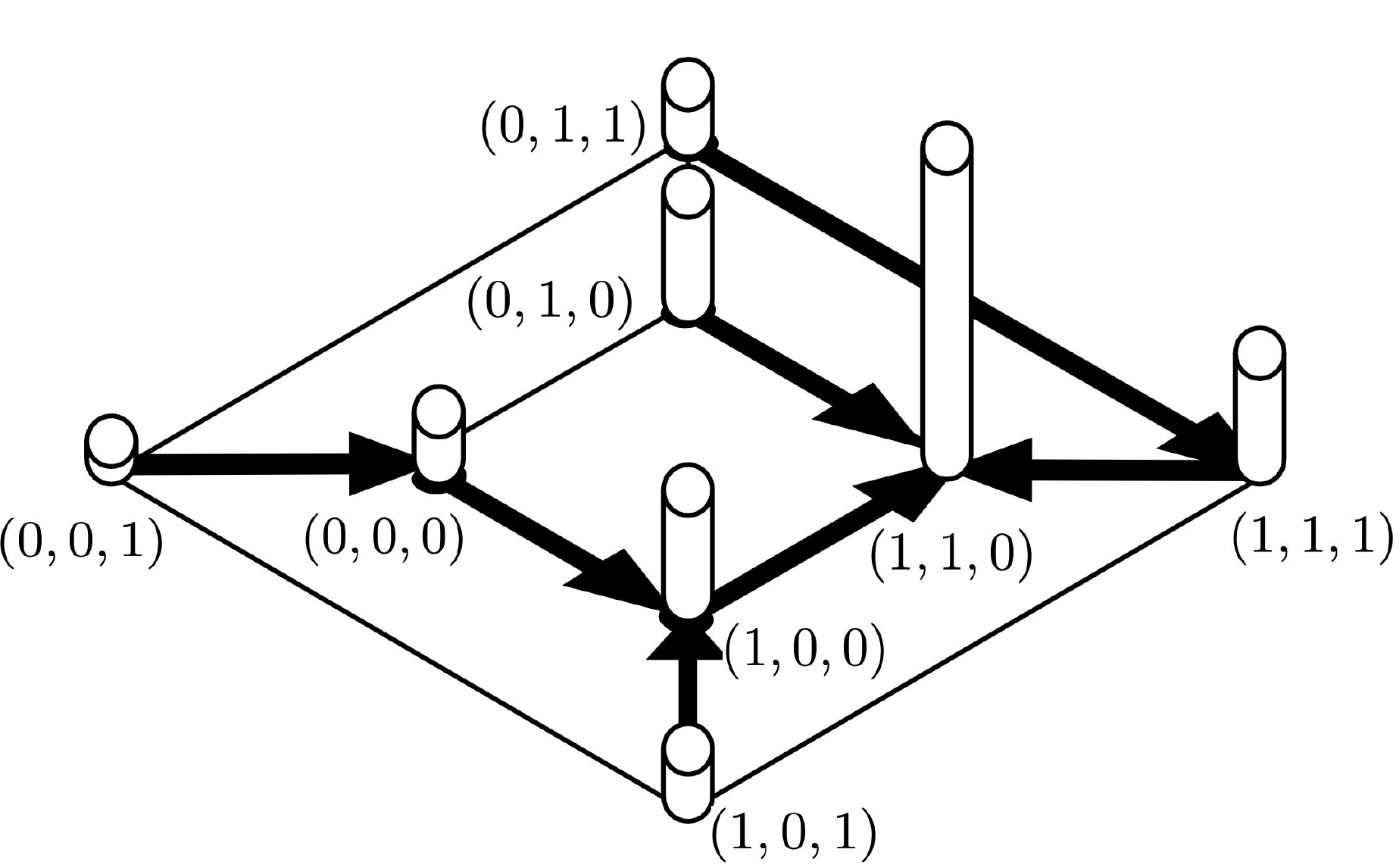}
    \end{subfigure}
    \hfill
    \begin{subfigure}{0.49\textwidth}
    \includegraphics[width=\textwidth]{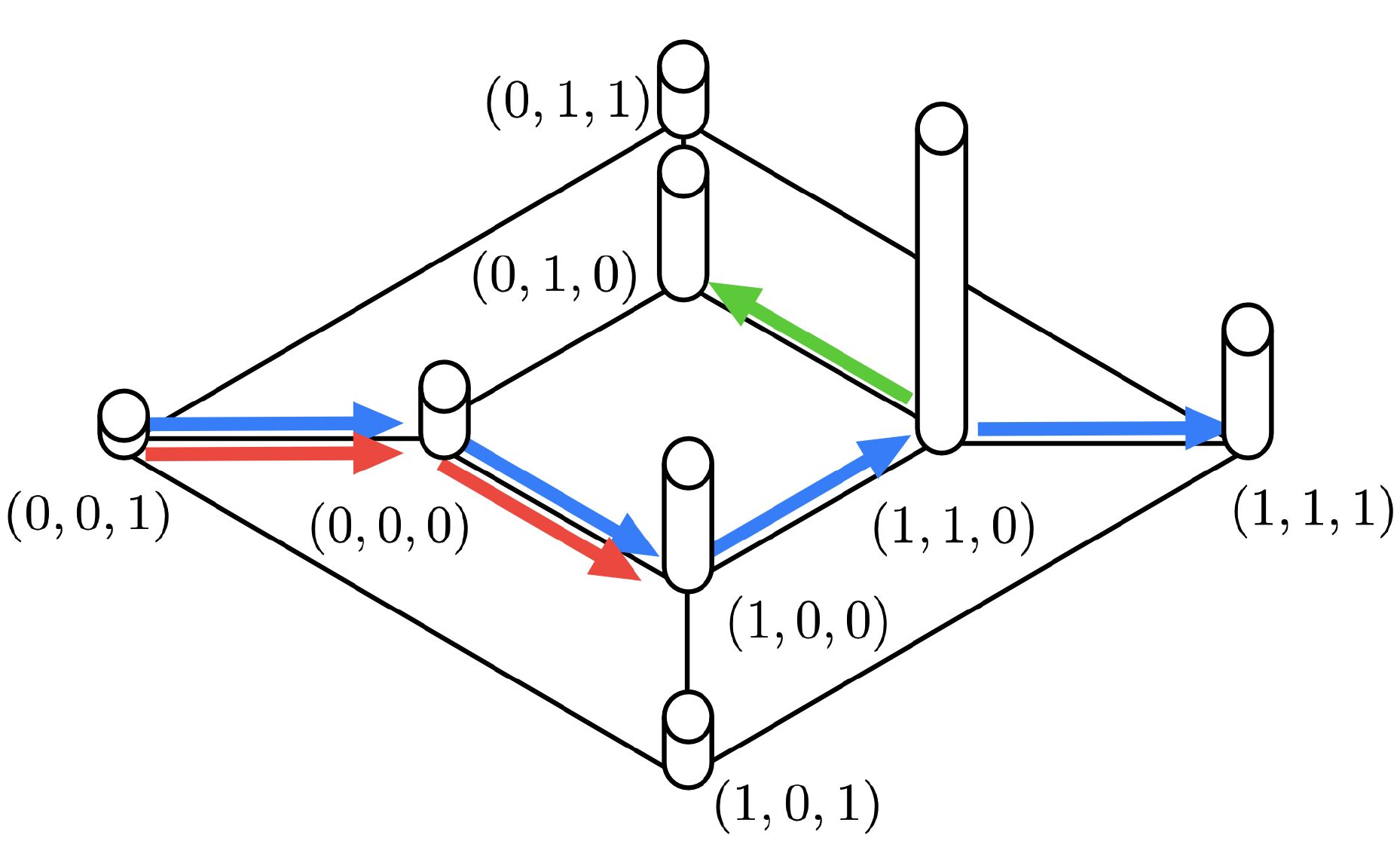}
    \end{subfigure}
    \caption{A toy example on path construction. Black undirected edges connect neighboring states and the target distribution $\pi$ is represented as the heights of the cylinders. (Left) Black directed edges indicate the defined moves by a function $g$. (Right) The colored paths exemplify the three possible cases of path construction.}
    \label{fig:path}
\end{figure}
\newpage

\section{Details of simulation studies}\label{sec:append.simulation}

The scope of this paper is to theoretically study the mixing time for the family of MTM algorithms, and hence we mainly focus on experiments to empirically verify our theoretical insights, that the MTM mixing time is smaller by a factor of the number of trials $N$ and that locally balanced weight functions tend to perform better under suitable assumptions. Nevertheless, in some experiments, we compare the MTM algorithm with the locally balanced MH algorithm (denoted as LBMH) \cite{zanella2020informed}, as it has been reported to outperform the other state-of-the-art methods. Before describing the details of simulation studies, here we briefly describe the locally balanced MH algorithm. 

Specification of LBMH requires balancing function $h$ and uninformed symmetric distribution $\mathbf{K}_{\mathrm{sym}}(x,\cdot)$ supported on $\mathcal{N}(x)$. LBMH chooses a proposal state $y$ from a pointwise informed proposal distribution
\begin{equation}
    Q_h(x,y) = (Z_h(x))^{-1}h\left(\pi(y)/\pi(x)\right) \mathbf{K}_{\mathrm{sym}}(x,y),
\end{equation}
where $Z_h(x) = \sum_{z\in \mathcal{N}(x)}h\left(\pi(z)/\pi(x)\right) \mathbf{K}_{\mathrm{sym}}(x,z)$ is a normalizing constant. Then, $y$ is accepted with probability $\alpha = \min\{1, \frac{\pi(y)Q_h(y,x)}{\pi(x)Q_h(x,y)}\} = \min\{1, \frac{Z_h(x)}{Z_h(y)}\}$, by defintion of balancing function $h(u) = uh(1/u)$ and symmetry of $\mathbf{K}_{\mathrm{sym}}$. Unlike MTM where a subset of $\mathcal{N}(x)$ is selected (with replacement) as a trial and choose a proposal among them, LBMH needs to evaluate $h(\pi(y)/\pi(x))$ for \textit{all} $y\in\mathcal{N}(x)$ to get a proposal state $y$ which can be viewed as an exhaustive search of $\mathcal{N}(x)$. In terms of computation, MTM requires calculating $2N-1$ weight functions at each iteration where $N$ can be chosen at one's disposal, LBMH requires calculating $|\mathcal{N}(y)|$ number of ratios to calculate $Z_h(y)$ at each iteration, where proposal probabilities $\{Q_h(x,y):y\in\mathcal{N}(x)\}$ and normalizing constant $Z_h(x)$ can be saved and reused from the previous iteration. Since random walk proposals in BVS and SBM examples are both symmetric, we compare MTM with LBMH by letting $\mathbf{K}_{\mathrm{sym}} = \MH$ with three different balancing functions: $h(u) = \sqrt{u}$, $h(u) = \min\{1,u\}$ and $h(u) = \max\{1,u\}$ (corresponding to $w_{\mathrm{sqrt}}, w_{\mathrm{min}}, w_{\mathrm{max}}$ respectively). 

\subsection{Details of Bayesian variable selection (BVS)}\label{sec:append.BVS}

After marginalizing out $\bm\beta$ and $\phi$, the posterior distribution $\pi(\gamma\given \bm{y})$ is written as \cite[][\S A.1]{yang2016computational}
\begin{equation}
\label{eq:bvsposterior}
    \pi(\gamma\given \bm{y}) = C\cdot \frac{1}{p^{\kappa|\gamma|}(1+\mathscr{G})^{|\gamma|/2}}\mathrm{SSR}(\gamma)^{-n/2}\ind(|\gamma|\le s_{\mathrm{max}}),
\end{equation}
where $\mathrm{SSR}(\gamma)=\bm{y}^{\top}\left(\mathbf{I}_{n}- \frac{\mathscr{G}}{\mathscr{G}+1}\bm{X}_{\gamma}\left(\bm{X}_{\gamma}^{\top} \bm{X}_{\gamma}\right)^{-1} \mathbf{X}_{\gamma}^{\top}\right) \bm{y}$ is a term having a similar role as a sum of squared residuals and $C$ is a normalizing constant.

\textbf{MCMC setup.} Hyperparameters are specified as $\mathscr{G} = p^3 = 5000^3$, $\kappa = 2$, and $s_{\mathrm{max}} = 100$. For each dataset, we run a chain of $10^5$ iteration for single-try MH, $2\times 10^4$ iteration for MTM with $N=5$, and $10^4$ iteration for MTM with $N=5,10,50,100,500,1000,2000,5000$ using four different weight functions. Algorithms are randomly initialized with state $\gamma_0$ such that $\gamma_0 \cap \gamma^* =\emptyset$ and $d_{\mathrm{H}}(\gamma_0, \gamma^*) = 20$ which implies $H = 20$ is the minimum required hitting iteration. 
For each simulated dataset, the true data generated model achieves the highest posterior probability ($\gamma^*=x^*$). All simulation studies are performed on a Linux cluster with Intel(R) Xeon(R) Gold 6132 CPU @ 2.60GHz and 96GB memory.

\begin{table}
    \small
  \caption{(BVS) Median of $H$, the number of iterations until the chain hit $\gamma^*$ over 50 replicates.  Entry with ``Fail'' indicates that chains never hit $\gamma^*$ in more than half of the replicated datasets.}
  \label{table:bvsH}
  \centering
  \begin{tabular}{ccc c c c c c c c c c c}
    \toprule
    & SNR &$N$ & 1 & 5 & 10 & 50 & 100 & 500 & 1000 & 2000 & 5000 & LBMH \\
    \midrule
    \multirow{8}{*}{ind.} & \multirow{4}{*}{$4$} & $w_{\mathrm{ord}}$ & \multirow{4}{*}{19414} & 3283 & 1742 & 392 & 203 & 100 & 211 & 3168 & Fail & N/A\\
     & & $w_{\mathrm{sqrt}}$ & & 3340 &1787 & 360 & 177 &  55 &  42 &  38 &  54 & Fail\\
     & & $w_{\mathrm{min}}$  & & 3365 &1948 & 354 & 180 &  50 &  33 &  28 &  24 & 20\\
     & & $w_{\mathrm{max}}$  & & 3246 & 1876 & 372 & 182 & 54 & 42 & 40 & 60 & Fail\\
     \cmidrule(lr){2-13}
     & \multirow{4}{*}{$2$} & $w_{\mathrm{ord}}$ & \multirow{4}{*}{20088} & 3684 & 1865 & 392 & 213 &  89 & 137 & 1020 & Fail & N/A\\
     & & $w_{\mathrm{sqrt}}$ & & 3666 & 1955 & 398 & 200 &  58 &  40 &  33 &  32 & 137 \\
     & & $w_{\mathrm{min}}$  & & 3928 & 2034 & 366 & 202 & 55 & 34 & 29 & 24 & 20\\
     & & $w_{\mathrm{max}}$  & & 3696 & 2000 & 418 & 229 & 62 & 44 & 36 & 34 & Fail\\
     \midrule
    \multirow{8}{*}{dep.} & \multirow{4}{*}{$4$} & $w_{\mathrm{ord}}$ & \multirow{4}{*}{21292} & 3898 & 1989 & 422 & 234 & 91 & 117 & 735 & Fail & N/A\\
     & & $w_{\mathrm{sqrt}}$ & & 3977 & 2256 & 394 & 209 & 64 & 44 & 36 & 45 & Fail\\
     & & $w_{\mathrm{min}}$  & & 4196 & 2065 & 412 & 226 & 51 & 35 & 30 & 24 & 20\\
     & & $w_{\mathrm{max}}$  & & 4360 & 2137 & 504 & 240 & 70 & 46 & 42 & 40 & Fail\\
     \cmidrule(lr){2-13}
    & \multirow{4}{*}{$2$} & $w_{\mathrm{ord}}$ & \multirow{4}{*}{66020} & 7724 & 6324 & 1033 & 528 & 150 & 145 & 404 & Fail & N/A\\
     & & $w_{\mathrm{sqrt}}$ & & 9458 & 4226 & 1088 & 660 & 180 & 97 & 59 & 72 & Fail\\
     & & $w_{\mathrm{min}}$  & & 8357 & 4363 & 1212 & 484 & 109 & 68 & 54 & 31 & 37\\
     & & $w_{\mathrm{max}}$  & & 11541 & 7057 & 6794 & 6782 & 3668 & 3124 & 3729 & 7246 & Fail\\
    \bottomrule
  \end{tabular}
\end{table}

\begin{table}
    \small
  \caption{(BVS) Median of $T_H$, wall-clock time (in seconds) until the chain hit $\gamma^*$ over 50 replicates.  Entry with ``Fail'' indicates that chains never hit $\gamma^*$ in more than half of the replicated datasets.}
  \label{table:bvsTH}
  \centering
  \begin{tabular}{ccc c c c c c c c c c c}
    \toprule
   & SNR &$N$ & 1 & 5 & 10 & 50 & 100 & 500 & 1000 & 2000 & 5000 & LBMH\\
    \midrule
    \multirow{8}{*}{ind.} & \multirow{4}{*}{$4$} & $w_{\mathrm{ord}}$ & \multirow{4}{*}{1.30} & 0.80 & 0.42 & 0.12 & 0.07 & 0.07 & 0.27 & 6.95 & Fail & N/A\\
     & & $w_{\mathrm{sqrt}}$ & & 0.81 & 0.46 & 0.11 & 0.07 & 0.04 & 0.05 & 0.09 & 0.30 & Fail\\
     & & $w_{\mathrm{min}}$  & & 0.89 & 0.53 & 0.12 & 0.07 & 0.04 & 0.04 & 0.06 & 0.11 & 0.07\\
     & & $w_{\mathrm{max}}$  & & 0.88 & 0.53 & 0.13 & 0.07 & 0.04 & 0.05 & 0.09 & 0.33 & Fail\\
     \cmidrule(lr){2-13}
    & \multirow{4}{*}{$2$} & $w_{\mathrm{ord}}$ & \multirow{4}{*}{1.25} & 0.82 & 0.43 & 0.11 & 0.07 & 0.06 & 0.15 & 1.93 & Fail & N/A\\
     & & $w_{\mathrm{sqrt}}$ & & 0.81 & 0.44 & 0.12 & 0.07 & 0.04 & 0.04 & 0.06 & 0.13 & 0.39\\
     & & $w_{\mathrm{min}}$  & & 1.01 & 0.55 & 0.11 & 0.08 & 0.04 & 0.04 & 0.06 & 0.11 & 0.07\\
     & & $w_{\mathrm{max}}$  & &0.92 & 0.51 & 0.13 & 0.08 & 0.05 & 0.05 & 0.06 & 0.13& Fail\\
     \midrule
    \multirow{8}{*}{dep.} & \multirow{4}{*}{$4$} & $w_{\mathrm{ord}}$ & \multirow{4}{*}{1.35} & 0.85 & 0.42 & 0.12 & 0.08 & 0.06 & 0.12 & 1.37 & Fail & N/A\\
     & & $w_{\mathrm{sqrt}}$ & & 0.87 & 0.50 & 0.11 & 0.07 & 0.04 & 0.05 & 0.06 & 0.17 & Fail\\
     & & $w_{\mathrm{min}}$  & & 1.04 & 0.55 & 0.13 & 0.08 & 0.04 & 0.04 & 0.06 & 0.10 & 0.06\\
     & & $w_{\mathrm{max}}$  & & 1.07 & 0.53 & 0.15 & 0.08 & 0.05 & 0.05 & 0.07 & 0.14 & Fail\\
     \cmidrule(lr){2-13}
    & \multirow{4}{*}{$2$} & $w_{\mathrm{ord}}$ & \multirow{4}{*}{3.38} & 1.67 & 1.29 & 0.26 & 0.16 & 0.09 & 0.14 & 0.70 & Fail & N/A\\
     & & $w_{\mathrm{sqrt}}$ & & 1.91 & 0.90 & 0.28 & 0.19 & 0.10 & 0.10 & 0.10 & 0.24 & Fail\\
     & & $w_{\mathrm{min}}$  & & 1.89 & 1.01 & 0.34 & 0.16 & 0.07 & 0.07 & 0.10 & 0.13 & 0.09\\
     & & $w_{\mathrm{max}}$  & & 2.58 & 1.66 & 1.86 & 1.97 & 2.10 & 2.78 & 5.74 & 22.36 & Fail\\
    \bottomrule
  \end{tabular}
\end{table}

Results from Table~\ref{table:bvsH} show that $H$ decreases roughly by a factor of $N$ until $N=100$, which confirms our theoretical findings, given that the model setting satisfies that the mixing time is equivalent to the hitting iteration up to constant factors~\cite{peres2015mixing}. When $N$ becomes larger, the performance of unscaled weight function $w_{\mathrm{ord}}$ deteriorates and never converges when $N=5000$. In contrast, locally balanced weight functions generally perform well even when $N$ is large. Table~\ref{table:bvsTH} suggests that choosing moderate $N$ is beneficial in terms of computational savings.
When the design matrix is correlated and SNR$=2$, the result suggests that the chain often stuck when we choose the weight function as $w_{\mathrm{max}}$. 
Since the shape of the posterior distribution becomes irregular when the SNR is intermediate \citep{yang2016computational} and design matrix is correlated, to get a more clear insight we further perform additional simulation study when the posterior distribution exhibits multimodality; see Appendix~\ref{sec:append.multimodal}.
Finally, under different settings of the design and SNR, the median $N$ estimated from Algorithm~\ref{alg:N} using $\psi = 0.9$ over 50 replicate datasets is $\hat{N}=349$ (indep, SNR=4), $ 501$ (indep, SNR=2), $328$ (dep, SNR=4) and $ 158$ (dep, SNR=2). 

In contrast to MTM, LBMH fails to converge to $\gamma^*$ when $h(u) = \sqrt{u}$ or $h(u) = \max\{1,u\}$. It is easier for MTM to escape from such local modes by randomly searching part of its neighborhood to select the proposal. The exhaustive search nature of LBMH makes it difficult to escape from the local mode since some high values of $\pi(y^*)/\pi(y)$, $y^*\in\mathcal{N}(y)$ involved in the denominator makes the acceptance ratio small. This phenomenon disappears when $h(u) = \min\{1,u\}$ is used. In terms of wall-clock hitting time $T_H$, LBMH is not as efficient as MTM with a smaller choice of $N$. 

In addition, we also consider the case when SNR $=0.5$ (very weak SNR) so that the null model $\gamma^* = \bm{0}$ receives the highest probability across all simulated datasets. For each replicated dataset, algorithms are randomly initialized with state $\gamma_0$ such that $d_\mathrm{H}(\gamma_0,\gamma^*) =10$ which implies $H=10$ is the minimum required hitting iteration. Table~\ref{table:bvsnull} provides the result similar to Tables \ref{table:bvsH} and \ref{table:bvsTH}, since the posterior distribution is unimodal with the peak at the null model $\gamma^* = \bm{0}$ due to the sparsity prior.
The median $N$ estimated from Algorithm~\ref{alg:N} using $\psi = 0.9$ over 50 replicate datasets is $\hat{N}=171$for independent design and is $\hat{N} = 212$ for dependent design.

\begin{table}[h]
    \small
  \caption{(BVS, very weak SNR $=0.5$) Median of $H$ and $T_H$ over 50 replicates.  Entry with ``Fail'' indicates that chains never hit $\gamma^*$ in more than half of the replicated datasets.}
  \label{table:bvsnull}
  \centering
  \begin{tabular}{ccc c c c c c c c c c}
    \toprule
    & SNR $=0.5$ &$N$ & 1 & 5 & 10 & 50 & 100 & 500 & 1000 & 2000 & 5000 \\
    \midrule
    \multirow{8}{*}{$H$} & \multirow{4}{*}{\shortstack{indep.}} & $w_{\mathrm{ord}}$ & \multirow{4}{*}{14358} & 2708 & 1309 & 281 & 142 & 50 & 46 & 106 & 8968\\
     & & $w_{\mathrm{sqrt}}$ & & 2788 & 1402 & 285 & 148 & 37 & 20 & 16 & 12 \\
     & & $w_{\mathrm{min}}$  & & 2680 & 1481 & 282 & 149 & 34 & 19 & 16 & 12 \\
     & & $w_{\mathrm{max}}$  & & 2476 & 1276 & 266 & 131 & 38 & 23 & 17 & 12\\
     \cmidrule(lr){2-12}
    & \multirow{4}{*}{dep.} & $w_{\mathrm{ord}}$ & \multirow{4}{*}{12596} & 2924 & 1432 & 304 & 142 & 43 & 42 & 104 & Fail\\
     & & $w_{\mathrm{sqrt}}$ & &2776 & 1532 & 278 & 146 & 35 & 21 & 15 & 12\\
     & & $w_{\mathrm{min}}$  & & 2836 & 1354 & 260 & 154 & 34 & 22 & 14 & 11\\
     & & $w_{\mathrm{max}}$  & & 2562 & 1289 & 270 & 140 & 32 & 22 & 15 & 12\\
     \midrule
    \multirow{8}{*}{$T_H$} & \multirow{4}{*}{indep.} & $w_{\mathrm{ord}}$ & \multirow{4}{*}{0.89} & 0.56 & 0.3 & 0.07 & 0.04 & 0.03 & 0.05 & 0.19 & 39.38\\
     & & $w_{\mathrm{sqrt}}$ & & 0.6 & 0.32 & 0.07 & 0.05 & 0.02 & 0.02 & 0.03 & 0.04\\
     & & $w_{\mathrm{min}}$  & & 0.68 & 0.36 & 0.09 & 0.05 & 0.02 & 0.02 & 0.03 & 0.04\\
     & & $w_{\mathrm{max}}$  & & 0.62 & 0.35 & 0.08 & 0.05 & 0.02 & 0.02 & 0.03 & 0.04 \\
     \cmidrule(lr){2-12}
    & \multirow{4}{*}{dep.} & $w_{\mathrm{ord}}$ & \multirow{4}{*}{0.67} & 0.57 & 0.29 & 0.07 & 0.04 & 0.03 & 0.04 & 0.17 & Fail\\
     & & $w_{\mathrm{sqrt}}$ & & 0.56 & 0.31 & 0.07 & 0.04 & 0.02 & 0.02 & 0.02 & 0.04\\
     & & $w_{\mathrm{min}}$  & & 0.66 & 0.33 & 0.07 & 0.05 & 0.02 & 0.02 & 0.02 & 0.04\\
     & & $w_{\mathrm{max}}$  & & 0.58 & 0.31 & 0.08 & 0.04 & 0.02 & 0.02 & 0.02 & 0.04\\
    \bottomrule
  \end{tabular}
\end{table}

\clearpage 

\subsection{Details of stochastic block model (SBM)}\label{sec:append.SBM}

After marginalizing out $\{Q_{uv}\}_{1\le u\le v\le K}$, the posterior distribution $\pi(\mathbf{z}\given \bm{A})$ is written as (see \cite[][\S 2.1]{legramanti2022extended} and \cite[][\S 2.2]{zhuo2021mixing})

\begin{equation}
\label{eq:sbmposterior}
    \pi(\mathbf{z}\given \bm{A}) = C\cdot \prod_{1\le u \le v \le K} B(\kappa_1 + m_{uv}, \kappa_2 + \overline{m}_{uv})\cdot \ind(\mathbf{z}\in S_\alpha),
\end{equation}
where $B(\kappa_1,\kappa_2) = \Gamma(\kappa_1)\Gamma(\kappa_2)/\Gamma(\kappa_1+\kappa_2)$ is a beta function, $C$ is a normalizing constant,
\[
m_{uv} = 
\begin{cases} \sum_{i,j} A_{ij}\ind(z_i=u, z_j=v) & \text{ if }u < v,\\
\sum_{i<j} A_{ij}\ind(z_i=u, z_j=u) & \text{ if } u=v,
\end{cases}
\] is the number of edges between blocks $u$ and $v$, and using the notation $n_u(\bfz) = \sum_i\ind(z_i=u)$, 
\[
\overline{m}_{uv} = 
\begin{cases}
n_u(\bfz)n_v(\bfz) - m_{uv}  & \text{ if }u < v,\\
n_u(\bfz)(n_u(\bfz)-1)/2 -m_{uu}  & \text{ if }u = v,
\end{cases}
\]
is the number of non-edges between blocks $u$ and $v$. We note that $\pi(\mathbf{z} \given \bm{A})$ is invariant of a label permutation.

\textbf{Data generation.} When $K=2$, there are two true clusters (blocks) of nodes each with 500 nodes. When $K=5$, there are five true clusters of nodes each with 200 nodes. We generated a graph from the homogeneous SBM and where within- and cross-community edge connection probabilities are $a$ and $b$ respectively. Specifically, for $K=2$ we set $(a,b) = (0.222,0.01)$ and $(a,b) = (0.07,0.01)$ so that $\mathrm{CH}\approx 10$ and $\mathrm{CH}\approx 2$, and for $K=5$ we set $(a,b) = (0.473,0.01)$ and $(a,b) = (0.13,0.01)$ so that $\mathrm{CH}\approx 10$ and $\mathrm{CH}\approx 2$. For each setting, we simulate 50 datasets.

\textbf{MCMC setup.} Hyperparameters are specified as $\kappa_1 = \kappa_2 = 1$, and $\alpha = 1000$ so that the size of the feasible set $S_\alpha$ is maximized. For each dataset, we run a chain of $10^5$ iteration for single-try MH, $5\times 10^4$ iteration for MTM with $N=5$, and $2\times 10^4$ iteration for MTM with $N=5,10,50,100,500,1000,2000,5000$ using four different weight functions. Algorithms are randomly initialized with state $\mathbf{z}_0$ such that $\tilde{d}_\mathrm{H}(\mathbf{z}_0,\mathbf{z}^*)=400$ which implies $H=400$ is the minimum required hitting iteration.  For each simulated dataset, the true data generated model achieves the highest posterior probability ($\bfz^*=x^*$).  All simulations are performed on a Linux cluster with Intel(R) Xeon(R) Gold 6132 CPU @ 2.60GHz and 96GB memory.

Results from Table \ref{table:sbmH} show that $H$ decreases roughly by a factor of $N$ until $N=10$ for locally balanced weight functions, but not for unscaled weight function $w_{\mathrm{ord}}$. Even when $N\ge 50$, MTM with $w_{\mathrm{ord}}$ never converges to the highest probability model, highlighting the necessity of the use of locally balanced weight function for the general model selection problems. When $N$ is very large, the performance of locally balanced weight functions generally deteriorates, which matches with our theoretical findings regarding the rate condition on $N$. Table~\ref{table:sbmTH} also suggests that a moderate choice of $N$ (in SBM case, around 10) is beneficial in terms of computation savings.
Finally, under different settings of ($K$, CH), the median $N$ estimated from Algorithm~\ref{alg:N} using $\psi = 0.9$ over 50 replicate datasets is $\hat{N}=15$ for ($K$, CH) = (2, 2), $\hat{N} = 8$ for (2, 10), $\hat{N}=5$ for (5, 2) and $\hat{N} = 4$ for (5, 10).

\begin{table}
    \small
  \caption{(SBM) Median of $H$, the number of iterations until the chain hit $\bfz^*$ over 50 replicates.  Entry with ``Fail'' indicates that chains never hit $\bfz^*$ in more than half of the replicated datasets.}
  \label{table:sbmH}
  \centering
  \begin{tabular}{ccc c c c c c c c c c c}
    \toprule
    & CH & $N$ & 1 & 5 & 10 & 50 & 100 & 500 & 1000 & 2000 & 5000 & LBMH\\
    \midrule
    \multirow{8}{*}{$K=2$} & \multirow{4}{*}{$\approx 10$} & $w_{\mathrm{ord}}$ & \multirow{4}{*}{11572} & 2495 & 5542 & Fail & Fail & Fail & Fail & Fail & Fail & N/A\\
     & & $w_{\mathrm{sqrt}}$ & & 1603 & 1136 & 692 & 644 & 602 & 610 & 637 & 684 & 740 \\
     & & $w_{\mathrm{min}}$  & & 1558 & 974 & 544 & 493 & 444 & 434 & 431 & 424 & 418\\
     & & $w_{\mathrm{max}}$  & & 1657 & 1142 & 762 & 708 & 726 & 830 & 1224 & 2842& 6364\\
     \cmidrule(lr){2-13}
    & \multirow{4}{*}{$\approx 2$} & $w_{\mathrm{ord}}$ & \multirow{4}{*}{13343} & 2722 & 3484 & Fail & Fail & Fail & Fail & Fail & Fail & N/A\\
     & & $w_{\mathrm{sqrt}}$ & & 1948 & 1432 & 874 & 818 & 808 & 703 & 700 & 728 & 682\\
     & & $w_{\mathrm{min}}$  & & 2244 & 1400 & 944 & 911 & 820 & 882 & 812 & 890 & 866\\
     & & $w_{\mathrm{max}}$  & & 1916 & 1354 & 851 & 774 & 701 & 696 & 681 & 680 & 709\\
     \midrule
    \multirow{8}{*}{$K=5$} & \multirow{4}{*}{$\approx 10$} & $w_{\mathrm{ord}}$ & \multirow{4}{*}{25328} & 5852 & 4614 & Fail & Fail & Fail & Fail & Fail & Fail & N/A\\
     & & $w_{\mathrm{sqrt}}$ & & 5400 & 3008 & 1210 & 987 & 858 & 874 & 992 & 1564 & Fail\\
     & & $w_{\mathrm{min}}$  & & 5376 & 2695 & 1127 & 907 & 719 & 690 & 675 & 660 & Fail\\
     & & $w_{\mathrm{max}}$  & & 5230 & 2977 & 1184 & 1022 & 905 & 960 & 1286 & 3895 & Fail \\
     \cmidrule(lr){2-13}
    & \multirow{4}{*}{$\approx 2$} & $w_{\mathrm{ord}}$ & \multirow{4}{*}{25883} & 6388 & 4422 & Fail & Fail & Fail & Fail & Fail & Fail& N/A\\
     & & $w_{\mathrm{sqrt}}$ & &5552 & 2885 & 1067 & 805 & 628 & 574 & 542 & 506 & Fail\\
     & & $w_{\mathrm{min}}$  & & 5426 & 3056 & 1168 & 966 & 802 & 775 & 752 & 740 & 1517\\
     & & $w_{\mathrm{max}}$  & & 5245 & 2904 & 1100 & 882 & 703 & 654 & 630 & 614 & Fail\\
    \bottomrule
  \end{tabular}
\end{table}

\begin{table}
    \small
  \caption{(SBM) Median of $T_H$, wall-clock time (in seconds) until the chain hit $\bfz^*$ over 50 replicates.  Entry with ``Fail'' indicates that chains never hit $\bfz^*$ in more than half of the replicated datasets.}
  \label{table:sbmTH}
  \centering
  \begin{tabular}{ccc c c c c c c c c c c}
    \toprule
    & CH &$N$ & 1 & 5 & 10 & 50 & 100 & 500 & 1000 & 2000 & 5000 & LBMH \\
    \midrule
    \multirow{8}{*}{$K=2$} & \multirow{4}{*}{$\approx 10$} & $w_{\mathrm{ord}}$ & \multirow{4}{*}{0.73} & 0.50 & 1.32 & Fail & Fail & Fail & Fail & Fail & Fail & N/A\\
     & & $w_{\mathrm{sqrt}}$ & & 0.33 & 0.27 & 0.39 & 0.62 & 1.64 & 2.75 & 5.04 & 13.49 & 1.95 \\
     & & $w_{\mathrm{min}}$  & & 0.41 & 0.29 & 0.36 & 0.53 & 1.29 & 2.03 & 3.52 & 8.27 & 1.09 \\
     & & $w_{\mathrm{max}}$  & &0.42 & 0.35 & 0.49 & 0.75 & 2.05 & 3.75 & 9.90 & 54.05 & 16.13\\
     \cmidrule(lr){2-13}
    & \multirow{4}{*}{$\approx 2$} & $w_{\mathrm{ord}}$ & \multirow{4}{*}{0.82} & 0.54 & 0.83 & Fail & Fail & Fail & Fail & Fail & Fail & N/A\\
     & & $w_{\mathrm{sqrt}}$ & & 0.4 & 0.35 & 0.48 & 0.75 & 1.49 & 2.15 & 3.78 & 9.30 & 1.28\\
     & & $w_{\mathrm{min}}$  & & 0.56 & 0.41 & 0.57 & 0.89 & 1.55 & 2.62 & 4.31 & 11.26 & 1.60\\
     & & $w_{\mathrm{max}}$  & & 0.47 & 0.39 & 0.51 & 0.75 & 1.29 & 2.07 & 3.57 & 8.35 & 1.35\\
     \midrule
    \multirow{8}{*}{$K=5$} & \multirow{4}{*}{$\approx 10$} & $w_{\mathrm{ord}}$ & \multirow{4}{*}{1.97} &1.35 & 1.39 & Fail & Fail & Fail & Fail & Fail & Fail & N/A\\
     & & $w_{\mathrm{sqrt}}$ & & 1.25 & 0.92 & 0.97 & 1.38 & 4.01 & 7.36 & 15.71 & 60.95 & Fail\\
     & & $w_{\mathrm{min}}$  & & 2.27 & 1.50 & 1.31 & 1.78 & 3.97 & 6.69 & 12.30 & 29.76 & Fail\\
     & & $w_{\mathrm{max}}$  & & 2.21 & 1.52 & 1.39 & 2.00 & 4.94 & 9.31 & 23.38 & 174.89 & Fail \\
     \cmidrule(lr){2-13}
    & \multirow{4}{*}{$\approx2$} & $w_{\mathrm{ord}}$ & \multirow{4}{*}{1.88} & 1.41 & 1.24 & Fail & Fail & Fail & Fail & Fail & Fail & N/A\\
     & & $w_{\mathrm{sqrt}}$ & & 1.22 & 0.84 & 0.81 & 1.05 & 2.31 & 3.84 & 6.89 & 15.90 & Fail\\
     & & $w_{\mathrm{min}}$  & & 1.35 & 0.97 & 0.91 & 1.31 & 2.99 & 5.17 & 9.61 & 23.64 & 6.52 \\
     & & $w_{\mathrm{max}}$  & & 1.33 & 0.89 & 0.83 & 1.16 & 2.61 & 4.36 & 7.93 & 19.18 & Fail\\
    \bottomrule
  \end{tabular}
\end{table}

Comparison of hitting iteration $H$ with LBMH gives an insight similar to the BVS example. When $K=5$, LBMH often gets stuck at a local mode and never converges to $\mathbf{z}^*$.  However when $K=2$, LBMH performs similarly to MTM with larger choices of $N$. We note that when $K=2$, the shape of posterior distribution can be significantly different from that of $K=5$, as the minimax rate and posterior contraction rate analysis are often treated separately when $K=2$ and $K\ge 3$  \citep{zhang2016minimax,zhuo2021mixing}. 
The comparison of wall-clock hitting time $T_H$ also suggests MTM with moderate choice of $N$ is much more efficient.

\clearpage

\subsection{Spatial clustering model (SCM)}\label{sec:append.SCM}
We consider a spatial clustering problem for a given set of spatial locations $\mathcal{S} = \{\mathbf{s}_1, \ldots, \mathbf{s}_p\} \subset \bbR^2$ where the responses $\zeta(\mathbf{s}_i)$'s are observed. The goal of an SCM is to identify a spatially contiguous partition on $\mathcal{S}$, denoted by $\mathcal{P} = \{\mathcal{S}_1, \ldots, \mathcal{S}_K\}$, where $\mathcal{S}_j$'s are disjoint subsets of $\mathcal{S}$ whose union is $\mathcal{S}$, such that the responses within a cluster $\{\zeta(\mathbf{s}): \mathbf{s} \in \mathcal{S}_j\}$  are identically distributed and have different means across clusters.

We follow \cite{luo2021bayesian} to adopt a probabilistic model for  $\mathcal{P}$ that utilizes a spanning tree graph $\mathcal{T}$ on $\mathcal{S}$ (with $p$ vertices and $p-1$ edges) as a ``spatial order'' of $\mathcal{S}$. The spanning tree $\mathcal{T}$ is chosen in a way that two locations connected by an edge are spatially proximate to each other. A partition $\mathcal{P}$ with $K$ clusters can be defined by removing $K\!-\!1$ edges from $\mathcal{T}$. Specifically, the SCM we consider can be written as
\begin{align}
    \zeta(\mathbf{s}_i) | \{\mu(\mathbf{s}_i)\}, \mathcal{P}, K, \sigma^2 &\stackrel{\text{ind}}{\sim} \mathsf{N}(\mu_j, \sigma^2), \quad \text{ with $\mu(\mathbf{s}_i) = \mu_j$ if $\mathbf{s}_i \in \mathcal{S}_j$, for $i \in [p]$}, \label{eq:bscc} \\
    \mu_j | \mathcal{P}, K, \sigma^2 &\stackrel{\text{iid}}{\sim} \mathsf{N}(0, \lambda^{-1}\sigma^2), \quad \text{ for $j \in [K]$}, \nonumber\\
    \sigma^2 &\sim \mathsf{InvGamma}(a_0/2, b_0/2), \nonumber\\
    \pi(\mathcal{P} | K) &\propto \ind\{\text{$\mathcal{P}$ can be obtained by removing $K-1$ edges from $\mathcal{T}$}\}, \nonumber\\
    \pi(K) &\propto (1 - c_0)^K,\quad K=1,\dots,p, \nonumber
\end{align}
where $a_0 > 0$, $b_0 > 0$, $0 \leq c_0 < 1$, and $\lambda > 0$ are hyperparameters. See Figure~\ref{fig:bscc_true} for an example of partition $\mathcal{P}$ obtained by cutting edges from the spanning tree $\mathcal{T}$. 

Thanks to the conjugate priors, $\mu_j$'s and $\sigma^2$ can be analytically marginalized out, and hence the inference problem boils down to drawing samples from the (discrete) posterior distribution $\pi(\mathcal{P} \given \mathrm{data})$.
Although \cite{luo2021bayesian} considered random spanning trees by assigning a prior distribution on $\mathcal{T}$, we stress that the main focus of this paper is the mixing time analysis on the posterior distribution of $\mathcal{P}$. Thus, following \cite{li2019spatial}, we fix $\mathcal{T}$ as the Euclidean minimum spanning tree in a Delaunay triangulation graph on $\mathcal{S}$, otherwise we can only sample from the conditional distribution $\pi(\mathcal{P} \given \mathcal{T},\mathrm{data})$ which complicates the mixing time analysis of our target distribution $\pi(\mathcal{P} \given \mathrm{data})$.

We consider the following proposal:
\begin{equation*}
    \mathbf{K}_{\mathrm{RW}}(\calP, \calP^\prime) = 1/(p-1) \ind_{\mathcal{N}_b(\calP) \cup \mathcal{N}_d(\calP)}(\calP^\prime),
\end{equation*}
where $\mathcal{N}_b(\calP)$ is the set of all possible partitions obtained by splitting a cluster in $\mathcal{P}$ into two clusters by selecting a cut-edge of $\mathcal{T}$ and $\mathcal{N}_d(\calP)$ is the set of all possible partitions obtained by merging two neighboring (with respect to $\mathcal{T}$) clusters in $\mathcal{P}$.
See \cite{luo2021bayesian,lee2021t} for detailed discussion on how to perform an appropriate split or merge on $\mathcal{P}$ given a spanning tree $\mathcal{T}$.

\textbf{Data generation}. We generate $p = 1000$ uniform locations $\bfs_i\iidsim \mathsf{Unif}([0,1]^2), i \in [p]$ and specify the true means $\{\mu(\mathbf{s}_i)\}$ as in Figure~\ref{fig:bscc_true}. Responses are generated according to \eqref{eq:bscc} with $\sigma = \sqrt{\text{Var}(\mu(\mathbf{s}))} / \text{SNR}$ and we simulate $50$ replicate datasets under $\text{SNR} \in \{3,10\}$ respectively.

\textbf{MCMC setup}. Following \cite{luo2021bayesian}, we initialize the chain using the estimates from the spatially clustered coefficient model of \citep{li2019spatial}. However, this initialization does not guarantee the same minimum number of iterations required to hit the true partition $\mathcal{P}^{\mathrm{true}}$ for different replicate datasets. For a fair comparison, throughout this subsection, we redefine $H$ as the number of \emph{extra} iterations until hit, which is the iterations until hit minus the minimum number of iterations required to reach $\mathcal{P}^{\mathrm{true}}$.  Hyperparameters are specified as $a_0 = b_0 = 1$, $c_0 = 0.5$, and $\lambda = 0.01$. We consider the number of trials $N \in \{5, 10, 100, 500, 1000\}$. For each replicate dataset, we run a chain of $10,000$ iterations for each MTM specification and a chain of $30,000$ for standard single-try MH. All simulation studies are performed on a Linux cluster with Intel(R) Xeon(R) Gold 6132 CPU @ 2.60GHz and 96GB memory.

Table~\ref{table:bscc_hit_iter} summarizes $H$ and $T_H$ of various weight functions and numbers of trials $N$. The distributions of $H$ for the setting of SNR = 10 are provided in Figure~\ref{fig:bscc_boxplot}. The results from single-try MH are also included as a baseline. When SNR = 10, the proposed locally balanced weight functions, especially $w_{\mathrm{sqrt}}$ and $w_{\mathrm{min}}$, considerably outperform the ordinary weight function $w_{\mathrm{ord}}$ and the single-try MH, in the sense that the proposed ones can reach $\mathcal{P}^{\mathrm{true}}$ by much fewer iterations when $N \in \{100, 500, 1000\}$. In contrast, the performance of $w_{\mathrm{ord}}$ deteriorates when $N \geq 100$ and it fails to reach $\mathcal{P}^{\mathrm{true}}$ when $N = 500$ or $1000$. 
For the proposed weight functions, the wall-clock time until hit $T_H$ is minimized when $N = 100$, since the benefit of having fewer iterations until hit is offset by the computational cost of extra trials when $N$ is large. 

When SNR = 3, the chains never visit the true partition, possibly because $\mathcal{P}^{\mathrm{true}}$ does not lead to the highest posterior probability. In this case, we redefine $H$ and $T_H$ to be the number of extra iterations and the wall-clock time, respectively, to reach the 0.99 Rand index neighborhood of $\mathcal{P}^{\mathrm{true}}$, defined as
\begin{equation*}
    \mathcal{N}_{\mathrm{Rand}}(\mathcal{P}^{\mathrm{true}}) \coloneqq \{\mathcal{P} : \mathrm{Rand}(\mathcal{P}, \mathcal{P}^{\mathrm{true}}) \geq 0.99\},
\end{equation*}
where $\mathrm{Rand}(\cdot, \cdot)$ is the Rand index \citep{rand1971objective} measuring the proportion of agreements between two partitions. The findings on $H$ and $T_H$ for reaching $\mathcal{N}_{\mathrm{Rand}}(\mathcal{P}^{\mathrm{true}})$ are similar to the ones when SNR = 10.

Finally, the median $N$ estimated from Algorithm~\ref{alg:N} using $\psi = 0.9$ over 50 replicate datasets is $\hat{N}=13$ when SNR = 10 and $\hat{N} = 21$ when SNR = 3.

\begin{figure}[h]
    \centering
    \begin{subfigure}[t]{0.3\textwidth}
    \caption{}
    \label{fig:bscc_true}
    \includegraphics[width=\textwidth]{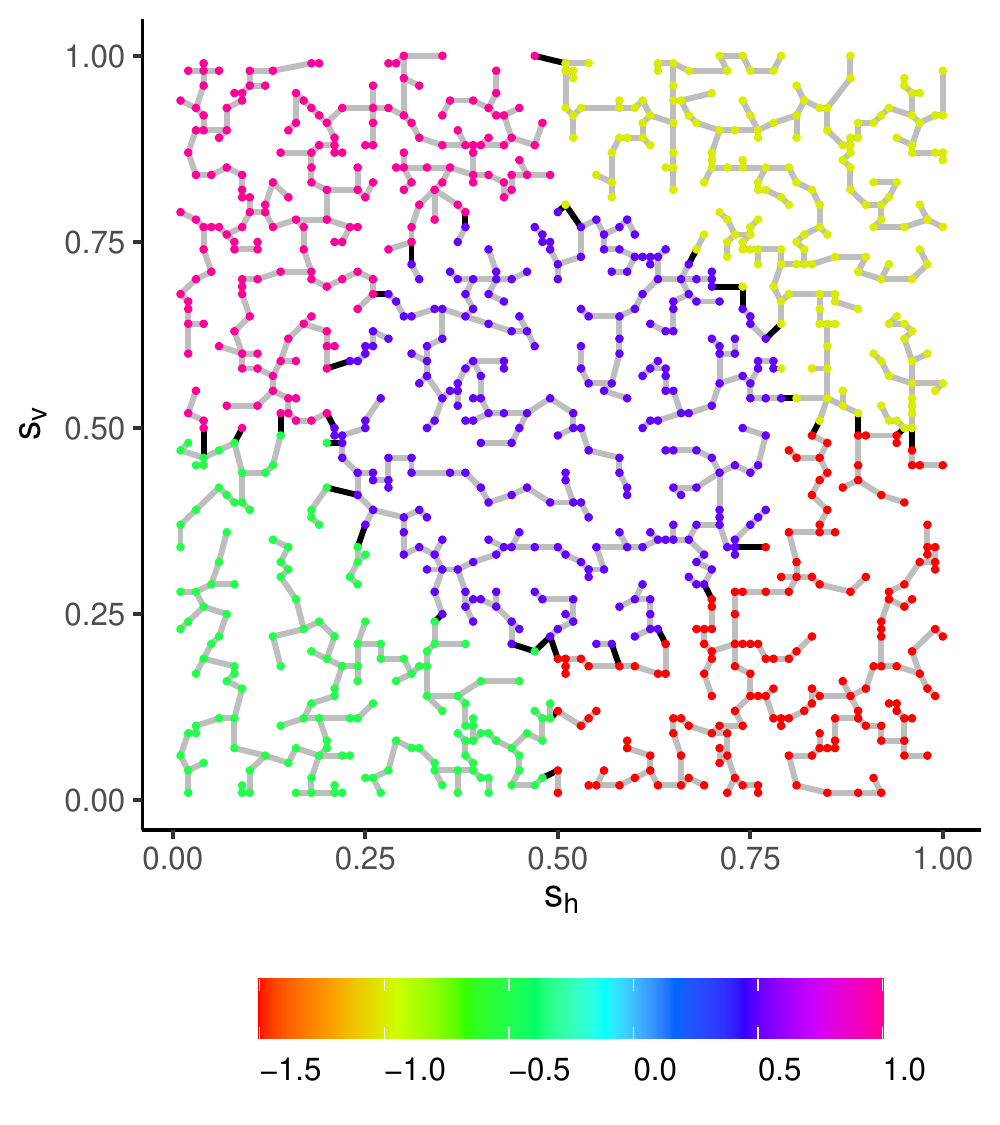}
    \end{subfigure}
    \hfill
    \begin{subfigure}[t]{0.68\textwidth}
    \caption{}
    \label{fig:bscc_boxplot}
    \includegraphics[width=\textwidth]{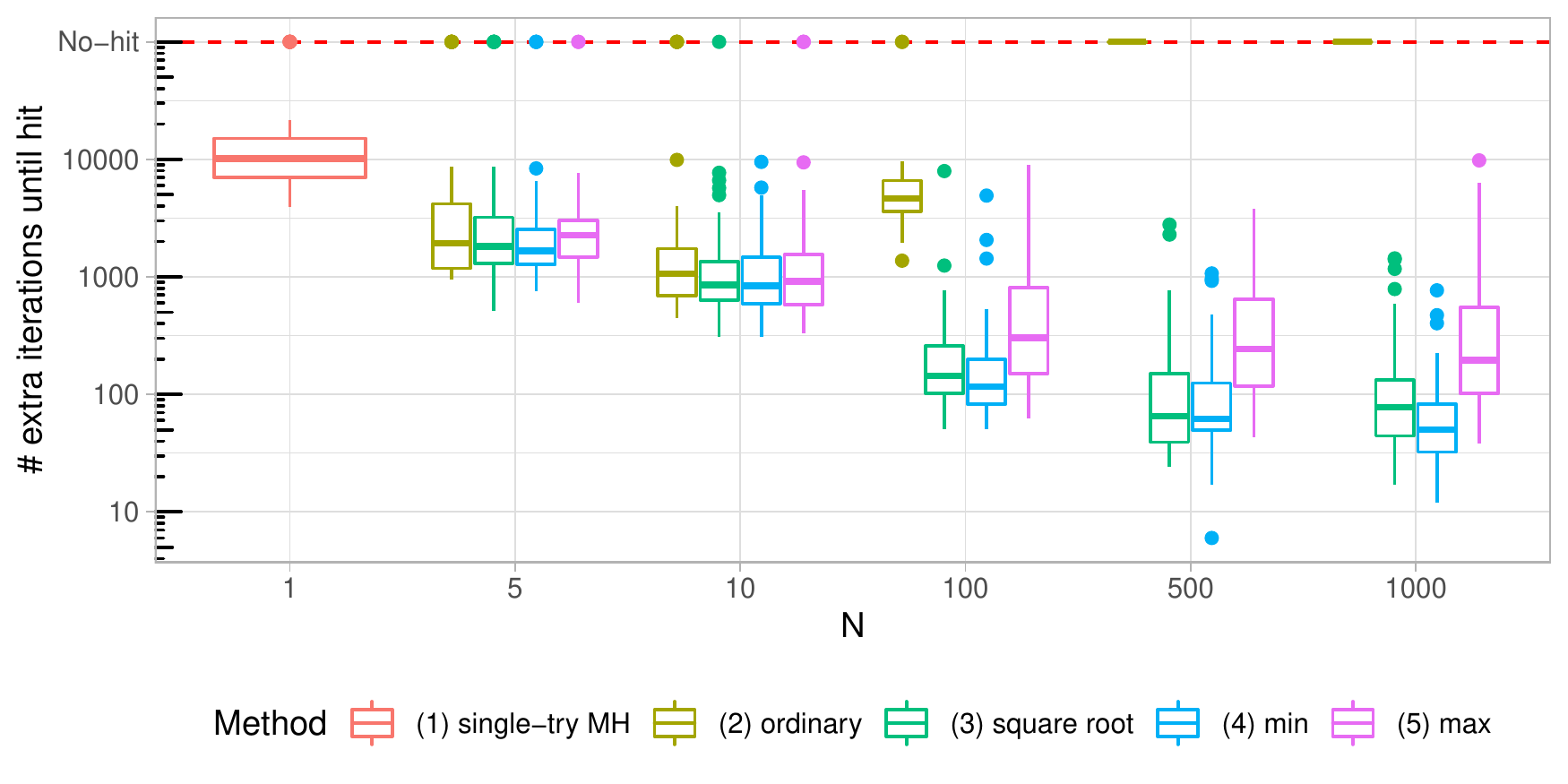}
    \end{subfigure}
    \caption{(a) True $\mu(\mathbf{s})$ and the Euclidean minimum spanning tree $\mathcal{T}$ on $\mathcal{S}$. Edges that should be removed in the true partition are marked in black. (b) Boxplot of numbers of extra iterations until hit for SCM under different numbers of trials and weight functions when SNR = 10.}
    \label{fig:bscc}
\end{figure}

\begin{table}[h]
    \small
  \caption{(SCM) Median of $H$ and $T_H$ (see text for definition; $T_H$ in seconds) over 50 replicates. Entry with ``Fail'' indicates that chains never hit the target state in more than half of the replicated datasets. }
  \label{table:bscc_hit_iter}
  \centering
  \begin{tabular}{c cc c c c c c c c}
    \toprule
   &  & $N$ & 1 & 5 & 10 & 100 & 500 & 1000  \\
    \midrule
    \multirow{8}{*}{$H$} & \multirow{4}{*}{SNR $=10$} & $w_{\mathrm{ord}}$ & \multirow{4}{*}{10124} & 1933 & 1058 & 4638 & Fail & Fail \\
     & & $w_{\mathrm{sqrt}}$ & & 1822 & 851 & 144 & 66 & 78 \\
     & & $w_{\mathrm{min}}$  & & 1661 & 840 & 117 & 62 & 50\\
     & & $w_{\mathrm{max}}$  & & 2260 & 918 & 304 & 244 & 195\\
    \cmidrule(lr){2-9}
    & \multirow{4}{*}{SNR $=3$} & $w_{\mathrm{ord}}$ & \multirow{4}{*}{23950} & 5083 & 2958 & 1046 & Fail & Fail \\
     & & $w_{\mathrm{sqrt}}$ & & 2860 & 2064 & 235 & 84 & 126 \\
     & & $w_{\mathrm{min}}$  & & 4443 & 2545 & 234 & 236 & 121 \\
     & & $w_{\mathrm{max}}$  & & 6470 & 2810 & 394 & 462 & 301 \\
    \midrule
     \multirow{8}{*}{$T_H$} & \multirow{4}{*}{SNR $=10$} & $w_{\mathrm{ord}}$ & \multirow{4}{*}{906.72} & 164.28 & 90.75 & 1291.13 & Fail & Fail \\
     & & $w_{\mathrm{sqrt}}$ & & 149.48 & 74.15 & 45.98 & 114.68 & 253.31 \\
     & & $w_{\mathrm{min}}$  & & 182.53 & 94.21 & 42.74 & 117.88 & 206.83 \\
     & & $w_{\mathrm{max}}$  & & 235.46 & 96.17 & 85.17 & 274.57 & 465.76 \\
     \cmidrule(lr){2-9}
    & \multirow{4}{*}{SNR $=3$} & $w_{\mathrm{ord}}$ & \multirow{4}{*}{2058.61} & 524.33 & 318.72 & 245.42 & Fail & Fail \\
    & & $w_{\mathrm{sqrt}}$ & & 330.36 & 229.80 & 62.23 & 113.15 & 299.45 \\
    & & $w_{\mathrm{min}}$  & & 486.18 & 264.62 & 58.71 & 233.03 & 288.35 \\
    & & $w_{\mathrm{max}}$  & & 721.71 & 303.91 & 97.47 & 438.06 & 628.02 \\
    \bottomrule
  \end{tabular}
\end{table}

\clearpage 

\subsection{MTM algorithm on multimodal target distributions}\label{sec:append.multimodal}

In this section, we analyze the performance of the MTM algorithm with different choices of weight functions and $N$ on the multimodal target distribution. Following \cite{zhou2021dimension}, we generate a multimodal dataset in the context of a Bayesian variable selection problem.

\textbf{Data generation}. We let sample size $n=1000$ and number of variables $p=5000$. Each row of design matrix is independently sampled from $\bm{x}_i\iidsim \mathsf{N}\left(0, \Sigma\right)$ for $i=1,\dots,n$ where $\Sigma=\operatorname{diag}\left(\Sigma_{20}, \ldots, \Sigma_{20}\right)$ is block-diagonal. Each block $\Sigma_{20}$ has dimension $20\times 20$, and $\left(\Sigma_{20}\right)_{j k}=\exp(-|j-k| / 3)$. We generate true coefficient $\bm\beta^{\mathrm{true}}$ by first sampling 100 indices $j_1,\dots,j_{100}$ uniformly at random (without replacement) from $[p]$ and let $\beta_{j_\ell} \iidsim \mathsf{N}(0,\sigma_\beta^2)$ for $\ell=1,\dots,100$, and $\beta_{k}=0$ if $k\not\in \{j_1,\dots,j_{100}\}$. Then the response vector $\bm{y}$ is generated from $\bm{y}\sim \mathsf{N}(\bm{X}\bm{\beta}^{\mathrm{true}},\mathbf{I}_n)$. We consider three settings of $\sigma_{\beta}=0.3,0.4,0.5$ to simulate the coefficients and data. For each setting, we simulate 20 datasets.

\textbf{MCMC setup}. We use the same BVS model described in Section~\ref{sec:simulation}. Hyperparameters are specified as $\mathscr{G} = p = 5000$, $\kappa = 1$, and $s_{\mathrm{max}} = 100$. For each dataset, we run a chain of 10,000 iteration for MTM with $N=50,100,500,1000,2000,5000$ using four different weight functions. The first 2000 iterations are discarded since the behavior of the chain (e.g. acceptance ratio) during the burn-in stage may be different from the behavior of the chain which entered stationarity; see also Figure~\ref{fig:traceplotmultimodal}. for trace plots. Algorithms are all initialized with null model $\gamma_0 = \bm{0}$. All simulation studies are performed on a Linux cluster with Intel(R) Xeon(R) Gold 6132 CPU @ 2.60GHz and 96GB memory.

Since the target distribution is no longer unimodal, hitting iteration $H$ and wall-clock hitting time $T_H$ are not appropriate metrics to compare mixing performance. Instead, we use three different metrics to evaluate the quality of the mixing: 1) acceptance ratio, 2) the number of unique states visited by the chain, denoted by \#(unique $\gamma$), and 3) ESS/Time, where ESS is the effective sample size calculated from the hamming distances $d_\mathrm{H}(\hat\gamma_{\mathrm{max}}, \gamma_t)$, $t=2001,\dots,10000$ from the maximum posterior state $\hat\gamma_{max}$ found in a chain, Time is wall-clock time usage, measured in seconds. The results are summarized in Table~\ref{table:multimodal}.

\begin{table}[h]
    \small
  \caption{Multimodal posterior simulation results based on the acceptance rate, the number of unique states visited by a chain, and the effective sample size divided by the running time. All statistics are based on chains with a length of 8000 (2000 burn-in), and are averaged over 20 datasets.}
  \label{table:multimodal}
  \centering
  \begin{tabular}{ccc  c c c c c c}
    \toprule
    & &$N$ & 50 & 100 & 500 & 1000 & 2000 & 5000 \\
    \midrule
    \multirow{12}{*}{$\sigma_\beta = 0.3$} & \multirow{4}{*}{Acc. Rate} & $w_{\mathrm{ord}}$ & 0.041 & 0.076 & 0.290 & 0.436 & 0.509 & 0.051\\
     & & $w_{\mathrm{sqrt}}$ & 0.039 & 0.069 & 0.191 & 0.257 & 0.314 & 0.385\\
     & &  $w_{\mathrm{min}}$ &0.042 & 0.075 & 0.299 & 0.430 & 0.576 & 0.717\\
     & & $w_{\mathrm{max}}$ & 0.008 & 0.011 & 0.021 & 0.031 & 0.043 & 0.067 \\
     \cmidrule(lr){2-9}
    & \multirow{4}{*}{\#(\text{unique }$\gamma$)} & $w_{\mathrm{ord}}$  & 233.4 & 425.5 & 1398.3 & 2052.6 & 2130.8 & 161.8\\
     & & $w_{\mathrm{sqrt}}$ &240.2 & 385.1 & 993.2 & 1314.7 & 1549.6 & 1885.2\\
     & & $w_{\mathrm{min}}$  & 255.1 & 408.4 & 1538.8 & 2056.4 & 2684.8 & 3168.8\\
     & & $w_{\mathrm{max}}$  & 53.0 & 71.0 & 118.1 & 158.9 & 211.8 & 311.0\\
     \cmidrule(lr){2-9}
    & \multirow{4}{*}{ESS/Time} & $w_{\mathrm{ord}}$ & 7.04 & 8.71 & 9.32 & 6.12 & 2.27 & 0.02\\
     & & $w_{\mathrm{sqrt}}$ &5.11 & 8.08 & 7.14 & 4.62 & 2.48 & 1.14 \\
     & & $w_{\mathrm{min}}$ &5.35 & 7.11 & 9.38 & 7.17 & 4.20 & 1.94\\
     & & $w_{\mathrm{max}}$ & 1.20 & 1.12 & 0.49 & 0.41 & 0.13 & 0.06\\
     \midrule
     \multirow{12}{*}{$\sigma_\beta = 0.4$} & \multirow{4}{*}{Acc. Rate} & $w_{\mathrm{ord}}$ &  0.036 & 0.069 & 0.260 & 0.389 & 0.417 & 0.017\\
     & & $w_{\mathrm{sqrt}}$ & 0.034 & 0.057 & 0.172 & 0.233 & 0.285 & 0.350 \\
     & & $w_{\mathrm{min}}$  & 0.036 & 0.064 & 0.261 & 0.404 & 0.541 & 0.674\\
     & & $w_{\mathrm{max}}$ & 0.007 & 0.009 & 0.018 & 0.026 & 0.040 & 0.059 \\
     \cmidrule(lr){2-9}
    & \multirow{4}{*}{\#(\text{unique }$\gamma$)} & $w_{\mathrm{ord}}$  & 211.4 & 392.6 & 1302.9 & 1782.8 & 1852.1 & 72.4\\
     & & $w_{\mathrm{sqrt}}$ & 205.8 & 316.1 & 932.5 & 1277.8 & 1470.2 & 1767.4\\
     & & $w_{\mathrm{min}}$  & 219.8 & 361.6 & 1310.4 & 1919.6 & 2565.6 & 3067.3\\
     & & $w_{\mathrm{max}}$  & 50.2 & 57.4 & 103.6 & 135.3 & 197.6 & 279.4\\
     \cmidrule(lr){2-9}
    & \multirow{4}{*}{ESS/Time} & $w_{\mathrm{ord}}$ & 6.13 & 8.06 & 6.31 & 6.58 & 0.52 & 0.03\\
     & & $w_{\mathrm{sqrt}}$ & 4.79 & 6.81 & 6.06 & 4.04 & 2.10 & 1.09 \\
     & & $w_{\mathrm{min}}$ & 5.42 & 7.73 & 7.26 & 5.07 & 3.83 & 1.87\\
     & & $w_{\mathrm{max}}$ &1.04 & 0.80 & 0.35 & 0.23 & 0.11 & 0.05\\
     \midrule
      \multirow{12}{*}{$\sigma_\beta = 0.5$} & \multirow{4}{*}{Acc. Rate} & $w_{\mathrm{ord}}$ &  0.028 & 0.051 & 0.209 & 0.337 & 0.263 & 0.001\\
     & & $w_{\mathrm{sqrt}}$ & 0.024 & 0.046 & 0.138 & 0.190 & 0.242 & 0.306 \\
     & & $w_{\mathrm{min}}$  &0.026 & 0.055 & 0.209 & 0.337 & 0.481 & 0.633\\
     & & $w_{\mathrm{max}}$ &0.005 & 0.007 & 0.015 & 0.024 & 0.036 & 0.056\\
     \cmidrule(lr){2-9}
    & \multirow{4}{*}{\#(\text{unique }$\gamma$)} & $w_{\mathrm{ord}}$  &146.2 & 261.6 & 924.7 & 1382.0 & 976.5 & 5.6\\
     & & $w_{\mathrm{sqrt}}$ &132.2 & 246.4 & 648.8 & 877.4 & 1054.0 & 1323.6  \\
     & & $w_{\mathrm{min}}$  &138.2 & 291.0 & 958.2 & 1445.6 & 2001.9 & 2444.1\\
     & & $w_{\mathrm{max}}$  &37.4 & 48.5 & 83.8 & 124.6 & 177.7 & 257.7\\
     \cmidrule(lr){2-9}
    & \multirow{4}{*}{ESS/Time} & $w_{\mathrm{ord}}$ &5.93 & 9.05 & 9.85 & 7.77 & 0.32 & 0.01\\
     & & $w_{\mathrm{sqrt}}$ & 6.75 & 6.83 & 7.33 & 5.05 & 3.27 & 1.17\\
     & & $w_{\mathrm{min}}$ & 6.73 & 6.45 & 8.73 & 5.15 & 5.41 & 2.54\\
     & & $w_{\mathrm{max}}$ & 1.97 & 1.00 & 0.53 & 0.39 & 0.15 & 0.07\\
     \bottomrule
  \end{tabular}
\end{table}

\begin{figure}[h]
    \centering
    \caption{Examples of MTM trace plot using different weight functions and the number of trials $N$. (Left) Simulated data with $\sigma_\beta = 0.3$, (Right) Simulated data with $\sigma_\beta = 0.5$. Each row corresponds to $N=50,500,5000$. All chains are initialized at the null model $\gamma_0 = \bm{0}$. }
    \label{fig:traceplotmultimodal}
    \includegraphics[width=\textwidth]{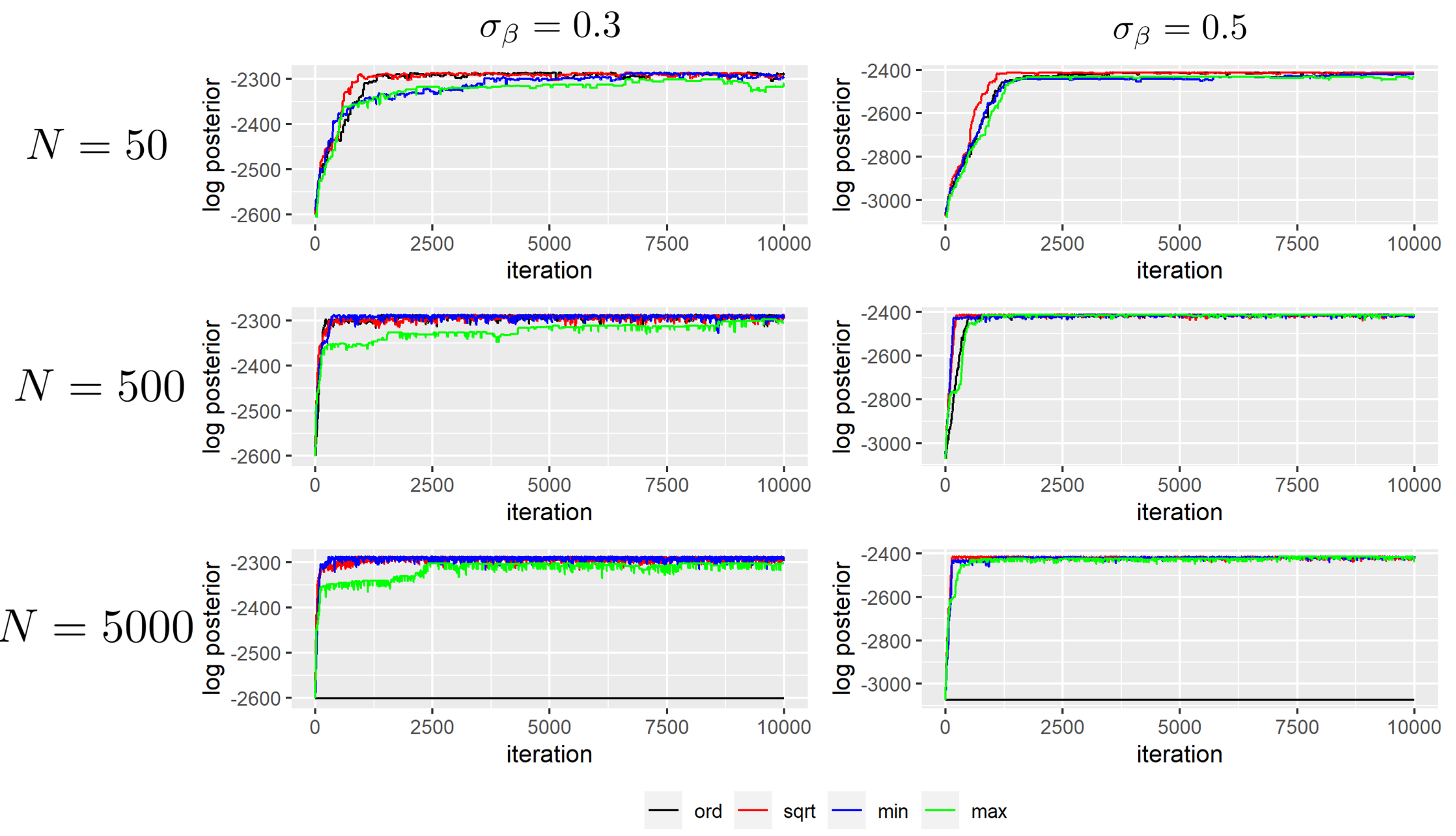}
\end{figure}

Figure~\ref{fig:traceplotmultimodal} shows that as $N$ increases, chains generally move faster towards the high posterior states which suggests that discarding the first 2000 samples is reasonable. One exception is when $N=5000$, the chain using ordinary weight function barely moves to another state, whereas the chain using locally balanced weight functions exhibits better mixing properties than $N=50$ and $N=500$.

Finally, Table~\ref{table:multimodal} summarizes the multimodal simulation results. Similar to the previous unimodal results, the performance of $w_{\mathrm{ord}}$ deteriorates as $N$ being large, especially when $\sigma_\beta$ is large.
When $N$ is moderate, the weighting functions $w_{\mathrm{ord}}$ and $w_{\mathrm{min}}$ has a better mixing property than $w_{\mathrm{sqrt}}$, while $w_{\mathrm{max}}$ being the worst. The inferior performance of $w_{\mathrm{max}}$ is also observed at Table~\ref{table:bvsH}. If the current state $x$ is one of the local modes, we may have $\max \{1, \pi(x^\star)/\pi(y)\} >\!\!> \max \{1, \pi(y) / \pi(x) \}$ where $y$ is the proposed state, and $x^\star$ is one of the trials of $y$ in Step 3 of Algorithm~\ref{alg:MTM}, which makes the acceptance probability very small so that the chain is stuck at $x$ when using $w_{\mathrm{max}}$. On the other hand, the magnitude of the difference between $h(\pi(x^\star)/\pi(y))$ and $h(\pi(y)/\pi(x))$ will be reduced if we use $w_{\mathrm{sqrt}}, w_{\mathrm{min}}$ so that the chain can traverse among local modes.
To summarize, the simulation results suggest that $w_{\mathrm{min}}$ would be the best choice of the weight function if the multimodality exists, since 
it not only traverses the multimodal posterior efficiently, but also is robust to large $N$.
We envision that there are a number of ways to improve the mixing under the multimodal posterior by combining it with techniques such as annealing or tempering \cite{casarin2013interacting}.

\section{Real data applications}\label{sec:appendix:real}

\subsection{GWAS dataset for Bayesian variable selection}\label{subsec:real:BVS}

We consider a genome-wide association study (GWAS) dataset on glaucoma studied in \cite{zhou2021dimension} with sample size $n=5418$ and number of genetic variants $p=7255$.  The response variable $\mathbf{y}\in\mathbb{R}^{5418}$ is the standardized cut-to-disk ratio measurements averaged over two eyes.
We use the BVS model described in Section~\ref{subsec:simul:BVS}, with hyperparameters $\mathscr{G} = 100$ and $\kappa = 0.8$. Since the ``true" state is not available, we compare the acceptance rate and the number of unique states visited, averaging over 5 chains. From Table~\ref{table:gwas}, it is clear that the performance of $w_{\mathrm{ord}}$ deteriorates significantly as $N$ grows whereas $w_{\mathrm{sqrt}}$, $w_{\mathrm{min}}$, $w_{\mathrm{max}}$ does not. 
We also report the posterior inclusion probabilities of the top 10 genetic variants in Tables \ref{table:gwas1}, \ref{table:gwas2}, \ref{table:gwas3}, \ref{table:gwas4}, \ref{table:gwas5}, and \ref{table:gwas6} in Appendix \ref{sec:addtable}. All results generally agree with the result of \cite{zhou2021dimension}, except that when we use $w_{\mathrm{ord}}$ with $N=5000$, the chain is stuck at local modes and fails to find the significant genetic variants. 

\vspace{-2mm}
\begin{table}[h]
    \small
    \caption{GWAS dataset analysis results, averaged over 5 chains with random seeds.}
  \centering
  \label{table:gwas}
  \begin{tabular}{cc  c c c c c c }
    \toprule
     &$N$ & 50 & 100 & 500 & 1000 & 2000 & 5000 \\
     & iteration & $10^6$ & $5\times 10^5$ & $10^5$ & $5\times 10^4$ & $2\times 10^4$ & $10^4$ \\
    \midrule
    \multirow{4}{*}{Acc. Rate} & $w_{\mathrm{ord}}$ & 0.4014 & 0.5081 & 0.1370 & 0.0471 & 0.0259 & 0.0085 \\
      & $w_{\mathrm{sqrt}}$ & 0.3407 & 0.4812 & 0.7571 & 0.8325 & 0.8777 & 0.9251 \\ 
      &  $w_{\mathrm{min}}$ & 0.4136 & 0.5851 & 0.8252 & 0.8797 & 0.9138 & 0.9455 \\ 
      & $w_{\mathrm{max}}$ & 0.2335 & 0.3404 & 0.6161 & 0.7199 & 0.7930 & 0.8698 \\
     \cmidrule(lr){1-8}
     \multirow{4}{*}{\#(\text{unique states})} & $w_{\mathrm{ord}}$  &  199442 & 126238 & 6796 & 1172 & 259 & 43 \\
      & $w_{\mathrm{sqrt}}$ & 169242 & 119396 & 37481 & 20588 & 8669 & 4563 \\
      & $w_{\mathrm{min}}$  & 205459 & 145372 & 40993 & 21848 & 9080 & 4696 \\ 
      & $w_{\mathrm{max}}$  & 115615 & 84124 & 30227 & 17591 & 7717 & 4211 \\ 
     \bottomrule
  \end{tabular}
\end{table}

\subsection{Single-cell RNA dataset for structure learning}\label{subsec:real:structure_learning}

We consider a gene expression dataset on Alzheimer's disease used in~\cite{chang2022order} with the sample size $n=1666$ and the number of genes $p=73$. The goal is to learn the underlying directed acyclic graph (DAG) model among the $p$ genes.  
Due to acyclicity, each DAG has at least one ordering of the nodes.
For example, the ordering for the DAG $a \rightarrow b \leftarrow c$ can be either $(a, c, b)$ or $(c, a, b)$. 
A popular Bayesian structure learning strategy is to use MCMC sampling to first learn the marginal posterior distribution on the order space and then find one or multiple best DAGs for each sampled ordering. 

We use an MTM implementation of the order MCMC sampler proposed in~\cite{chang2022order}, which aims to learn the posterior distribution on the order space $\mathbb{S}^p$, the permutation group on $\{1, \dots, p\}$.
The size of our model space $\mathbb{S}^p$ is equal to $73! \approx 4.5 \times 10^{105}$. 
For each weight function and each setting, we simulate 30 chains, initialized at $(1,\dots, 73)$. 
It is clear from Table~\ref{table:dag} that the acceptance probability with ordinary weight function $w_{\mathrm{ord}}$ significantly deteriorates, which is consistent with our theory. We can see this tendency more clearly in the log-posterior trace plots for all weight functions in Figure~\ref{fig:traj_dag}. 

\vspace{-2mm}
\begin{table}[h]
    \small
  \caption{The single-cell RNA database for Alzheimer’s disease analysis results, averaged over 30 chains with random seeds. The number in the parenthesis is the standard error.}
  \label{table:dag}
  \centering
  \begin{tabular}{cc  c c  c  c  c}
    \toprule
     &$N$ & 5 & 50  \\
     & iteration & $5\times 10^2$ & $2 \times 10^2$  \\
    \midrule
    \multirow{4}{*}{Acc. Rate} &
    $w_{\mathrm{ord}}$ &  0.7187 (0.004) & 0.0012 (0.000) \\
      & $w_{\mathrm{sqrt}}$ & 0.8029 (0.004) &  0.9186 (0.003) \\
      &  $w_{\mathrm{min}}$ &0.8329 (0.002) & 0.9506 (0.001) \\
      & $w_{\mathrm{max}}$ & 0.6643 (0.006) & 0.6806 (0.008) \\
     \cmidrule(lr){1-4}
     \multirow{4}{*}{\#(\text{unique orderings})} & 
     $w_{\mathrm{ord}}$  & 361.1 (3.1) & 1.333333 (0.1) \\
      & $w_{\mathrm{sqrt}}$ &  402.3 (2.2)  & 184.3 (0.8) \\
      & $w_{\mathrm{min}}$  & 416.7 (1.8) & 191.8 (0.6) \\
      & $w_{\mathrm{max}}$  &  332.0 (2.9)   & 136.2 (1.6)   \\
     \bottomrule
  \end{tabular}
\end{table}

\begin{figure}
    \centering
    \includegraphics[width=0.34\textwidth]{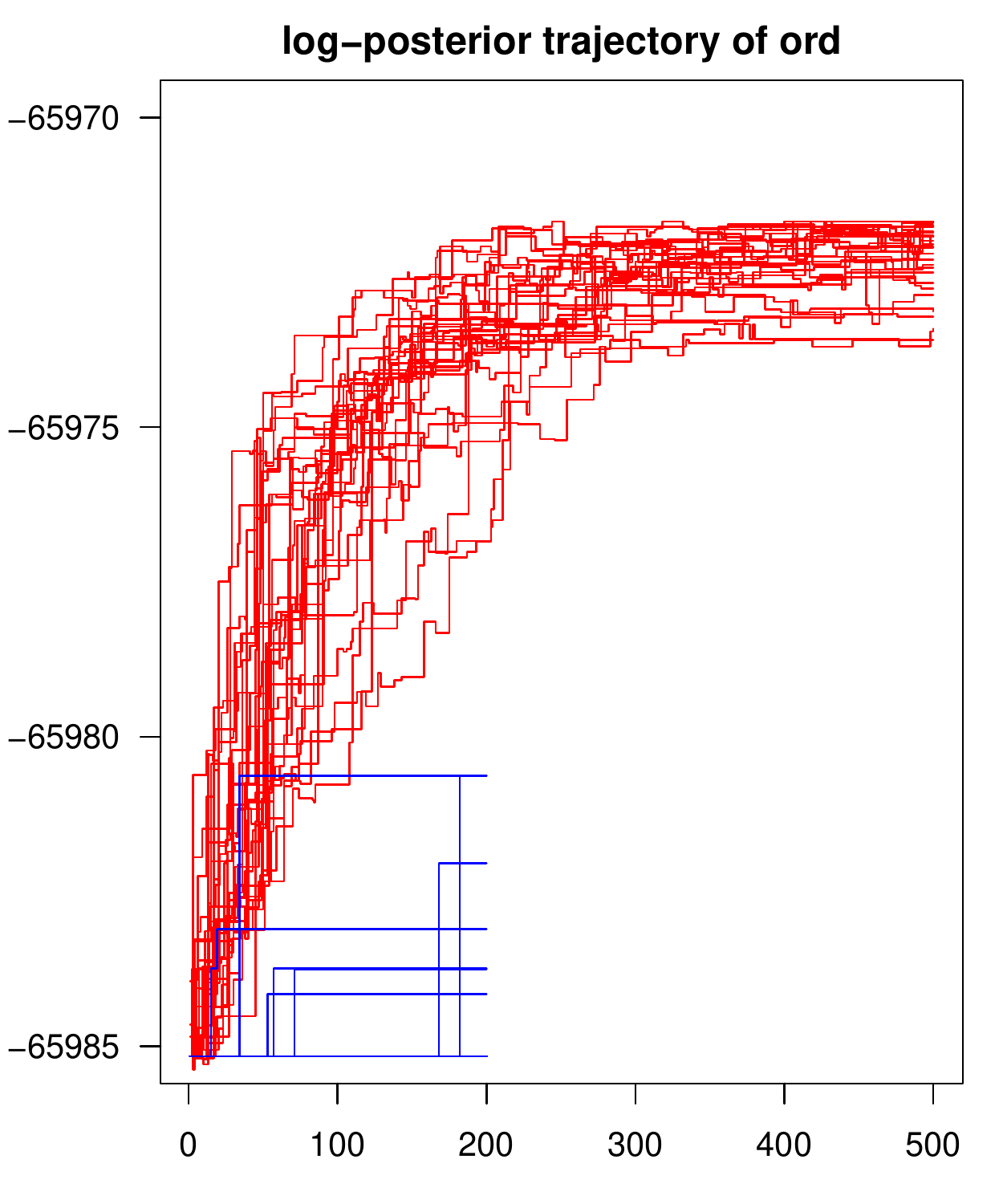}
    \includegraphics[width=0.34\textwidth]{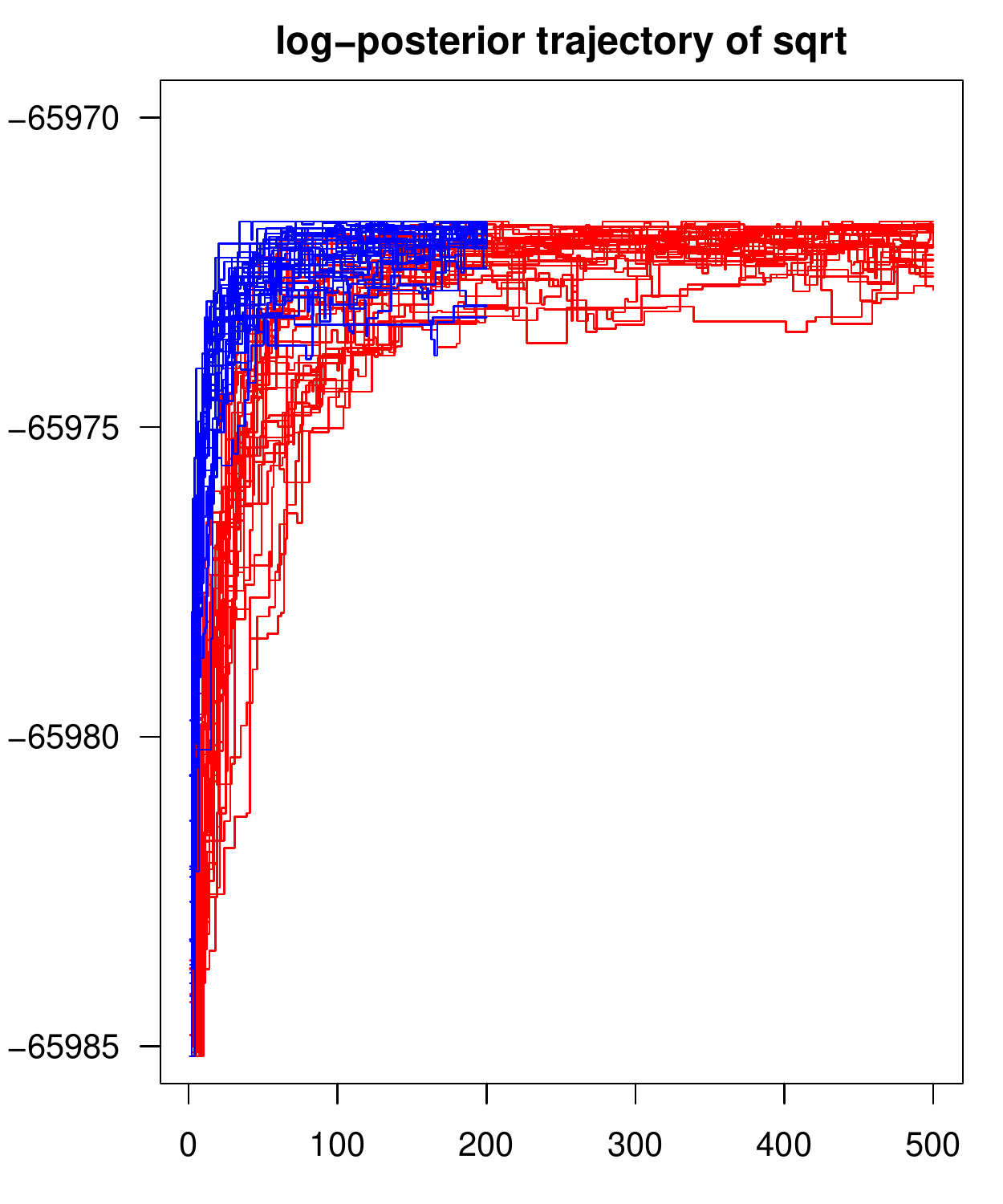}
    \includegraphics[width=0.34\textwidth]{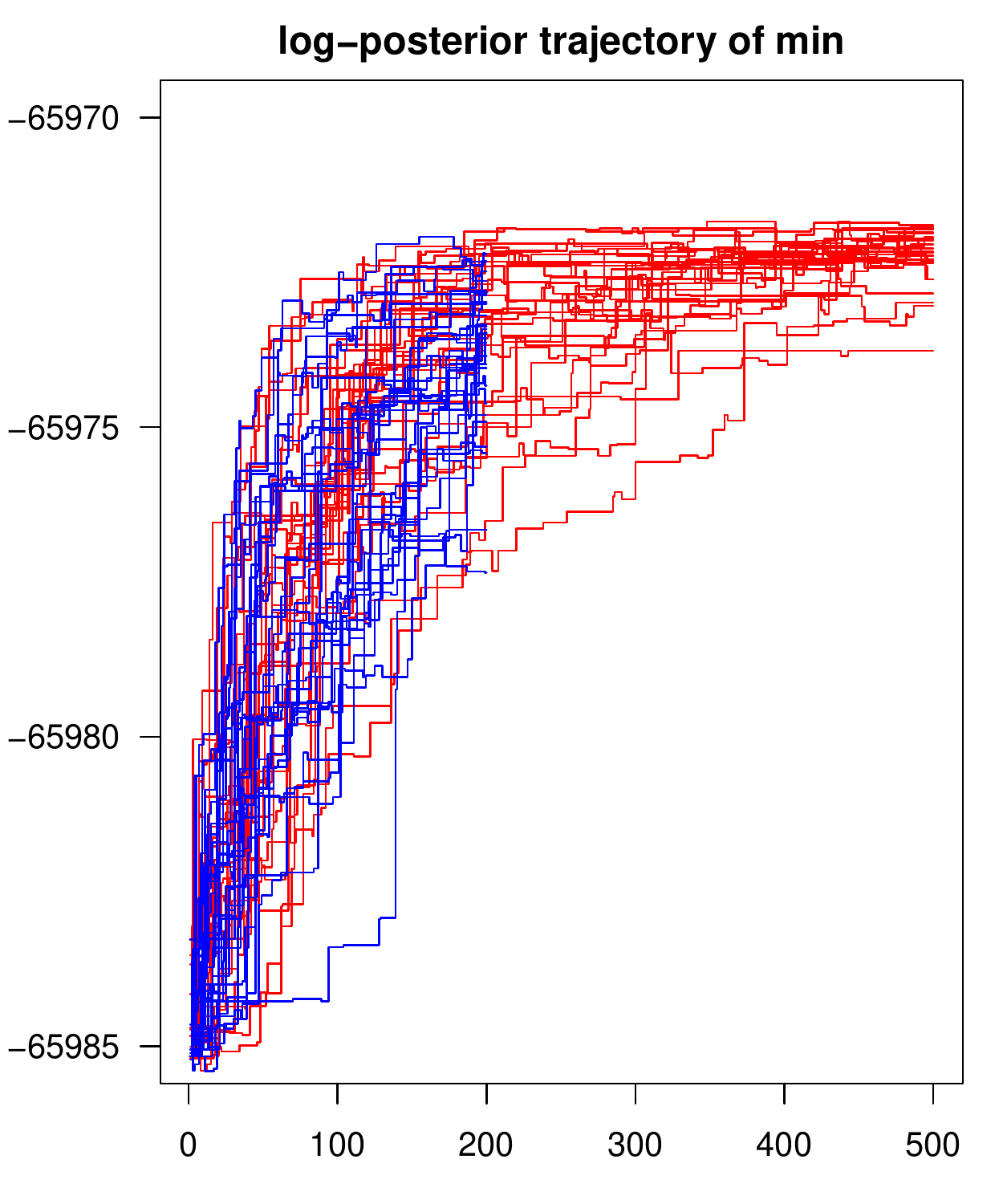}
    \includegraphics[width=0.34\textwidth]{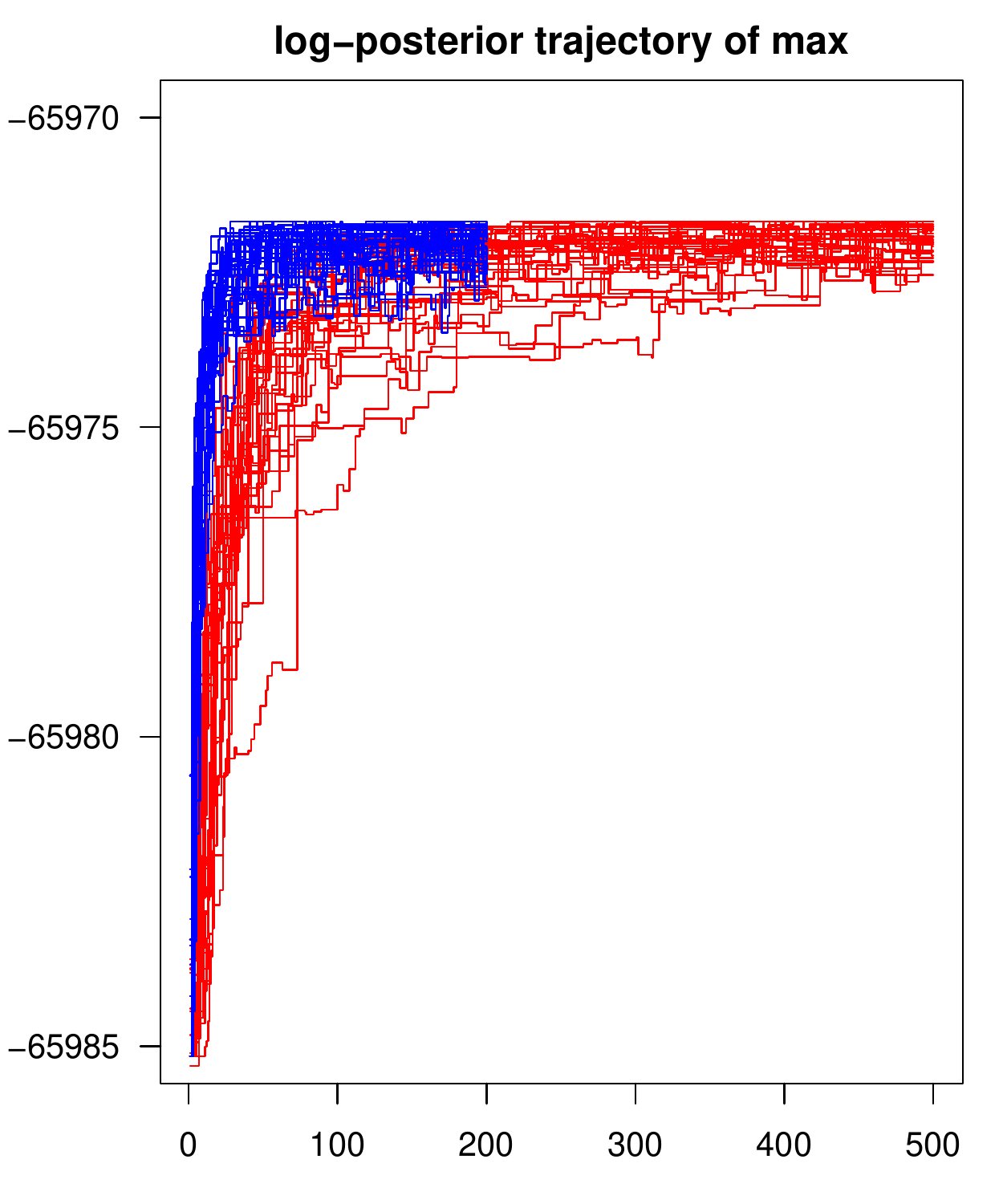}
    \caption{Log-posterior trace plots. Red trajectories indicate MTM with the number of trial $N = 5$ with 500 iterations, and blue trajectories indicate MTM with $N = 50$ with 200 iterations. }
    \label{fig:traj_dag}
\end{figure}

\clearpage 

\section{Additional discussion}\label{sec:adddiscussion}

\subsection{On parallelization (vectorization) }\label{subsec:parallelization}

As discussed in \textit{Remark 3} in Section \ref{subsec:weight}, the overall theoretical computational complexity of the MTM algorithm until the convergence remains the same as a usual MH algorithm. However, MTM enables parallel computations when evaluating $N$ weight functions and hence leads to a significant practical computational gain as evidenced by the reduced wall-clock hitting time reported in Table \ref{table:th}.
Under the random walk proposal $\MH$, the evaluation of weight functions is equivalent to the evaluation of target distribution at $N$ states $y_1,\dots,y_N$. 
Here we clarify that the scope of parallelism we consider is the instruction-level parallelism~\cite{rau1993instruction}, also called vectorization. 
Thanks to modern 
When task-level parallelism, assigning a set of independent tasks in parallel across several processors, is employed to the MTM algorithm within each MCMC iteration, it suffers from communication overhead unless the evaluation of target distributions takes extremely long.
Thanks to the optimized linear algebra libraries such as BLAS~\cite{blackford2002updated}, the easiest way to achieve instruction-level parallelism is to convert the problem of evaluating target distribution at multiple states $\pi(y_1),\dots,\pi(y_N)$ to a series of matrix multiplication problems. 

Here we outline the computational strategy to simultaneously calculate $\pi(y_1),\dots,\pi(y_N)$ for BVS and SBM. For BVS, since only one variable is added or deleted in the proposal, the Cholesky rank-1 update \citep{smith1996nonparametric,george1997approaches} is utilized to get $\pi(y_1),\dots,\pi(y_N)$ from $\pi(x)$. To be specific, assume $\gamma_1'\dots,\gamma_N'$ are obtained by adding a variable from $\gamma$. By \eqref{eq:bvsposterior}, the evaluation of $\pi(\gamma_j\given\bm{y})$, $j=1,\dots,N$ corresponding to evaluating $\mathrm{SSR}(\gamma_j)$, $j=1,\dots,N$ from $\mathrm{SSR}(\gamma)$ saved from the previous iteration. We refer \cite[][Appendix B]{zanella2019scalable} for details of vectorization procedure with Cholesky rank-1 update. 
For SBM, let $A_i\in\{0,1\}^{p}$ be the $i$th column of adjacency matrix and $Z\in\{0,1\}^{p\times K}$ be one-hot encoded partition matrix such that  $Z_{i,k}=1$ if $z_i = k$ and 0 otherwise. Since a node is assigned to another block one at a time, the calculation of $\pi(y_1)\dots,\pi(y_N)$ given current state $\pi(x)$ can be done by counting the change of the number of edges between blocks; see \eqref{eq:sbmposterior}. Letting $A_J\in\{0,1\}^{p\times N}$ where column $A_j$ corresponds to the $j$th proposal, the matrix-matrix multiplication $Z^\top A_J$ allows to calculate $\pi(y_1)\dots,\pi(y_N)$ simultaneously from $\pi(x)$. In addition, if the graph is sparse, then sparse matrix multiplication algorithms can be utilized for further speedup.

\subsection{On state space}\label{subsec:state_space}

Our state space of interest is finite (so discrete), but the proposed locally balanced MTM algorithm is also applicable to continuous state spaces which will be shown shortly.  
We choose to focus on the discrete case since the theory on continuous state spaces is usually developed under very different frameworks (and likewise, the theory on continuous spaces often cannot be readily applied to discrete ones). 
Indeed, developing MCMC theory or methodology on discrete spaces is often regarded as more challenging than on continuous ones~\citep[][Section 1]{zanella2020informed}, due to the lack of gradient information and a widely accepted theoretical framework supported by statistical theory (for comparison, on continuous spaces, one often assumes log-concavity or asymptotic normality of the target posterior distribution). 

To some extent, the proposed MTM method is conceptually similar to MALA (Metropolis adjusted Langevin algorithm) or HMC (Hamiltonian Monte Carlo) on continuous spaces in that MTM evaluates the ``gradient'' by a random search of neighboring states. This suggests that for continuous-state-space problems where the gradient of log-posterior cannot be easily evaluated (e.g. Bayesian inverse problems and Gaussian process regression models), the proposed MTM method can be quite useful.

We conclude this section with a simulation study that shows the weight function proposed in Proposition 2 can lead to an improved MTM algorithm on continuous spaces. Suppose our target distribution is the 10-dimensional Gaussian distribution $\mathsf{N}(0, \mathbf{I}_{10})$. We set our proposal distribution $q(\cdot | x ) = \mathsf{N}_{10}(x, 10^{-2} \mathbf{I}_{10})$, initialize the chains at $x_0 = (10, 10, \dots, 10)$, and run 10,000 iterations for each chain.
The result is summarized in Table~\ref{table:gaussian}, where for each setting we repeat the simulation 30 times. The advantage of the weight functions considered in Proposition~2 over $w_{\rm{ord}}$ is substantial.
We present the log-posterior traceplots in Figure~\ref{fig:conti_traj} and the MCMC sample trajectories in Figure~\ref{fig:conti_coord}.

\begin{table}[!h]
    \small
  \caption{Sampling from 10-d standard Gaussian distribution with 10,000 iterations using MTM. Averaged over 30 chains with random seeds.
  The number in the parenthesis is the standard error. }
  \label{table:gaussian}
  \centering
  \begin{tabular}{cc  c c  c  c  c}
    \toprule
     &$N$ & MH $(N=1)$ & 10 & 100 & 1,000  \\
    \midrule
    \multirow{4}{*}{Acc. Rate} &
    $w_{\mathrm{ord}}$ &   \multirow{4}{*}{0.8604 (0.001)}  & 0.7488 (0.002)   &  0.2971 (0.012)  &  0.0394 (0.008)  \\
      & $w_{\mathrm{sqrt}}$ &  & 0.9667 (0.000) &  0.9897 (0.000) &  0.9967 (0.000)  \\
      & $w_{\mathrm{min}}$ &  &  0.9656 (0.000) & 0.9890 (0.000) &  0.9963 (0.000) \\
      & $w_{\mathrm{max}}$ &  &  0.9593 (0.000)  &  0.9861  (0.000) &  0.9950 (0.000) \\
     \bottomrule
  \end{tabular}
\end{table}

\begin{figure}[!h]
    \centering
    \includegraphics[width=0.4\textwidth]{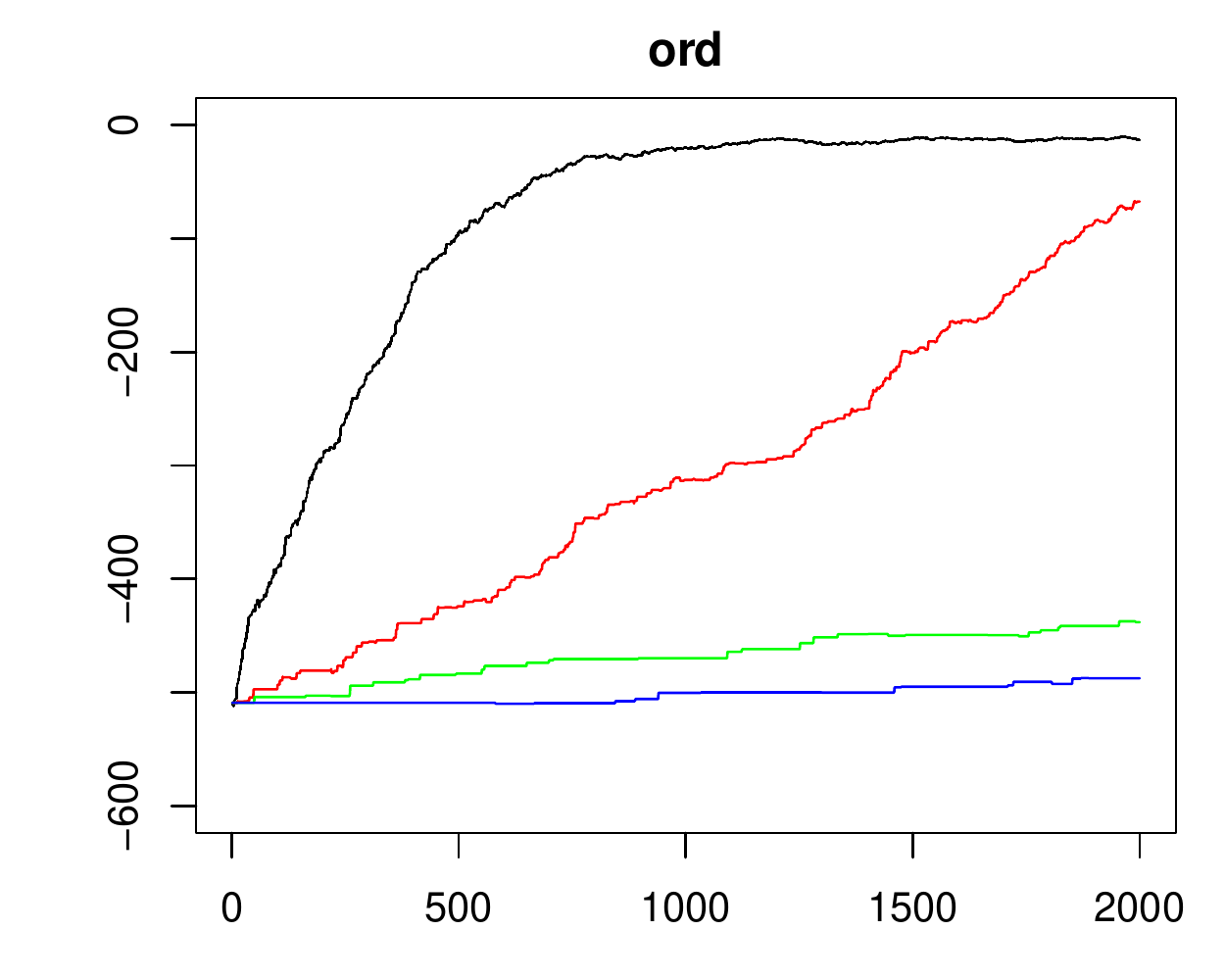}
    \includegraphics[width=0.4\textwidth]{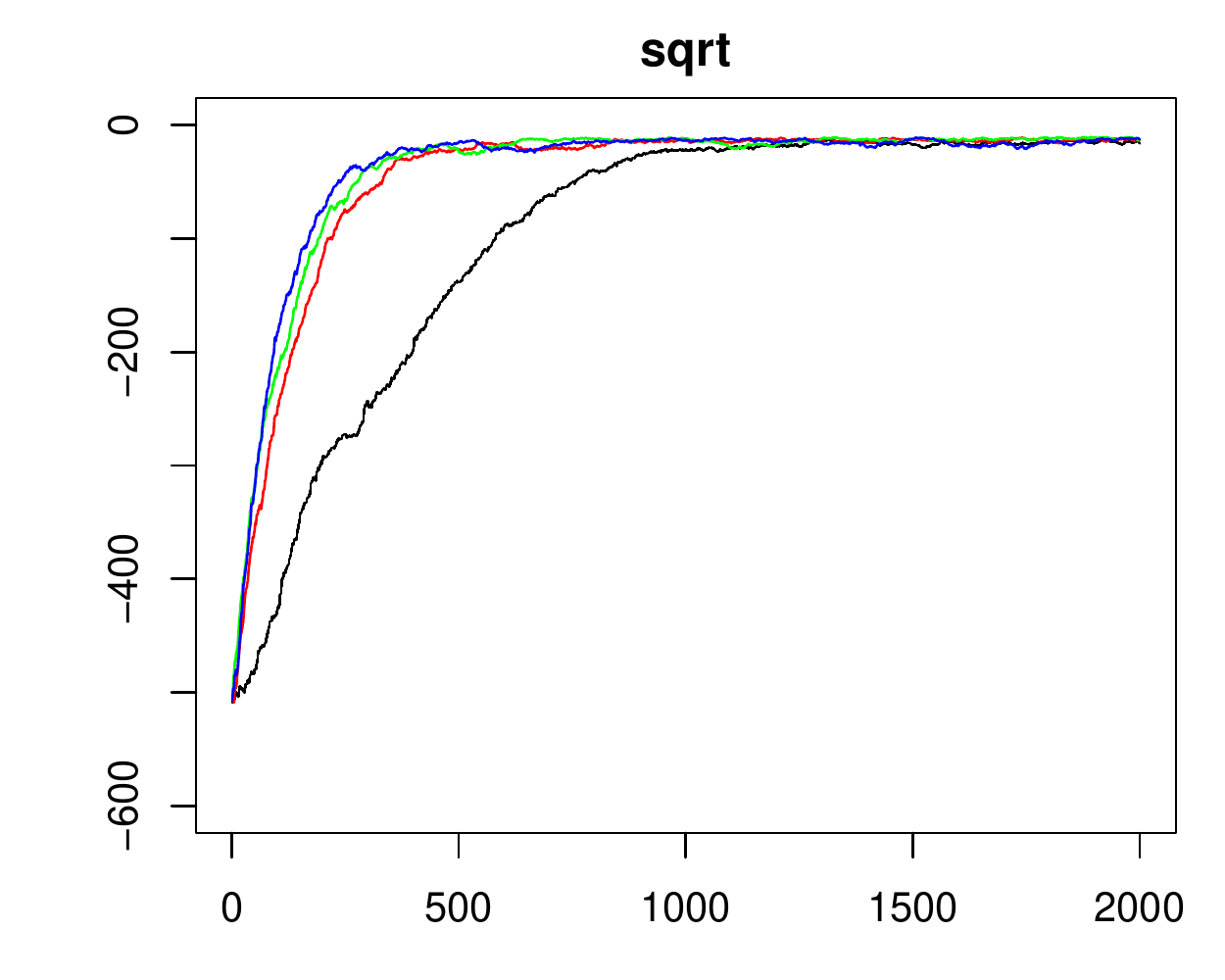}
    \includegraphics[width=0.4\textwidth]{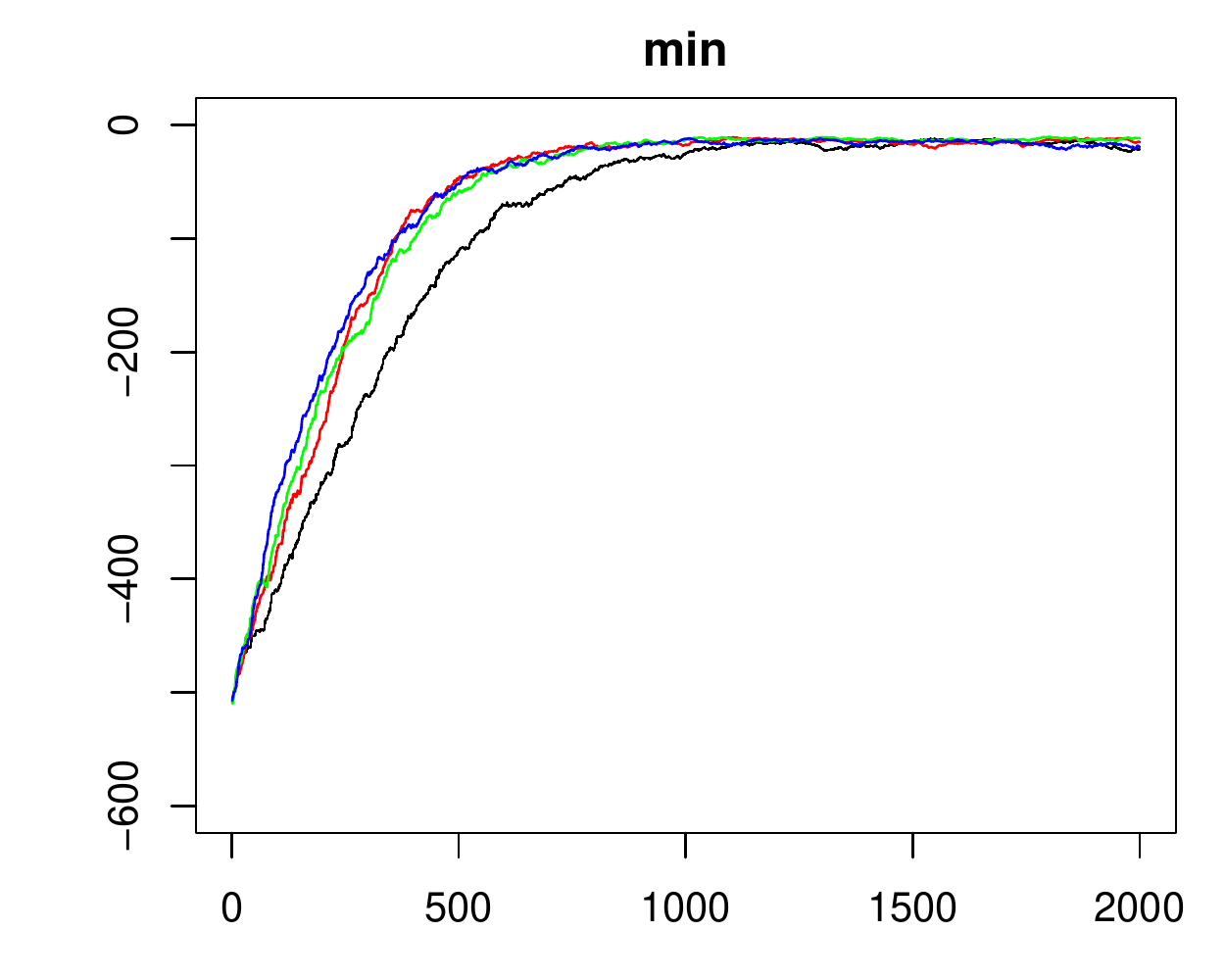}
    \includegraphics[width=0.4\textwidth]{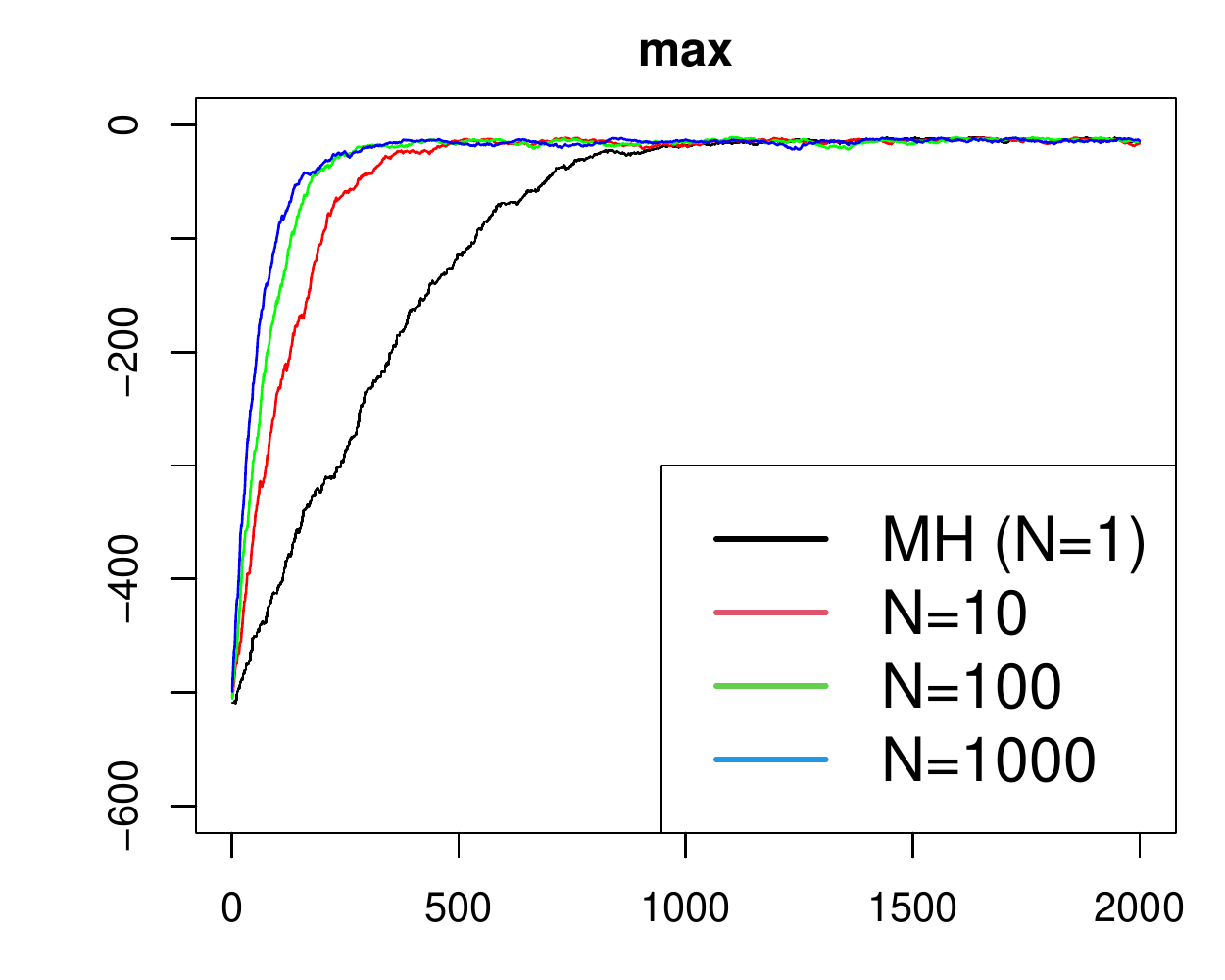}
    \caption{Log-posterior trace plots for 4 different weight functions $w_{\mathrm{ord}}, w_{\mathrm{sqrt}}, w_{\mathrm{min}}$ and $w_{\mathrm{max}}$. Different colors indicate a different number of trials as specified in the legend.}
    \label{fig:conti_traj}
\end{figure}

\begin{figure}[!h]
    \centering
    \includegraphics[width=0.5\textwidth]{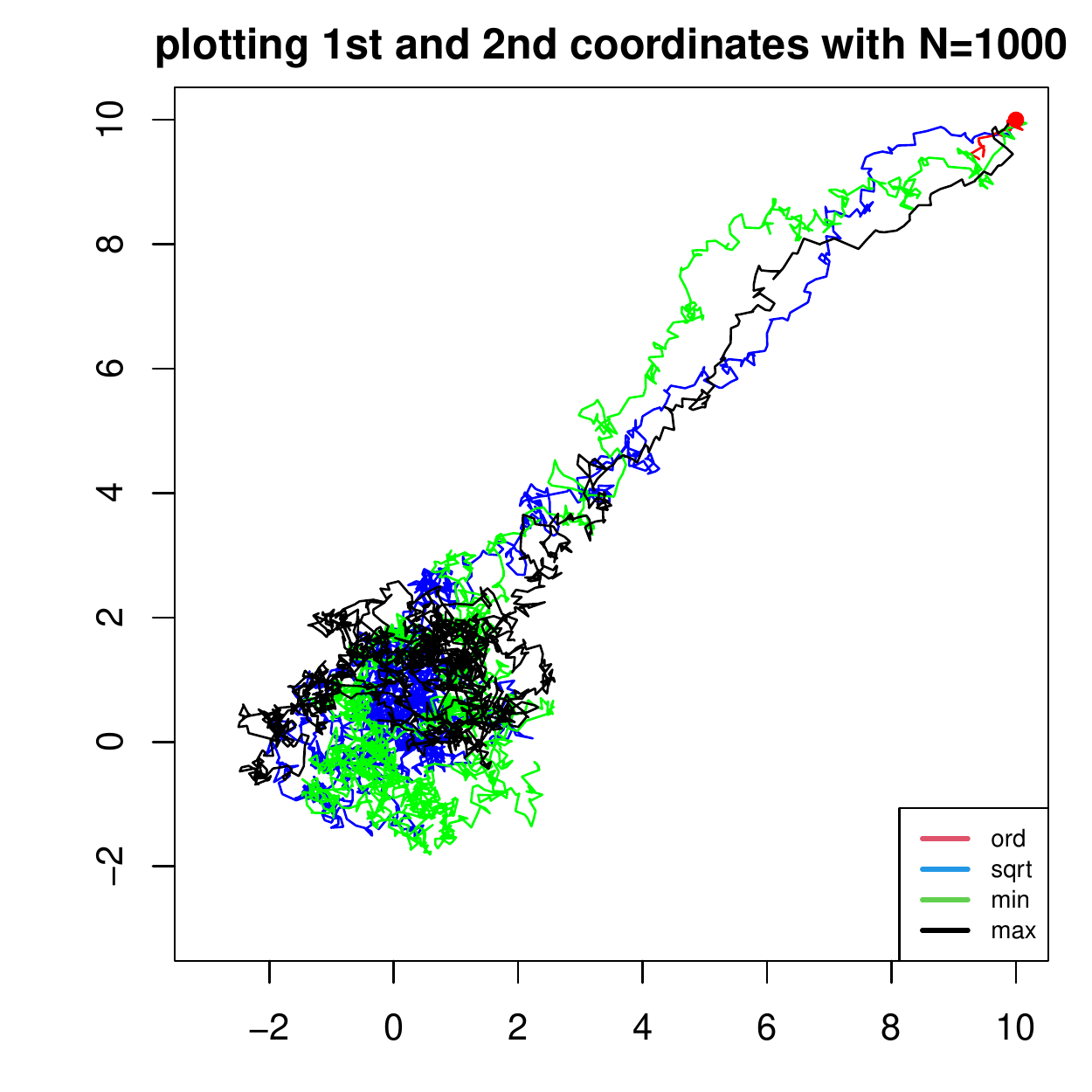}
    \caption{MCMC sample trajectories under 10-dimensional continuous target distribution $\mathsf{N}(\bm{0},\mathbf{I}_{10})$ with different weight functions, initialized at $x_0 = (10,10,\dots,10)$. Note that the chain with the ordinary weight function $w_{\mathrm{ord}}$ gets stuck at its early stage, whereas the chains with the other weight functions move to the region with a high posterior, which supports our claim.}
    \label{fig:conti_coord}
\end{figure}

\clearpage 

\section{Additional tables} \label{sec:addtable}

Here we report genetic variants with the top 10 highest posterior inclusion probabilities found from BVS model with MTM algorithm, under the different choices of $N=$ 50, 100, 500, 1000, 2000, 5000 and weight functions $w_{\mathrm{ord}}, w_{\mathrm{sqrt}}, w_{\mathrm{min}}$ and $w_{\mathrm{max}}$.

\begin{table}[h]
\centering
\caption{Genetic variants with top 10 posterior inclusion probability (PIP), obtained from MTM algorithm with $N=50$ and averaged over 5 chains. Blue are genetic variants reported by \cite{zhou2021dimension}.}
\label{table:gwas1}
\begin{tabular}{|r|ll|ll|ll|ll|}
  \hline
 & \multicolumn{2}{c|}{$w_{\mathrm{ord}}$} &  \multicolumn{2}{c|}{$w_{\mathrm{sqrt}}$} &  \multicolumn{2}{c|}{$w_{\mathrm{min}}$}&  \multicolumn{2}{c|}{$w_{\mathrm{max}}$}\\
 & Name & PIP &  Name & PIP  &  Name & PIP  & Name & PIP  \\
  \hline
1 & \tcb{rs1063192} & 0.987 & \tcb{rs10483727} & 0.997 & \tcb{rs1063192} & 0.998 & \tcb{rs1063192} & 0.982 \\ 
  2 & \tcb{rs653178} & 0.979 & \tcb{rs1063192} & 0.979 & \tcb{rs653178} & 0.976 & \tcb{rs10483727} & 0.976 \\ 
  3 & \tcb{rs10483727} & 0.972 & \tcb{rs653178} & 0.978 & \tcb{rs10483727} & 0.974 & \tcb{rs653178} & 0.972 \\ 
  4 & \tcb{rs2275241} & 0.908 & \tcb{rs2275241} & 0.914 & \tcb{rs2275241} & 0.915 & \tcb{rs2275241} & 0.903 \\ 
  5 & \tcb{rs319773} & 0.806 & \tcb{rs319773} & 0.798 & \tcb{rs4557053} & 0.802 & \tcb{rs319773} & 0.827 \\ 
  6 & \tcb{rs4557053} & 0.77 & \tcb{rs4557053} & 0.773 & \tcb{rs319773} & 0.777 & \tcb{rs4557053} & 0.752 \\ 
  7 & rs2369705 & 0.667 & rs2369705 & 0.678 & rs2369705 & 0.671 & rs2369705 & 0.667 \\ 
  8 & \tcb{rs10491971} & 0.619 & \tcb{rs10491971} & 0.639 & \tcb{rs10491971} & 0.63 & \tcb{rs10491971} & 0.633 \\ 
  9 & rs3177954 & 0.598 & rs3177954 & 0.604 & rs3177954 & 0.595 & rs3177954 & 0.602 \\ 
  10 & rs11087973 & 0.567 & rs11087973 & 0.581 & rs11087973 & 0.576 & rs3843894 & 0.533 \\ 
   \hline
\end{tabular}
\end{table}

\begin{table}[h]
\centering
\caption{Genetic variants with top 10 posterior inclusion probability (PIP), obtained from MTM algorithm with $N=100$ and averaged over 5 chains. Blue are genetic variants reported by \cite{zhou2021dimension}.}
\label{table:gwas2}
\begin{tabular}{|r|ll|ll|ll|ll|}
  \hline
 & \multicolumn{2}{c|}{$w_{\mathrm{ord}}$} &  \multicolumn{2}{c|}{$w_{\mathrm{sqrt}}$} &  \multicolumn{2}{c|}{$w_{\mathrm{min}}$}&  \multicolumn{2}{c|}{$w_{\mathrm{max}}$}\\
 & Name & PIP &  Name & PIP  &  Name & PIP  & Name & PIP  \\
  \hline
1 & \tcb{rs10483727} & 1 & \tcb{rs10483727} & 0.993 & \tcb{rs10483727} & 0.988 & \tcb{rs10483727} & 1 \\ 
  2 & \tcb{rs1063192} & 0.996 & \tcb{rs1063192} & 0.99 & \tcb{rs653178} & 0.97 & \tcb{rs1063192} & 0.982 \\ 
  3 & \tcb{rs653178} & 0.973 & \tcb{rs653178} & 0.973 & \tcb{rs1063192} & 0.953 & \tcb{rs653178} & 0.975 \\ 
  4 & \tcb{rs2275241} & 0.918 & \tcb{rs2275241} & 0.922 & \tcb{rs2275241} & 0.914 & \tcb{rs2275241} & 0.925 \\ 
  5 & \tcb{rs319773} & 0.795 & \tcb{rs4557053} & 0.806 & \tcb{rs319773} & 0.802 & \tcb{rs319773} & 0.793 \\ 
  6 & \tcb{rs4557053} & 0.795 & \tcb{rs319773} & 0.78 & \tcb{rs4557053} & 0.792 & \tcb{rs4557053} & 0.765 \\ 
  7 & rs2369705 & 0.673 & rs2369705 & 0.663 & rs2369705 & 0.662 & rs2369705 & 0.63 \\ 
  8 & \tcb{rs10491971} & 0.626 & \tcb{rs10491971} & 0.624 & \tcb{rs10491971} & 0.617 & rs11087973 & 0.604 \\ 
  9 & rs3177954 & 0.618 & rs3177954 & 0.62 & rs3177954 & 0.598 & \tcb{rs10491971} & 0.602 \\ 
  10 & rs11087973 & 0.583 & rs11087973 & 0.597 & rs11087973 & 0.587 & rs11040978 & 0.563 \\ 
   \hline
\end{tabular}
\end{table}

\begin{table}[h]
\centering
\caption{Genetic variants with top 10 posterior inclusion probability (PIP), obtained from MTM algorithm with $N=500$ and averaged over 5 chains. Blue are genetic variants reported by \cite{zhou2021dimension}.}
\label{table:gwas3}
\begin{tabular}{|r|ll|ll|ll|ll|}
  \hline
 & \multicolumn{2}{c|}{$w_{\mathrm{ord}}$} &  \multicolumn{2}{c|}{$w_{\mathrm{sqrt}}$} &  \multicolumn{2}{c|}{$w_{\mathrm{min}}$}&  \multicolumn{2}{c|}{$w_{\mathrm{max}}$}\\
 & Name & PIP &  Name & PIP  &  Name & PIP  & Name & PIP  \\
  \hline
1 & \tcb{rs1063192} & 1 & \tcb{rs1063192} & 1 & \tcb{rs1063192} & 1 & \tcb{rs10483727} & 1 \\ 
  2 & \tcb{rs653178} & 0.992 & \tcb{rs10483727} & 0.983 & \tcb{rs10483727} & 1 & \tcb{rs653178} & 0.925 \\ 
  3 & \tcb{rs2275241} & 0.943 & \tcb{rs653178} & 0.979 & \tcb{rs653178} & 0.973 & \tcb{rs2275241} & 0.915 \\ 
  4 & \tcb{rs319773} & 0.857 & \tcb{rs2275241} & 0.922 & \tcb{rs2275241} & 0.887 & \tcb{rs319773} & 0.839 \\ 
  5 & \tcb{rs10483727} & 0.8 & \tcb{rs319773} & 0.791 & \tcb{rs319773} & 0.792 & \tcb{rs4557053} & 0.828 \\ 
  6 & \tcb{rs4557053} & 0.787 & \tcb{rs4557053} & 0.774 & \tcb{rs4557053} & 0.791 & \tcb{rs1063192} & 0.799 \\ 
  7 & \tcb{rs10491971} & 0.658 & \tcb{rs10491971} & 0.67 & rs2369705 & 0.65 & rs2369705 & 0.764 \\ 
  8 & rs3843894 & 0.579 & rs2369705 & 0.647 & rs3177954 & 0.622 & \tcb{rs10491971} & 0.757 \\ 
  9 & \tcb{rs587409} & 0.553 & rs3177954 & 0.624 & \tcb{rs10491971} & 0.615 & rs12133371 & 0.617 \\ 
  10 & rs3177954 & 0.546 & rs11087973 & 0.613 & rs11087973 & 0.561 & rs11087973 & 0.605 \\
   \hline
\end{tabular}
\end{table}

\begin{table}[h]
\centering
\caption{Genetic variants with top 10 posterior inclusion probability (PIP), obtained from MTM algorithm with $N=1000$ and averaged over 5 chains. Blue are genetic variants reported by \cite{zhou2021dimension}.}
\label{table:gwas4}
\begin{tabular}{|r|ll|ll|ll|ll|}
  \hline
 & \multicolumn{2}{c|}{$w_{\mathrm{ord}}$} &  \multicolumn{2}{c|}{$w_{\mathrm{sqrt}}$} &  \multicolumn{2}{c|}{$w_{\mathrm{min}}$}&  \multicolumn{2}{c|}{$w_{\mathrm{max}}$}\\
 & Name & PIP &  Name & PIP  &  Name & PIP  & Name & PIP  \\
  \hline
1 & \tcb{rs653178} & 0.959 & \tcb{rs1063192} & 0.991 & \tcb{rs1063192} & 1 & \tcb{rs10483727} & 1 \\ 
  2 & \tcb{rs2275241} & 0.952 & \tcb{rs653178} & 0.982 & \tcb{rs10483727} & 1 & \tcb{rs653178} & 0.999 \\ 
  3 & \tcb{rs319773} & 0.918 & \tcb{rs10483727} & 0.939 & \tcb{rs653178} & 0.994 & \tcb{rs2275241} & 0.886 \\ 
  4 & \tcb{rs10483727} & 0.8 & \tcb{rs2275241} & 0.927 & \tcb{rs2275241} & 0.918 & rs2369705 & 0.842 \\ 
  5 & rs2369705 & 0.737 & \tcb{rs4557053} & 0.866 & \tcb{rs319773} & 0.791 & \tcb{rs319773} & 0.841 \\ 
  6 & \tcb{rs10491971} & 0.713 & \tcb{rs319773} & 0.859 & \tcb{rs4557053} & 0.779 & \tcb{rs1063192} & 0.804 \\ 
  7 & \tcb{rs4557053} & 0.704 & rs2369705 & 0.651 & \tcb{rs10491971} & 0.652 & rs3177954 & 0.787 \\ 
  8 & rs3177954 & 0.688 & \tcb{rs10491971} & 0.637 & rs2369705 & 0.646 & \tcb{rs10491971} & 0.768 \\ 
  9 & rs12133371 & 0.615 & rs3177954 & 0.6 & rs3177954 & 0.613 & rs12133371 & 0.656 \\ 
  10 & \tcb{rs1063192} & 0.6 & rs11087973 & 0.574 & rs11087973 & 0.588 & rs3843894 & 0.65 \\
   \hline
\end{tabular}
\end{table}

\begin{table}[h]
\centering
\caption{Genetic variants with top 10 posterior inclusion probability (PIP), obtained from MTM algorithm with $N=2000$ and averaged over 5 chains. Blue are genetic variants reported by \cite{zhou2021dimension}.}
\label{table:gwas5}
\begin{tabular}{|r|ll|ll|ll|ll|}
  \hline
 & \multicolumn{2}{c|}{$w_{\mathrm{ord}}$} &  \multicolumn{2}{c|}{$w_{\mathrm{sqrt}}$} &  \multicolumn{2}{c|}{$w_{\mathrm{min}}$}&  \multicolumn{2}{c|}{$w_{\mathrm{max}}$}\\
 & Name & PIP &  Name & PIP  &  Name & PIP  & Name & PIP  \\
  \hline
1 & \tcb{rs2275241} & 1 & \tcb{rs653178} & 0.983 & \tcb{rs1063192} & 0.998 & \tcb{rs1063192} & 1 \\ 
  2 & \tcb{rs653178} & 1 & \tcb{rs10483727} & 0.922 & \tcb{rs653178} & 0.971 & \tcb{rs10483727} & 0.997 \\ 
  3 & \tcb{rs319773} & 0.819 & \tcb{rs2275241} & 0.857 & \tcb{rs2275241} & 0.916 & \tcb{rs653178} & 0.986 \\ 
  4 & rs11087973 & 0.805 & \tcb{rs1063192} & 0.829 & \tcb{rs319773} & 0.794 & \tcb{rs2275241} & 0.948 \\ 
  5 & \tcb{rs1063192} & 0.8 & \tcb{rs4557053} & 0.824 & \tcb{rs4557053} & 0.771 & \tcb{rs319773} & 0.87 \\ 
  6 & \tcb{rs4557053} & 0.755 & \tcb{rs10491971} & 0.775 & \tcb{rs10483727} & 0.724 & rs2369705 & 0.675 \\ 
  7 & rs1460509 & 0.737 & \tcb{rs319773} & 0.763 & \tcb{rs10491971} & 0.714 & \tcb{rs4557053} & 0.634 \\ 
  8 & rs3843894 & 0.729 & rs2369705 & 0.697 & rs3177954 & 0.672 & rs3177954 & 0.527 \\ 
  9 & rs3858886 & 0.648 & rs11040978 & 0.651 & rs2369705 & 0.628 & rs2567344 & 0.519 \\ 
  10 & \tcb{rs10483727} & 0.61 & \tcb{rs587409} & 0.612 & rs3843894 & 0.576 & rs11634375 & 0.472 \\ 
   \hline
\end{tabular}
\end{table}
\begin{table}[h]
\centering
\caption{Genetic variants with top 10 posterior inclusion probability (PIP), obtained from MTM algorithm with $N=5000$ and averaged over 5 chains. Blue are genetic variants reported by \cite{zhou2021dimension}.}
\label{table:gwas6}
\begin{tabular}{|r|ll|ll|ll|ll|}
  \hline
 & \multicolumn{2}{c|}{$w_{\mathrm{ord}}$} &  \multicolumn{2}{c|}{$w_{\mathrm{sqrt}}$} &  \multicolumn{2}{c|}{$w_{\mathrm{min}}$}&  \multicolumn{2}{c|}{$w_{\mathrm{max}}$}\\
 & Name & PIP &  Name & PIP  &  Name & PIP  & Name & PIP  \\
  \hline
1 & rs2151280 & 0.4 & \tcb{rs1063192} & 1 & \tcb{rs1063192} & 1 & \tcb{rs10483727} & 1 \\ 
  2 & \tcb{rs10483727} & 0.367 & \tcb{rs10483727} & 1 & \tcb{rs10483727} & 1 & \tcb{rs653178} & 0.997 \\ 
  3 & rs12457539 & 0.29 & \tcb{rs653178} & 0.983 & \tcb{rs653178} & 0.972 & \tcb{rs319773} & 0.901 \\ 
  4 & rs10508818 & 0.255 & \tcb{rs319773} & 0.804 & \tcb{rs2275241} & 0.821 & rs3177954 & 0.844 \\ 
  5 & rs3858886 & 0.231 & rs2369705 & 0.791 & \tcb{rs4557053} & 0.732 & \tcb{rs2275241} & 0.668 \\ 
  6 & rs7995962 & 0.207 & \tcb{rs2275241} & 0.772 & \tcb{rs319773} & 0.709 & rs12457539 & 0.649 \\ 
  7 & rs12125527 & 0.2 & rs3177954 & 0.732 & rs3177954 & 0.662 & \tcb{rs1063192} & 0.6 \\ 
  8 & rs2738755 & 0.2 & \tcb{rs587409} & 0.652 & rs12133371 & 0.593 & \tcb{rs587409} & 0.6 \\ 
  9 & rs6661853 & 0.2 & \tcb{rs10491971} & 0.636 & rs11087973 & 0.569 & rs4924156 & 0.51 \\ 
  10 & rs9869577 & 0.2 & \tcb{rs4557053} & 0.611 & \tcb{rs587409} & 0.565 & rs4236601 & 0.498 \\ 
   \hline
\end{tabular}
\end{table}

\end{document}